\documentclass[11pt,letterpaper]{article}
\usepackage{fullpage}
\usepackage{graphicx}
\usepackage{url}
\usepackage{amsmath}
\usepackage{accents}
\usepackage{amssymb}
\usepackage{amsthm}
\usepackage{mathrsfs}
\usepackage{mathtools}
\usepackage{thmtools}
\usepackage{thm-restate}
\usepackage{hyperref}
\usepackage[utf8]{inputenc} 
\usepackage[capitalize]{cleveref}
\usepackage{undertilde}
\usepackage{comment}
\usepackage[most]{tcolorbox}
\usepackage{microtype}

\usepackage{textpos}

\makeatletter
\newcommand{\subalign}[1]{%
  \vcenter{%
    \Let@ \restore@math@cr \default@tag
    \baselineskip\fontdimen10 \scriptfont\tw@
    \advance\baselineskip\fontdimen12 \scriptfont\tw@
    \lineskip\thr@@\fontdimen8 \scriptfont\thr@@
    \lineskiplimit\lineskip
    \ialign{\hfil$\m@th\scriptstyle##$&$\m@th\scriptstyle{}##$\hfil\crcr
      #1\crcr
    }%
  }%
}
\makeatother

\newcommand{\bbF}{\mathbb{F}}

\DeclareMathOperator{\RR}{R}
\DeclareMathOperator{\BR}{\underline{R}}
\DeclareMathOperator{\AR}{\utilde{\mathrm R}}
\DeclareMathOperator{\CC}{C}
\DeclareMathOperator{\AC}{\utilde{\mathrm C}}
\DeclareMathOperator{\poly}{poly}
\DeclareMathSymbol{\upset}{\mathopen}{symbols}{"22}
\DeclareMathSymbol{\downset}{\mathopen}{symbols}{"23}

\renewcommand{\emptyset}{\varnothing}
\newcommand{\cw}{\mathrm{cw}}
\newcommand{\CW}{\mathrm{CW}}
\newcommand{\ov}{\overline}
\DeclareMathOperator{\eH}{H}

\newcommand{\cat}[3][]{S^{#2}_{#3}}

\newcommand{\pro}{\leq_{\mathrm{s}}}
\newcommand{\res}{\leq}
\newcommand{\equ}{\equiv}

\newcommand{\isom}{\equiv}
\newcommand{\tw}[1]{\mathrm{tw}(#1)}
\newcommand{\ltw}[1]{\mathrm{ltw}(#1)}

\usepackage[textwidth=2.5cm]{todonotes}

\usepackage{xspace}
\newcommand{\VP}{\ensuremath{\mathsf{VP}}\xspace}
\newcommand{\VNP}{\ensuremath{\mathsf{VNP}}\xspace}

\title{Beyond Bilinear Complexity:\\ What Works and What Breaks with Many Modes?}
\date{}
\author{Cornelius Brand\thanks{University of Regensburg. Funded by ERC project CountHom.}
    \and Radu Curticapean\thanks{University of Regensburg and IT University of Copenhagen. Funded by the
European Union (ERC, CountHom, 101077083). Views and opinions expressed are however those of the author(s) only and do not necessarily reflect those of the European Union or the European Research Council Executive Agency.}
    \and Petteri Kaski\thanks{Aalto University.}
    \and Baitian Li\thanks{Columbia University. Funded by a Columbia SEAS Presidential Fellowship.}
    \and Ian Orzel\thanks{University of Copenhagen. Funded by the European Research Council (ERC) \\under grant agreement no. 101125652 (ALBA).}
    \and Tim Seppelt\thanks{IT University of Copenhagen. Funded by ERC project CountHom.}
    \and Jiaheng Wang\thanks{University of Regensburg and University of Helsinki. Funded by ERC project CountHom.}}

\newtheorem{lemma}{Lemma}

\newtheorem{theorem}[lemma]{Theorem}
\newtheorem{corollary}[lemma]{Corollary}
\newtheorem{proposition}[lemma]{Proposition}
\newtheorem{fact}[lemma]{Fact}

\theoremstyle{definition}
\newtheorem{remark}[lemma]{Remark}
\newtheorem*{remark*}{Remark}
\newtheorem{definition}[lemma]{Definition}

\newtheorem{example}[lemma]{Example}

\theoremstyle{remark}

\Crefname{claim}{Claim}{Claims}

\usepackage{enumitem}
\setlist[enumerate, 1]{font=\upshape, noitemsep, nolistsep}
\setlist[enumerate, 2]{font=\upshape, noitemsep, nolistsep}
\setlist[itemize, 1]{noitemsep, nolistsep,font=\upshape}
\setlist[itemize, 2]{noitemsep, nolistsep,font=\upshape}

\usetikzlibrary{shapes}
\definecolor{cbfp1}{RGB}{120,94,240}
\definecolor{cbfp2}{RGB}{220,38,127}
\definecolor{cbfp3}{RGB}{254,97,0}
\definecolor{cbfp4}{RGB}{255,176,0}
\definecolor{verylightgray}{RGB}{216,216,216}

\begin{document}

\maketitle
\begin{textblock}{5}(8.5, 8) \includegraphics[width=120px]{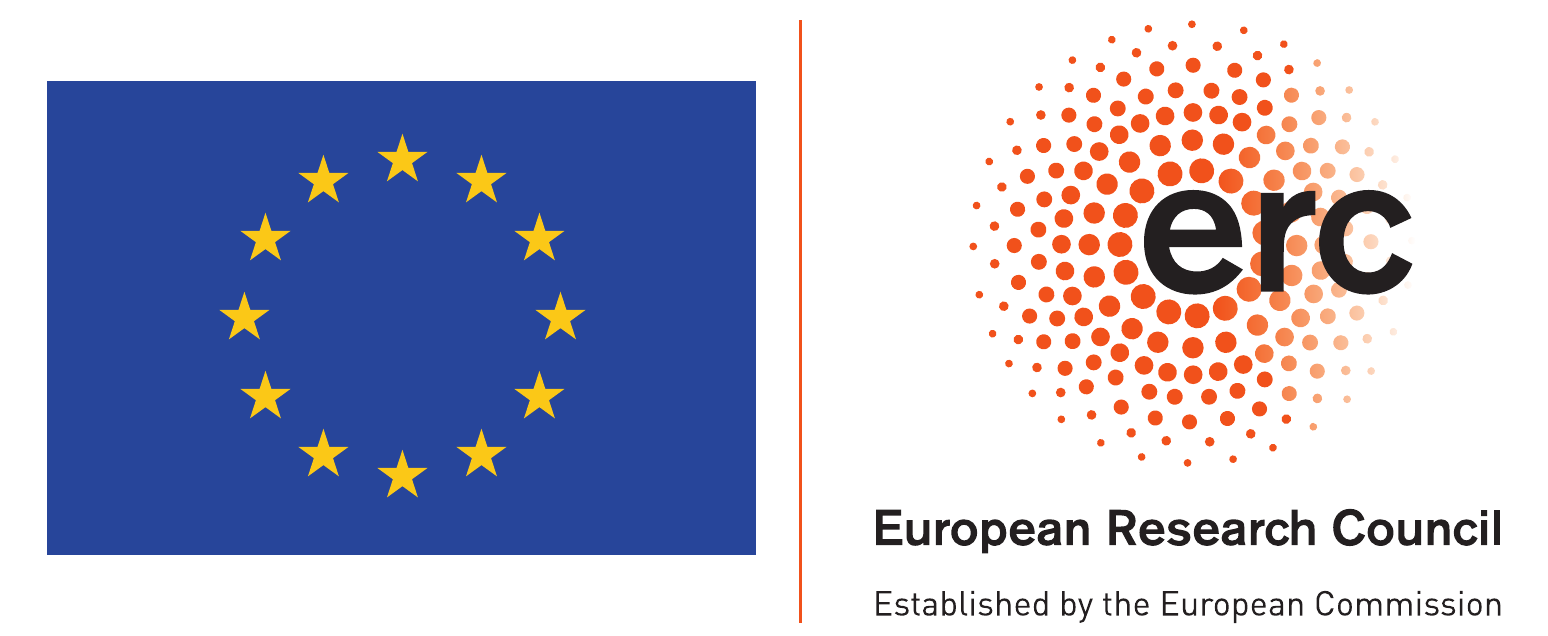} \end{textblock}

\begin{abstract}
    The complexity of bilinear maps (equivalently, of $3$-mode tensors) has been studied extensively, most notably in the context of matrix multiplication. While circuit complexity and tensor rank coincide asymptotically for $3$-mode tensors, this correspondence breaks down for $d \geq 4$ modes. As a result, the complexity of $d$-mode tensors for larger fixed $d$ remains poorly understood, despite its relevance, e.g., in fine-grained complexity. Our paper explores this intermediate regime. 
    
    First, we give a ``graph-theoretic'' proof of Strassen's $2\omega/3$ bound on the asymptotic rank exponent of $3$-mode tensors. Our proof directly generalizes to an upper bound of $(d-1)\omega/3$ for $d$-mode tensors. Using refined techniques available only for $d\geq 4$ modes, we improve this bound beyond the current state of the art for $\omega$.  
    We also obtain a bound of $d/2+1$ on the asymptotic exponent of \emph{circuit complexity} of generic $d$-mode tensors and optimized bounds for $d \in \{4,5\}$.
    
    To the best of our knowledge, asymptotic circuit complexity (rather than rank) of tensors has not been studied before.
    To obtain a robust theory, we first ask whether low complexity of $T$ and $U$ imply low complexity of their Kronecker product $T \otimes U$.
    While this crucially holds for rank (and thus for circuit complexity in $3$ modes), we show that assumptions from fine-grained complexity rule out such a \emph{submultiplicativity} for the circuit complexity of tensors with many modes. In particular, assuming the Hyperclique Conjecture, this failure occurs already for $d=8$ modes.
    Nevertheless, we can salvage a restricted notion of submultiplicativity.

    From a technical perspective, our proofs heavily make use of the \emph{graph tensors} $T_H$, as employed by Christandl and Zuiddam ({\em Comput.~Complexity}~28~(2019)~27--56) and Christandl, Vrana and Zuiddam ({\em Comput.~Complexity}~28~(2019)~57--111), whose modes correspond to the vertices of undirected graphs $H$. We make the simple but conceptually crucial observation that Kronecker products $T_G \otimes T_H$ are isomorphic to $T_{G+H}$, and that $G$ and $H$ may also be \emph{fractional} graphs. By asymptotically converting generic tensors to specific graph tensors, we can use nontrivial results from algorithmic graph theory to study the rank and complexity of $d$-mode tensors for fixed $d$.
    \vspace{1cm}
\end{abstract}

\tableofcontents

\section{Introduction}

In several fundamental computational problems, the input data is represented by two vectors $u,v$ and the output is a vector $f(u,v)$ that captures a meaningful combination of the input data.
Prominent examples are \emph{matrix multiplication} (when the input matrices are flattened to vectors), \emph{polynomial multiplication} (when the input polynomials are given as coefficient vectors), and more general \emph{convolution} problems such as the \emph{subset convolution} in parameterized and exact algorithms, which is essentially multiplication in the algebra of square-free multivariate polynomials. 

\paragraph{Bilinear Maps and Rank.}
All of the above problems ask us to evaluate maps $f : \mathbb F^n  \times \mathbb F^n \to \mathbb F^n$ that are \emph{bilinear} over the field $\mathbb F$; they satisfy $f(\alpha x+\beta x',y) = \alpha f(x,y) + \beta f(x',y)$, with an analogous rule for the second argument.
A bilinear map can also be viewed as a polynomial map that outputs the evaluations of $n$ set-multilinear polynomials in variable sets $x_1,\ldots,x_n$ and $y_1,\ldots,y_n$.
(In this setting, a polynomial is set-multilinear if each monomial is multilinear and contains at most one variable among $x_1,\ldots,x_n$ and at most one among $y_1,\ldots,y_n$.)

As laid out in the remainder of the introduction, bilinear maps enjoy a strong connection between computational complexity and \emph{rank}, an algebraic complexity measure. 
The rank $\RR(f)$ is the minimum number of terms required to express $f$ as a linear combination of rank-$1$ maps, where a rank-$1$ map is a product of linear functions in disjoint variable sets.

\paragraph{Rank and Kronecker Products.}

Several extremely well-studied bilinear maps, such as matrix multiplication and subset convolution, have a very useful structure: They can be expressed as \emph{Kronecker products} of lower-dimensional bilinear maps.
To make this precise, given vectors $u,u' \in \mathbb F^{n}$, define $u \otimes u' \in \mathbb F^{n^2}$ by setting $(u \otimes u')_{i,j}=u_iu'_j$ and viewing the index set $[n]^2$ as $[n^2]$ under some bijection.
For bilinear maps $f,f' : \mathbb F^{n} \times \mathbb F^{n} \to \mathbb F^{n}$, the \emph{Kronecker product} $f \otimes f' : \mathbb F^{n^2} \times \mathbb F^{n^2} \to \mathbb F^{n^2}$ is defined by mapping 
\[(u\otimes u' , v \otimes v')\mapsto f(u,v)\otimes f'(u',v'),\]
and taking the unique bilinear extension of this map.
For example, the Kronecker product of the bilinear map $g_n$ for multiplication of two $n\times n$ matrices and the map $g_m$ for multiplication of two $m \times m$ matrices is the map $g_n \otimes g_m \equ g_{nm}$ for multiplication of two $nm \times nm$ matrices. 
Likewise, subset convolution over a universe of size $n$ (i.e., square-free multiplication of $n$-variate multilinear polynomials) is the Kronecker power of subset convolution over a universe of size $1$.

Rank and Kronecker products interact very favorably: If $f_1$ and $f_2$ have rank $r_1 =\RR(f_1)$ and $r_2 = \RR(f_2)$, then $\RR(f_1 \otimes f_2) \leq r_1r_2$, since the rank decompositions for $f_1$ and $f_2$ can be combined to a rank decomposition for $f_1 \otimes f_2$.
Unlike for linear maps, this combined rank decomposition of bilinear maps may however not be optimal, and we may have $\RR(f_1 \otimes f_2) < r_1r_2$. 
Consequently, we say that rank is \emph{submultiplicative} under Kronecker products.

\paragraph{Asymptotic Rank and Circuit Complexity.}
The rank of bilinear maps like matrix multiplication and subset convolution of increasing dimensions can be understood by fixing a constant-dimension bilinear map $g : \mathbb F^n \times \mathbb F^n \to \mathbb F^n$ of rank $r=\RR(g)$ and considering its \emph{Kronecker power} $g^{\otimes k} = g\otimes \ldots \otimes g$, a bilinear map between spaces of dimension $n^k$. As discussed above, we then have $\RR(g^{\otimes k}) \leq r^k$, but this bound may be exponentially loose:
One of the simplest examples of this phenomenon is the subset convolution map $g : \mathbb F^2 \times \mathbb F^2 \to \mathbb F^2$ on $1$-element universes, or what is essentially the same, the three-qubit W-state in quantum information theory.
We have $\RR(g)=3$ and thus $\RR( g^{\otimes k}) \leq 3^k$, while it is known that $\RR(g^{\otimes k}) \leq O(2^k k)$.
Combined with algorithmic ideas dating back to Yates~\cite{yates}, this nontrivial rank bound for $g^{\otimes k}$ translates into nontrivial $2^{k+o(k)}$ size circuits for subset convolution on $k$-element universes.
Similar observations are crucial in the study of matrix multiplication.

This algorithmic connection between rank and algorithmic complexity of Kronecker powers motivates defining the \emph{asymptotic rank exponent} $\omega(g)$ of a bilinear map $g : \mathbb F^n \times \mathbb F^n \to \mathbb F^n$ as the infimum over all $\beta >0$ such that $\RR(g^{\otimes k}) \leq O(n^{\beta k})$.
Likewise, we define the \emph{asymptotic circuit complexity exponent} $\eta(g)$ as the infimum over all $\beta>0$ such that $g^{\otimes k}$ has arithmetic circuits of size $O(n^{\beta k})$.
For bilinear maps, these two exponents are known to be \emph{equal} (cf.~the discussion following \cref{thm:yates}), 
so asymptotic circuit complexity can be understood entirely in terms of asymptotic rank, which in turn is amenable to techniques from pure mathematics.  

Trivially, the asymptotic rank exponent of every bilinear map is bounded by $2$, because $\RR(g)\leq n^2$.
Surprisingly however, Strassen~\cite{Strassen1988} showed a bound of $2\omega/3 < 2$, where $\omega < 3$ is the asymptotic rank exponent of matrix multiplication. 
This translates directly to improved algorithms, e.g., for convolution problems~\cite{BrandCLP26}.
The \emph{asymptotic rank conjecture}~\cite{Strassen1994} postulates that Strassen's upper bound can be pushed further down to $1$, i.e., that every bilinear map $g$ among $n$-dimensional vector spaces satisfies $\RR(g^{\otimes k}) \leq O(n^{k})$. 
This would lead to near-linear time algorithms for evaluating any map of the form $g^{\otimes k}$, which in turn would directly imply that $n\times n$ matrices can be multiplied with $n^{2+o(1)}$ operations.
Additionally, it was recently shown that the asymptotic rank conjecture would imply breakthrough exponential-time algorithms, e.g., for the permanent~\cite{BKKN25}, for the chromatic number of graphs~\cite{DBLP:conf/soda/BjorklundCHKP25}, and for the set cover problem~\cite{DBLP:conf/stoc/BjorklundK24,pratt_stronger_2024}.

\paragraph{More Modes.}
So far, we described \emph{bi}-linear maps, which take \emph{two} input vectors and produce an output vector.
More generally, a \emph{multi}-linear map takes some number $d$ of input vectors, called \emph{modes}, and produces an output vector, and is linear in each argument. This is a natural generalization of the bilinear case, and is one of the open research directions mentioned, e.g., in the survey by Wigderson and Zuiddam~\cite[Sect.~13]{wigderson2022asymptotic}. 
Beyond their intrinsic algebraic appeal, such maps also capture multi-party states in quantum information theory and important problems in fine-grained complexity.

A prominent example for a multilinear map with $d$ modes is the \emph{iterated matrix multiplication}, which takes as input $d$ matrices $A_1,\ldots,A_d$ (flattened to vectors) and outputs their product $A_1\ldots A_d$. 
Other examples include the \emph{determinant} and \emph{permanent} of an $n\times n$ matrix $A$, which are multilinear \emph{forms} (rather than \emph{maps}, i.e., they output scalars) with $n$ modes corresponding to the columns of $A$.
In fine-grained complexity, the $k$-hyperclique problem admits a natural formulation as a multilinear form with $k$ modes~\cite{LincolnWW18}. The complexity of this problem is interesting even for small, fixed values of~$k$, starting with $k=4$, and has attracted significant attention in recent years, especially as a source of conditional hardness~\cite{Kunnemann22,BringmannS21,GorbachevK23}.

\begin{remark*}
We will only speak of multilinear \emph{forms} in the following, because maps with $d-1$ modes canonically correspond to forms with $d$ modes by taking duals in the involved (finite-dimensional) spaces. 
After fixing coordinates, multilinear forms $g$ with $d$ modes can be viewed as set-multilinear polynomials: Denoting the entries of the input vector $\mathbf{x^{(i)}}$ at mode $i\in [d]$ by $x_1^{(i)},\ldots,x_n^{(i)}$, we have
\[
g(\mathbf{x^{(1)}},\ldots,\mathbf{x^{(d)}}) = \sum_{i_1,\ldots,i_d} t_{i_1,\ldots,i_d} \,x^{(1)}_{i_1} \ldots x^{(d)}_{i_d}
\]
with coefficients $t_{i_1,\ldots,i_d} \in \mathbb F$ that fully specify $g$.
Arranging these coefficients in an array with $d$ modes, we call this array the \emph{tensor} of $g$, and overloading terminology further, we call $g$ \emph{itself} a tensor with $d$ modes.
After tacitly performing all required identifications, we can view bilinear maps as $3$-mode tensors.    
\end{remark*}

\paragraph{Rank versus Circuit Complexity for More Modes.}
Rank and circuit complexity can be generalized directly to tensors with an arbitrary number of modes.
For $d \geq 4$ modes however, the direct correspondence between asymptotic circuit complexity and asymptotic rank breaks down.
As a simple, folklore example, consider the tensor $P_n$ with four $n$-dimensional modes defined by
\begin{equation}
\label{eq:intro-pip}
P_n(\mathbf x,\mathbf y,\mathbf z,\mathbf w)= 
\left( \sum_{j=1}^nx_jy_j   \right)
\left( \sum_{j=1}^nw_jz_j   \right).
\end{equation}
In algebraic complexity, $P_n$ is known as a product of inner products (PIP).
We observe that $P_n\otimes P_m \equ P_{nm}$, so the asymptotic rank exponent of every single $P_n$ equals the exponent of the monotonically increasing sequence $\RR(P_n)$ as $n\to \infty$.
The defining formula for $P_n$ of size $O(n)$ implies an asymptotic circuit complexity exponent of $1$, while the asymptotic rank exponent of $P_n$ is $2$: Multiplying out \eqref{eq:intro-pip} gives an upper bound of $n^2$, matched by a simple flattening lower bound.

Thus, unlike for $d=3$ modes, asymptotic rank and asymptotic circuit complexity of $4$-mode tensors for $d\geq 4$ need not agree.
The separation can be exacerbated by generalizing the PIP tensor to a fixed number of $2d$ modes, which increases the asymptotic rank exponent to $d$, while the asymptotic circuit complexity exponent stays $1$.
As the example shows, the separation holds even for $\Pi\Sigma\Pi$-formula complexity, which morally still resembles rank in that rank decompositions can be viewed as $\Sigma\Pi\Sigma$-formul\ae.

\subsection{Our Results: Phenomena of Tensors with More Modes}
Starting from the observation that rank and circuit complexity need not agree for relevant tensors with $d\geq 4$ modes, we are led to systematically study asymptotic rank and asymptotic circuit complexity for $d$-mode tensors when $d\geq 3$ is fixed and moderately large or may tend to infinity.
As part of this study, we obtain better upper bounds on asymptotic rank and asymptotic circuit complexity for generic $d$-mode tensors. Moreover, we investigate to which extent desirable properties like submultiplicativity fail for the circuit complexity of $d$-mode tensors with $d\geq 4$.
The foundation to our results are \emph{graph tensors} in the form used by Christandl and Zuiddam~\cite{ChristandlZ19} and Christandl, Vrana and Zuiddam~\cite{ChristandlVZ19}, which allow us to import techniques from graph theory.

We elaborate on these tensors in \Cref{sec:intro-graph-tensors}; for now, it suffices to note that the modes of the graph tensor $T_H$ correspond to the vertices of a graph $H$. 
We also shortly remark that graph tensors can be viewed as a special class of tensors that admit representation as a {\em tensor network}. Tensor networks have received extensive attention across a range of disciplines, each pursuing its own distinctive terminology and notation, from pure mathematics and physics \cite{L2021,O2014,PR1987} to the study of probabilistic graphical models in artificial intelligence and machine learning \cite{KF2009,KFL2001,Pearl1986}, and to the study of restricted circuit models in algebraic computation, where tensor networks give rise to a class of set-multilinear circuits for evaluating multilinear maps~\cite{austrin_tensor_2022}.

\paragraph{Generic Asymptotic Rank.}
To our surprise and best of knowledge, an analogue of Strassen's $2\omega/3$ bound on the asymptotic exponent of $3$-mode tensors for $d$-tensors with $d\geq 4$ is not mentioned in the literature.
We give a simple proof of Strassen's bound in the framework of graph tensors that yields an upper bound of $(d-1)\omega /3$ on the asymptotic exponent of all $d$-mode tensors for $d \geq 3$.

Our proof admits some freedom in the choice of low-rank tensor $L$ to reduce to: While Strassen's original proof requires $L$ to be the matrix multiplication tensor, which will turn out to be the graph tensor of a triangle~$K_3$, our proof more generally allows $L$ to be the graph tensor of the complete graph $K_s$ for $s\leq d$. In the case $d=3$, we recover the aforementioned bound of $(d-1)\omega /3$, but for $d\geq 4$, better upper bounds on the exponent of the graph tensor of $K_d$ are known and directly yield better upper bounds.
Moreover, with a nontrivial application of Strassen's \emph{laser method}, we improve the best known exponent on the graph tensor for $K_4$ and obtain:
\begin{restatable}[Upper bound on asymptotic rank]{theorem}{strassengeneral}
\label{thm:asymptotic-submultiplicativity}
    For every $d$-mode tensor with $d\geq 4$, the asymptotic rank exponent of $T$ is at most $0.772318(d-1)$.
\end{restatable}
For comparison, the PIP tensor presented above \eqref{eq:intro-pip} establishes a lower bound of $\lfloor d/2\rfloor$ on the rank exponent of $d$-mode tensors. Moreover, $\omega=2$ would imply an upper bound of $2/3\cdot (d-1)$.
This is elaborated more formally in Remark~\ref{rem:strassen}.

\paragraph{Submultiplicativity.}
While the PIP tensor rules out a functional dependence between rank and circuit complexity for $d\geq 4$ modes, we ask whether circuit complexity is at least \emph{submultiplicative}:
From \emph{rank} decompositions for tensors $T$ and $U$, we can explicitly construct a \emph{rank} decomposition for $T\otimes U$ of rank $\RR(T)\cdot\RR(U)$.
If $T$ and $U$ are given by circuits of small size, can we analogously expect a circuit of small size for $T\otimes U$? Restricted circuit models are known to have submultiplicativity under Kronecker powering; in particular, a set-multilinear circuit arising from a tensor network and its contraction tree for a tensor $T$ gives rise to a tensor network and its contraction tree for the Kronecker power $T^{\otimes k}$, with submultiplicativity of amortized cost~\cite{austrin_tensor_2022}. But does {\em unrestricted} circuit complexity have submultiplicativity? 

Our results strongly suggest a negative answer, but we can obtain such statements only \emph{conditionally}, since an unconditional result would imply strong circuit lower bounds.
In the following, we write $\CC(T)$ for the minimum size of an arithmetic circuit that computes the tensor $T$.
\begin{restatable}[Submultiplicativity of $\CC$ implies $\VP = \VNP$]{theorem}{submultvpvnp}
\label{thm:submultiplicativity-vp-vnp}
    There are explicit tensors $T_1,T_2,\ldots$ and $U_1,U_2,\ldots$ with $T_d,U_d \in  (\mathbb C^{2})^{\otimes d}$ for all $d\in \mathbb N$ such that the following holds:
    If $\CC(T_d\otimes U_d) \leq \poly(d,\CC(T_d),\CC(U_d))$ for all $d\in \mathbb N$, then $\VP = \VNP$.
\end{restatable}
Thus, the assumption $\VP \neq \VNP$ rules out submultiplicativity on $d$-mode tensors for all $d \in \mathbb N$.
Note that $\VP \neq \VNP$ implies super-polynomial circuit lower bounds for permanents \cite{valiant_completeness_1979}.
If we assume even further that permanents require exponential-size circuits, then we can rule out submultiplicativity on $d$-mode tensors for concrete values of $d$.
\begin{restatable}[Submultiplicativity of $\CC$ implies faster permanents]{theorem}{submultper}
\label{thm:intro-submult-per}
    For all $0 < c < 1$, 
    if circuit complexity is submultiplicative on $\lceil 4/c \rceil^2$-mode tensors, then permanents have circuits of size $O(2^{cn})$.
\end{restatable}

In particular, unless the permanent has circuits of size $2^{o(n)}$, there is a constant $D\in \mathbb N$ such that $\CC$ is not submultiplicative for $D$-mode tensors.
This assumption on the permanent is implied by a suitable non-uniform variant of the \emph{exponential-time hypothesis}.
For a more fine-grained approach, we can consider concrete values of $c>0$ in Theorem~\ref{thm:intro-submult-per}.
For example, if the permanent does not have circuits of size $2^{0.1n}$, then submultiplicativity fails for tensors with $d=\lceil4/0.1\rceil^2=40^2 = 1600$ modes.
If even circuits of size $2^{0.8n}$ can be ruled out for the permanent, then submultiplicativity fails already on $d=25$ modes; this is the smallest number $d$ of modes for which \Cref{thm:intro-submult-per} still yields a statement.
As for the plausibility of this assumption, Knuth asks in \textit{The Art of Computer Programming} \cite[{Volume 2, \S{}4.6.4, Problem~11}]{knuth1998art} whether there is any way to evaluate the permanent of a general $n\times n$ matrix using fewer than $2^n$ arithmetic operations.
This is still an open problem.
We remark that there is no asymptotic notation around $2^n$.
\footnote{Knuth's question is currently known to admit an affirmative answer only under the asymptotic rank conjecture \cite{BKKN25}; more specifically, the asymptotic rank conjecture implies uniform arithmetic circuits of size $2^{cn}$ for the permanent where $c\approx 0.9183$.}

Using another conjecture from fine-grained complexity theory, we come closer to $d=4$, the smallest number of modes where rank and circuit complexity differ:
The \emph{$(h,k)$-hyperclique conjecture}~\cite{LincolnWW18} rules out $O(n^{k-\varepsilon})$-time algorithms for detecting $k$-hypercliques in $h$-uniform hypergraphs,
and it is easy to formulate a non-uniform algebraic variant of it (see Section~\ref{sec:acc} for details).
\begin{restatable}[Submultiplicativity of $\CC$ implies faster hypercliques]{theorem}{submulthc}
\label{thm:hyperclique-submult}
For every even $d\geq 8$, the non-uniform $(\frac{d}{2}-1,\frac{d}{2})$-hyperclique conjecture rules out submultiplicativity of circuit complexity on $d$-mode tensors.
\end{restatable}

While our results suggest that circuit complexity is not submultiplicative for $d$-mode tensors, we can construct small circuits for $T\otimes U$ from a low-rank decomposition of $T$ and a small circuit for $U$.
This proves useful for us in obtaining better upper bounds on the asymptotic circuit complexity.

\begin{restatable}[Mixed asymptotic rank and circuit complexity]{theorem}{rankcirc} 
\label{thm:circuit-and-rank}
    Given $d$-mode tensors $T$ and $U$, where $T$ has asymptotic rank exponent $r\geq 1$ and $U$ has asymptotic circuit complexity exponent $s\geq 1$, the asymptotic circuit complexity exponent of $T \otimes U$ is at most $r+s$.
\end{restatable}

\paragraph{Asymptotic Circuit Complexity.}
With tools like \Cref{thm:circuit-and-rank} and insights from graph theory, we adapt the graph-theoretic proof of our asymptotic rank bound in \Cref{thm:asymptotic-submultiplicativity} to obtain stronger upper bounds for asymptotic circuit complexity rather than rank. 

First, using a known nontrivial upper bound on the treewidth of line graphs of complete graphs \cite{harvey_treewidth_2015}, 
we obtain an upper bound on the circuit complexity of generic $d$-mode tensors:
\begin{theorem}[Upper bound on asymptotic circuit complexity]
\label{thm:intro-cc-ltw}
    The asymptotic circuit complexity exponent of every $d$-mode tensor $T$ is less than $d/2+1$.
\end{theorem}
While a lower bound of $d/2$ on the asymptotic \emph{rank} exponent of $d$-mode tensors is immediate from \eqref{eq:intro-pip}, an unbounded \emph{circuit complexity} exponent in $d$ would constitute a strong circuit lower bound and hence seems exceedingly unlikely to be obtainable.

Next, we observe that \Cref{thm:intro-cc-ltw} does not yield optimal upper bounds for small fixed $d$. For example, the theorem gives a circuit complexity exponent bound of $3$ for $d=4$, but $3$ is even above the asymptotic \emph{rank} upper bound of $2.317$. 
To obtain meaningful upper bounds for small fixed $d$, we combine \Cref{thm:circuit-and-rank} with fractional decompositions of specific small graphs into low-rank and low-treewidth parts. These graph decompositions are found through computer-aided optimization. Our technique is applicable for arbitrary fixed $d$, albeit with higher computational load, and we obtain exemplary results for $d=4,5$:
\begin{theorem}[Upper bound on asymptotic circuit complexity for small orders]
\label{thm:intro-45ub}
    For every $4$-mode tensor $T$, the asymptotic circuit complexity exponent of $T$ is at most $2.2967$. For every $5$-mode tensor, it is at most $2.8774$.
\end{theorem}
For comparison, the upper bounds on the asymptotic \emph{rank} exponents from \Cref{thm:asymptotic-submultiplicativity} for $d=4,5$ are about $2.317$ and $3.0893$, respectively.
\Cref{tab:results} provides a comprehensive overview of the kinds of results we obtain. 
\begin{table}
\begin{center}
\begin{tabular}{||l c c c c c||} 
 \hline
 Method & $d=3$ & $d=4$ & $d=5$ & $d=6$ & $d=10$ \\ [0.5ex] 
 \hline\hline
 $\AR(T)$, asymptotic submult.~(Thm.~\ref{thm:asymptotic-submultiplicativity})~ & 1.59 & 2.32 & 3.09 & 3.87 & 6.96 \\ 
 \hline
 \hline
 $\AC(T)$, treewidth-based (Thm.~\ref{thm:clique-treewidth}) & 2.00 & 2.50 & 3.20 & 3.67 & 5.80 \\
 \hline
 $\AC(T)$, specialized decompositions~(Thm.~\ref{thm:4-mode-exponent}, \ref{thm:5-mode-exponent}) & -- & 2.30 & 2.88 & -- & -- \\
 \hline
 \hline
 Flattening lower bound on $\AR(T)$ from~\eqref{eq:intro-pip} & 1.00 & 2.00 & 2.00 & 3.00 & 5.00 \\ \hline
\end{tabular}
\caption{Bounds on asymptotic rank and circuit complexity exponents of generic $d$-mode tensors. The first row states rank upper bounds obtained via \cref{thm:asymptotic-submultiplicativity}. The next two rows state upper bounds on asymptotic circuit exponents. For comparison, we include the known case $d=3$ and the strongest known lower bounds.
} \label{tab:results}
\end{center}
\end{table}

\subsection{Our Techniques: New Insights into Graph Tensors}

The protagonists in our proofs are so-called \emph{graph tensors}, i.e., tensors with many modes that are composed from $2$-mode tensors by using graphs as ``composition blueprints''.
Such tensors appear in quantum information theory as \emph{matrix-product states} and \emph{projected entangled pairs}~\cite{RevModPhys.93.045003}, in quantum machine learning as (restricted) \emph{tensor networks}~\cite{Huggins_2019}, and in counting complexity (under the right abstraction) as \emph{Holant problems}~\cite{Valiant08,CL11,DBLP:conf/innovations/Cai026}. 
For us, graph tensors act as a bridge between graph theory and multilinear algebra that allows us to import nontrivial graph-theoretic concepts (e.g., graph decompositions and treewidth) and results (e.g., bounds on treewidth for specific graph classes) to the study of tensors with many modes.
Moreover, graph tensors are closed under Kronecker products, which makes them particularly well-suited for studying the asymptotic behaviour of rank and circuit complexity. We remark that graph tensors can be generalized further to tensors admitting representation by more general tensor networks or factor-graph models (e.g.~\cite{austrin_tensor_2022,KFL2001}), but graph tensors will be sufficient for our present purposes.

As a simple example, the outer tensor product of $k$ tensors on $b$ modes gives a tensor $T$ with $kb$ modes. In the extreme case of $b=1$, the resulting tensor $T$ has rank $1$, but already for $b=2$, examples of large rank such as the PIP tensor \eqref{eq:intro-pip} emerge. 
Graph tensors are obtained by additionally allowing identification (with flattening) of modes. Such tensors admit natural interpretations both as quantum states and as generic instances of Holant problems; we discuss both in the following.

\subsubsection*{Graph Tensors in Terms of Quantum States}
\label{sec:intro-graph-tensors}
In quantum information theory, $d$-mode tensors are the pure states of $d$-party systems. In the following, let $e^{(i)}_1,\ldots,e^{(i)}_r$ be a basis of the $r$-dimensional local state space of system $i\in [d]$. 
A canonical example, the \emph{Greenberger--Horne--Zeilinger} state (equivalently, the \emph{unit tensor} on $d$ modes and dimension $r$ per mode), is defined by
\[
U_{d,r} \coloneqq \sum_{s=1}^r e^{(1)}_s \otimes \ldots \otimes e^{(d)}_s.
\]
It represents $d$ parties of local dimension $r$ in genuine multipartite entanglement, and every $d$-mode tensor of rank $r$ can be obtained as a projection from $U_{d,r}$. The two-party case $U_{2,r}$ is known as a generalized \emph{Bell state}.

In important applications, the global $d$-party state decomposes into states on collections of $b$ parties for $b < d$;
these collections may overlap nontrivially.
We focus on the case $b=2$, i.e., on states that decompose into local Bell states.
These so-called \emph{graph tensors} are fully determined by entanglements between \emph{pairs} of parties and can be described by an undirected graph $H$ with $d$ vertices, where an edge $uv$ is present if parties $u$ and $v$ share a Bell state. Formally, we have
\begin{equation}
\label{eq:intro-graphstate}
    T_{H,n} \coloneqq \bigotimes_{uv \in E(H)}
    \left( \sum_{s=1}^r e_s^{(u)} \otimes e_s^{(v)} \right).
\end{equation}

As a concrete example, assume the $d$ parties correspond to vertices of a path $P$ with entanglement between adjacent parties $i$ and $i+1$ for $1\leq i < d$. 
This gives rise to the global state
\begin{equation}
\label{eq:intro-path-tensor}
T_{P,n} 
= \left( \sum_{s=1}^r e_s^{(1)} \otimes e_s^{(2)} \right) 
\otimes 
\left( \sum_{s=1}^r e_s^{(2)} \otimes e_s^{(3)} \right) 
\otimes
\ldots 
\otimes
\left( \sum_{s=1}^r e_s^{(d-1)} \otimes e_s^{(d)} \right).
\end{equation}

A more ``computational perspective'' on graph tensors is also possible, and it will be more useful when studying their complexity.
Namely, by distributing the Kronecker product over the $d-1$ pairs of parentheses in \eqref{eq:intro-path-tensor}, we can also interpret $T_{P,n}$ as a sum over $r$-ary assignments $s \colon [d-1]\to [r]$ to the $d-1$ edges between adjacent parties:
\begin{equation}
\label{eq:intro-path-edgeassignment}    
T_{P,n} = \sum_{s \colon  [d-1] \to [r]} e_{s(1)}^{(1)} \otimes \left( e_{s(1)}^{(2)} \otimes e_{s(2)}^{(2)} \right) \otimes \ldots \otimes \left( e_{s(d-1)}^{(d-1)} \otimes e_{s(d)}^{(d-1)} \right) \otimes e_{s(d)}^{(d)}.
\end{equation}

From \eqref{eq:intro-path-edgeassignment}, we see that the vector space at modes $1$ and $d$ is isomorphic to $\mathbb C^r$, while it is isomorphic to $(\mathbb C^r)^{\otimes 2}$ at all other modes. In other words, if $I(v)$ denotes the edges incident with mode $v$, then the space at mode $v$ admits a basis of vectors $e^{(v)}_a$ that are indexed by \emph{local} assignments of the form $a\colon I(v)\to [r]$.
Applying the same reasoning to general graph tensors in \eqref{eq:intro-graphstate}, we obtain:
\begin{equation}
\label{eq:intro-graphstate-assignments}
T_{H,n} = \sum_{s \colon E(H) \to [r]} \, \bigotimes_{v\in V(H)} e_{s|_{I(v)}}^{(v)},   
\end{equation}
where $s|_{I(v)}$ is the restriction of $s$ to the edges $I(v)$ incident with vertex $v\in V(H)$.

\subsubsection*{Graph Tensors in Terms of Holant Problems} 
To analyze the complexity of tensors $T_{H,n}$ for graphs $H$, we observe that they can be interpreted as \emph{Holant problems}, which are very well-studied in counting complexity~\cite{Valiant08,CL11,cai2017complexity,CGW16,SC20}:
On input a graph $G=(V,E)$, a Holant problem asks to compute a weighted count of edge-assignments $s \colon E\to [r]$, with weights determined locally at vertices;
the vast majority of the literature focuses on the case $r=2$.
As a concrete example, the number of \emph{perfect matchings} in $G$ is the number of edge-assignments $s \colon E(H)\to \{0,1\}$ such that every vertex $v \in V$ has exactly one $1$-labeled edge in the local restriction $s|_{I(v)}$ of the global assignment $s$.

In general Holant problems, the weights are determined by \emph{signatures} $f_v \colon [r]^{I(v)}\to \mathbb C$ at the vertices $v\in V$.
On input $G$ and signatures $\{f_v\}_{v\in V}$, we then wish to determine 
\begin{equation}
\label{eq:intro-holant}
\mathrm{Holant}(G)=\sum_{s \colon E\to [r]}\prod_{v\in V}f_v(s|_{I(v)}).
\end{equation}

While the analogy between \eqref{eq:intro-holant} and $T_{G,r}$ from \eqref{eq:intro-graphstate-assignments} is immediate, Holant problems have been studied from a very different perspective in the literature:
Usually, a set of possible signatures $\mathcal F$ is fixed, and the input is a graph $G$ with signatures $f_v\in \mathcal F$.
We consider the converse setting: The graph $G$ is fixed, but the signatures can vary freely in that they are provided as input vectors to the $|V(G)|$ modes of the multilinear form $T_{G,n}$.

As an example of particular relevance for us, fix $G$ as the $4$-regular toroidal grid on $t\times t$ vertices.
Then $T_{G,2}$ is a $t^2$-mode tensor with a copy of $(\mathbb C^2)^{\otimes 4}$ at each mode.
By specifying a vertex signature $f_v : \{0,1\}^4\to \mathbb C$ for each grid vertex and inputting it as a vector $u_v \in \mathbb C^{16}$ into the multilinear form $T_{G,2}$, we can count perfect matchings, Eulerian subgraphs, or evaluate any Holant problem definable on the fixed grid $G$.
In particular, for growing grid sizes, this includes $\VNP$-hard problems: By appropriately engineering the signatures, we express the permanent of arbitrary $n \times n$ matrices as a particular Holant problem on the fixed $n \times n$ \emph{grid}, and we obtain hardness of the grid graph tensor.
On the algorithmic side, known fixed-parameter tractable algorithms over low-treewidth graphs establish a nontrivial complexity upper bound of $n^{O(t)}$ on the evaluation of $T_{G,n}$, where our generic results would only yield an $n^{O(t^2)}$ bound.
Such upper and lower complexity bounds informed by (parameterized) algorithmic graph theory are ubiquitous in our arguments.

\section{Preliminaries}

For a nonnegative integer $n$ we write $[n]=\{1,2,\ldots,n\}$. We write $\mathbb{N}=\{1,2,\ldots\}$.
The Kronecker delta $\delta_{ij}$ is $1$ if $i=j$ and $0$ else.

\subsection{Algebraic Complexity}
\paragraph{Tensors.}
It will be convenient to work with tensors in coordinates and view tensors as set-multilinear polynomials with an underlying tuple of modes. 
Let $\mathbb{F}$ be a field. We write $\mathbb{F}[X]$ for the ring of polynomials over $\mathbb{F}$ in a tuple $X$ of one or more indeterminates. 
We assume the indeterminates $X$ are partitioned into a tuple of $d$ pairwise disjoint 
nonempty sets $X_1,X_2,\ldots,X_d$ called {\em modes}. 
A monomial is {\em set-multilinear} if each mode has exactly one indeterminate that has degree one in the monomial, and all other indeterminates have degree zero. A polynomial is {\em set-multilinear} if all of its monomials are set-multilinear. A {\em $d$-mode tensor} or briefly a $d$-{\em tensor} $T$ is a set-multilinear polynomial together with its $d$-tuple $(X_1,X_2,\ldots,X_d)$ of modes. The {\em shape} of $T$ is 
$|X_1|\times |X_2|\times \cdots\times |X_d|$. We also call $d$ the \emph{degree} or \emph{order} of the tensor, in line with our polynomial interpretation. We suppress $d$ and the modes whenever they are clear from the context.
For brevity, we will even suppress the variable names themselves and write $(\mathbb F^n)^{\otimes d}$ for the set of $d$-tensors over $\mathbb F$ with each mode of dimension $n$.

Addition, subtraction, and scalar multiplication of tensors with identical tuples of modes are inherited from polynomial arithmetic. 
The {\em Kronecker product} of two $d$-tensors $S\in\mathbb{F}[X]$ and $T\in\mathbb{F}[Y]$ 
with $X,Y$ disjoint is the $d$-tensor $S\otimes T$ obtained from the polynomial product $ST$ by viewing the Cartesian products $X_i\times Y_i$ for $i\in[d]$ as the $d$ modes. When the sets $X,Y$ are not disjoint, such as when taking a Kronecker power of a tensor, we tacitly assume a disjoint copy of one of the sets of indeterminates is formed before taking the Kronecker product.

\paragraph{Projection, Equivalence, and Restriction.}
A {\em simple substitution} for sets of indeterminates $X,Y$ is a map $\sigma:Y\rightarrow \mathbb{F}\cup X$. For tensors $S\in\mathbb{F}[X]$ and $T\in\mathbb{F}[Y]$ we say that 
$S$ is a {\em projection} of $T$ and write $S\pro T$ if there exists a simple substitution $\sigma$ such that the polynomial identity $S=T^\sigma$ holds and the image of each mode $Y_j$ of $T$ under $\sigma$ intersects at most one mode $X_i$ of $S$.
A simple substitution is {\em nonscalar} if it does not assume values in $\mathbb{F}$.
We say that $S$ and $T$ are {\em equivalent} and write $S\equ T$ if the polynomial identity $S=T^\sigma$ holds for a bijective nonscalar $\sigma$ that is bijective on modes. 
For $d$-tensors $S$ and $T$ we say that $T$ {\em restricts to} $S$ and write $S\res T$ if 
the polynomial identity $S=T^{\mu}$ holds for a substitution $\mu:Y\rightarrow\mathbb{F}[X]$
such that for all for all modes $Y_j$ and all indeterminates $y\in Y_j$ the polynomial
$\mu(y)$ is a linear polynomial in the indeterminates $X_j$.

\paragraph{Tensor Rank.} A nonzero $d$-tensor $T$ has {\em rank one} if there exist linear
polynomials $\ell_j\in\mathbb{F}[X_j]$ for $j\in[d]$ such that the polynomial identity 
$T=\ell_1\ell_2\cdots\ell_d$ holds. The {\em rank} $\RR(T)$ of a $d$-tensor $T$ is the minimum number of rank one tensors whose sum is $T$. 
(Note that $d=2$ recovers matrix rank.) For two $d$-tensors $S$ and $T$, we have $\RR(S\otimes T)\leq\RR(S)\cdot\RR(T)$, with equality for $d=2$ but strict inequality may hold for $d\geq 3$. This property is commonly referred to as the \emph{submultiplicativity} of tensor rank. 
Moreover, the rank of any tensor $T\in(\mathbb{F}^n)^{\otimes d}$ is obviously bounded by $n^d$, the number of set-multilinear monomials (and the stronger bound $n^{d-1}$ can easily be shown).
Consequently, $T^{\otimes k}$ has rank at most $n^{(d-1)k}$.
The {\em asymptotic rank} of $T\in(\mathbb{F}^n)^{\otimes d}$
is defined as the limit $\AR(T) = \lim_{k\rightarrow\infty}\RR(T^{\otimes k})^{1/k}$, which exists by submultiplicativity of rank and Fekete's Lemma (the result is folklore, but see~\cite[Lemma 2.10]{wigderson2022asymptotic} and the references given there). Like rank, also the asymptotic rank satisfies submultiplicativity with respect to Kronecker products. We say that a $d$-tensor is {\em concise} if all of its $d$ flattenings to a matrix with one mode forming the rows and all the other modes forming the columns have full rank. A concise tensor $T\in (\mathbb F^n)^{\otimes d}$ in particular satisfies both $\RR(T)\geq n$ and $\AR(T)\geq n$. The \emph{exponent} of a concise tensor $T \in (\mathbb F^n)^{\otimes d}$ is defined by 
\[
\omega(T) = \inf\,\{ \beta \geq 0 : \RR(T^{\otimes k}) = O(n^{\beta k})\}
\]
and we have 
\[
\AR(T) = n^{\omega(T)}\,.
\]
For example, the tensor $\mathrm{MM}_m \in \mathbb F^{m^2} \otimes \mathbb F^{m^2} \otimes \mathbb F^{m^2}$ that represents the bilinear map that multiplies two $m\times m$ matrices for any constant $m \geq 2$ is concise and recovers the exponent $\omega$ of square matrix multiplication either via the asymptotic rank identity $\AR(\mathrm{MM}_m)=m^\omega$ or, equivalently, via the tensor 
exponent identity $\omega(\mathrm{MM}_m)=\omega/2$.

\paragraph{Arithmetic Circuits.}
Tensor rank fails to capture the arithmetic complexity of a tensor for $d \geq 4$, see Example~\ref{ex:matching} below.
In contrast, circuits provide a more faithful model of arithmetic complexity.
An \emph{arithmetic circuit over $\mathbb F[X]$} is a directed acyclic graph (DAG) with sinks called \emph{outputs}, sources called \emph{inputs}, vertices called \emph{gates}, and directed edges called \emph{wires.} 
Non-input gates are labeled with $+$ and $\times$, while inputs are labeled with elements from $\mathbb F \cup X$, and wires are labeled with elements of $\mathbb F$. 
A gate \emph{computes} a polynomial in the obvious inductive manner, and the \emph{size} of a circuit is the number of its wires.
For a tensor $T \in (\mathbb F^n)^{\otimes d}$, we write $\CC(T)$ for the size of a smallest circuit computing $T$, called the \emph{arithmetic circuit complexity} of $T$.

A rank-$r$ decomposition of $T \in (\mathbb F^n)^{\otimes d}$ can be converted to an arithmetic circuit for $T$ of size $O(d\cdot n \cdot r)$.
Hence, the trivial bound $\CC(T^{\otimes k})\leq O(d\cdot n^{(d+1)k})$ allows us to define, in analogy with the tensor case, the \emph{circuit exponent} $\eta(T)$ and the \emph{asymptotic circuit complexity} $\AC(T)$ of $T$ via
\[
\eta(T) = \inf\{\beta \geq 0 : \CC(T^{\otimes k}) \leq O(n^{k\beta}) \},\quad \AC(T) = n^{\eta(T)}\,.
\]
We note here that, again, $\AC(T) = \limsup_{k\to\infty} \CC(T^{\otimes k})^{1/k}$ can be shown, but it is not clear whether the limit itself exists, as it does for rank.

The crude upper bound of $\CC(T)=O(d\cdot n \cdot \RR(T))$ only proves $\eta(T) \leq \omega(T)+1$.
However, a classic result of Yates~\cite{yates} 
(see also Knuth~\cite[\S{}4.6.4]{knuth1998art})
shows that for $d \geq 2$ and concise $T\in(\mathbb{F}^n)^{\otimes d}$, we have that for all $k \geq 1$, 
    \begin{align}\label{eq:yates-bound}
    \CC(T^{\otimes k}) = O(dk \cdot \RR(T)^{k+1}).
    \end{align}
All treatments in the literature we are aware of focus on the cases $d=2,3$ of Kronecker powers of matrices and $3$-tensors, but the result easily generalizes to $d \geq 4$; for completeness we include a short exposition in Appendix~\ref{sec:app-circuits}
and record the asymptotic conclusion in the following theorem.
\begin{theorem}[Yates's algorithm] \label{thm:yates}
    Let $d \geq 2$ and let $T \in (\mathbb F^n)^{\otimes d}$ be concise. Then, 
    \[
    \eta(T) \leq \omega(T),\ \text{and therefore}\ \AC(T) \leq \AR(T).
    \]
\end{theorem}

Strikingly, for $d=3,$ it is true that $\AC(T) \geq \AR(T)$ holds as well, providing a justification for considering tensor rank as a measure of arithmetic complexity. As mentioned, this behavior does not extend to $d \geq 4.$

\paragraph{Closure Properties.}
Since restrictions of rank-one tensors are themselves of rank one (or identically zero),  $\RR(S) \leq \RR(T)$ whenever $S \res T$, and hence $\RR(S) = \RR(T)$ if $S \equ T.$
The same statements hold for asymptotic rank. 
For circuits, it is only true that $\CC(S) \leq \CC(T)$ when $S \pro T$ for $S,T\in(\mathbb{F}^n)^{\otimes d}$ while $S \res T$ only implies $\CC(S) \leq \CC(T) + n^{d+1}$. Still, asymptotic circuit complexity again admits $\AC(S) \leq \AC(T)$ if $S \res T$ as a consequence of Yates's algorithm~\eqref{eq:yates-bound}, see Fact~\ref{fact:res-submult-circuit} in the appendix. Moreover, we do have that $\CC(U) \leq \CC(U \otimes V)$ provided $V \neq 0$, see \cref{fact:prod-res}.
The questions of whether and when $\AC$ and $\CC$ are robust with respect to Kronecker products will be a main theme of the article.

A very general description of these properties in relation to tensor rank (and far beyond) is laid out in the work of Wigderson and Zuiddam~\cite{wigderson2022asymptotic}.
For further background on algebraic complexity we refer to the book of Bürgisser, Clausen and Shokrollahi~\cite{burgisser_algebraic_1997}.

\subsection{Graph Theory}

\paragraph{Graphs.} A {\em graph} $G$ consists of a finite set $V(G)$ of {\em vertices} and a finite set $E(G)$ of {\em edges} such that each edge is associated with a set of two vertices called the {\em end-vertices} or {\em ends} of the edge. We stress that multiple edges may have the same set of ends; such edges are called {\em parallel} edges.  For a vertex $v\in V(G)$ we write $I_G(v)\subseteq E(G)$ for the set of all edges that have $v$ as an end, and call $|I_G(v)|$ the \emph{degree} of $v$. The maximum degree of a vertex in $G$ is denoted by $\Delta(G).$

For two graphs $G$ and $H$ with disjoint edge sets, the {\em sum} $G+H$ is the graph defined by $V(G+H)=V(G)\cup V(H)$ and $E(G+H)=E(G)\cup E(H)$. When the edge sets are not disjoint, such as when taking the sum $G+G$, we tacitly assume a disjoint copy of one of the edge sets is formed before taking the sum. 

We say that two graphs $G$ and $H$ are {\em isomorphic} and write $G\isom H$ if there exist bijections $\phi\colon V(G)\rightarrow V(H)$ and $\psi\colon E(G)\rightarrow E(H)$ such that each edge $e\in E(G)$ has the end-vertices $\{v_1,v_2\} \subseteq V(G)$ if and only if the edge $\psi(e)$ has the end-vertices $\{\phi(v_1),\phi(v_2)\}\subseteq V(H)$. Equivalently, $\psi(I_G(v)) = I_H(\phi(v))$ should hold.

\paragraph{Subgraphs.} A graph $H$ is a {\em subgraph} of a graph $G$ if $V(H)\subseteq V(G)$, $E(H)\subseteq E(G)$, and common edges of $H$ and $G$ have identical end-vertices. 
A {\em subdivision} of a graph $G$ is a graph obtained by applying the following operation zero or more times: select an edge $e$ with end $u$ and $v$, insert a new vertex $w$, delete the edge $e$, insert a new edge with ends $u$ and $w$, and insert a new edge with ends $v$ and $w$. A graph $H$ is a {\em topological subgraph} of a graph $G$ if there exists a subdivision of $H$ that is isomorphic to a subgraph of $G$.

\paragraph{Treewidth.} 
A {\em tree decomposition} of a graph $G$ is a pair $(T,\beta)$ where $T$ is a tree and $\beta$ is a map that associates each vertex $a\in V(T)$ with a {\em bag} $\beta(a)\subseteq V(G)$ such that (i) the union of all bags is $V(G)$; (ii) both ends of each edge of $G$ are contained in at least one bag; and (iii) for all $a,c\in V(T)$ it holds that $\beta(a)\cap\beta(c)\subseteq\beta(b)$ for all $b\in V(T)$ on the path joining $a$ and $c$ in $T$. 
The {\em width} of $(T,\beta)$ is $\max_{a\in V(T)}|\beta(a)|-1$. The {\em treewidth} $\tw G$ is the minimum width of
a tree decomposition of $G$.
More background on algorithmic aspects on graph theory, in particular algorithms exploiting tree decompositions of small width, can be found in the textbook of Cygan et al.~\cite{CyganFKLMPPS15}.

\section{Graph Tensors and the Asymptotic Rank} \label{sec:graph-tensors}

This section develops our conventions for graph tensors together with their basic algebraic properties in relation to the underlying graphs, and then proceeds to prove our main theorem (\cref{thm:asymptotic-submultiplicativity}) on bounding the asymptotic rank for $d$-mode tensors for $d\geq 4$, extending Strassen's result for $d=3$.

\subsection{Graph Tensors}

Our graphs are undirected, loopless, and may have parallel edges.
Let $G$ be a graph and let $n\in\mathbb{N}$. For a mapping $f\colon E(G)\rightarrow [n]$ and a vertex $v\in V(G)$,
let us write $f|_{I(v)}$ for the restriction of $f$ to the set $I(v)$ of edges incident to $v$ in $G$. For each vertex $v\in V(G)$, introduce a mode $X^{(v)}=\{x^{(v)}_{g} \mid g\colon I(v)\rightarrow [n]\}$. 

\begin{definition}[Graph tensor]
The {\em graph tensor} $T_{G,n}$ of the graph $G$ with {\em length parameter} $n$
has the modes $(X^{(v)}:v\in V(G))$ and is defined by the polynomial identity
\begin{equation}\label{eq:graph-tensor}
    T_{G,n} = \sum_{f \colon E(G)\to [n]}\prod_{v\in V} x^{(v)}_{f|_{I(v)}} .
\end{equation}
Note that $T_{G,n}$ has $|V(G)|$ modes and the {\em length} of 
the mode $X^{(v)}$ for $v\in V(G)$ is $|X^{(v)}|=n^{|I(v)|}$. Graph tensors are concise.
\end{definition}

\begin{remark}[The Holant view]
As laid out in the introduction, graph tensors and Holant problems are intimately connected. Namely, the tensor $T_{G,n}$ essentially captures all Holant problems definable on the fixed graph $G$ by choosing concrete signatures $s_v:[n]^{I(v)}\to \mathbb C$ for the vertices $v\in V(G)$ and substituting $x_{a}^{(v)} \gets s_v(a)$ for all $a \in [n]^{I(v)}$.
\end{remark}

\begin{example}[The $k$-matching tensor]\label{ex:matching}
    Let us write $H_k$ for the $k$-{\em matching graph} consisting of vertices $u_i$ and $v_i$ joined by an edge for $i \in [k]$ and let $n \in \mathbb{N}$. The {\em $k$-matching tensor} $T_{H_k,n}$ is known as
    the \emph{product of inner products} in algebraic complexity theory, and in a quantum-theoretic interpretation it corresponds to $n$ disjoint pairs of parties such that each pair holds a Bell state, and no further entanglement is present. For our purposes,
    the tensor $T_{H_k,n}$ for $k\geq 2$ provides an example of the separation of algebraic complexity and tensor rank for tensors 
    with at least four modes. Namely, we observe the polynomial identity 
    \[
        T_{H_k, n} \equ \sum_{f \colon [k] \to [n]} \prod_{i \in [k]} x^{(u_i)}_{i \mapsto f(i)} x^{(v_i)}_{i \mapsto f(i)}
        = \prod_{i \in [k]} \sum_{j \in [n]} x^{(u_i)}_{i \mapsto j} x^{(v_i)}_{i \mapsto j}\,,
     \]
     where the last formula establishes $\CC(T_{H_k, n}) \leq 2 k n$,
     yet we have $\RR(T_{H_k, n}) = n^k$. To see that the tensor rank
     is at most $n^k$, observe the first formula above; to see that
     the rank is at least $n^k$, recall that the tensor rank of a tensor is bounded from below by the matrix rank of any flattening of the tensor into a matrix, and study the $n^k\times n^k$ matrix flattening of $T_{H_k,n}$ defined by
     the modes $u_1,\ldots,u_k$ and $v_1,\ldots,v_k$. In particular observe that the flattening is, up to permutation of rows and columns, the $n^k\times n^k$ identity matrix.
\end{example}

Several useful connections between algebraic properties of $T_{G,n}$ and graph-theoretic properties of $G$ can be shown. Crucially for us, the Kronecker product of graph tensors of graphs $G$ and $H$ is the graph tensor of the \emph{sum graph} $G+H$, as defined in the preliminaries. (Recall that the sum may introduce multiedges.) The proof is elementary.

\begin{lemma}[Product of graph tensors corresponds to graph sum]\label{lem:decomp}

    Let $G$ and $H$ be graphs and let $n\in\mathbb{N}$. Then,
    \[
    T_{G, n} \otimes T_{H, n} \equ T_{G + H, n}\,.
    \]
\end{lemma}
\begin{proof}
Let us set $I_{H}(v) = \emptyset$ for all $v \in V(G) \setminus V(H)$, and, symmetrically, $I_G(v) =\emptyset$ for all $v \in V(H)\setminus V(G)$. By bilinearity of the Kronecker product and the definition of the sum $G+H$, we have 
    \begin{align*}
        T_{G, n} \otimes T_{H, n} &=
        \left ( \sum_{f\colon E(G)\to [n]}\prod_{v\in V(G)} x^{(v)}_{f|_{I_G(v)}} \right)
        \otimes
        \left( \sum_{f'\colon E(H)\to [n]}\prod_{v\in V(H)} x^{(v)}_{f'|_{I_H(v)}} \right) \\ 
        &=\sum_{\subalign{f\colon E(G)&\to [n] \\ f'\colon E(H)&\to [n]}}
        \left( \prod_{v\in V(G)} x^{(v)}_{f|_{I_G(v)}} \right) \otimes \left( \prod_{v\in V(H)} x^{(v)}_{f'|_{I_H(v)}} \right) \\
        &\equ\sum_{\subalign{f\colon E(G)&\to [n] \\ f'\colon E(H)&\to [n]}}
        \prod_{v\in V(G)\cup V(H)} x^{(v)}_{f|_{I_G(v)}, f'|_{I_H(v)}} \\
        &\equ \sum_{f\colon E(G + H) \to [n]}\prod_{v\in V(G+H)} x^{(v)}_{f|_{I_{G + H}(v)}} \\
        &= T_{G + H, n}. \qedhere
    \end{align*}
\end{proof}

\begin{remark}[Single-edge decomposition]
\label{rem:edge-decomposition}
Christandl and Zuiddam~\cite{ChristandlZ19} used a special case of Lemma~\ref{lem:decomp} to \emph{define} $T_{G,n}$. Specifically, let $G=(V,E)$ with $E=\{e_1,\ldots,e_m\}$, and let $F_i=(V,\{e_i\})$ for $i \in [m]$ be the graph on vertex set $V$ that contains only $e_i$ as edge. Then $G=F_1+\ldots+F_m$ and 
\begin{equation}
\label{eq:decomp-single-edges}
T_{G,n}\equ T_{F_1,n}\otimes \ldots \otimes T_{F_m,n}.    
\end{equation}
Lemma~\ref{lem:decomp}, which was not stated in their paper, then follows by commutativity of $\otimes$ up to relabeling of modes.
Their definition in terms of sums of single-edge graphs is arguably more elegant, but our definition in 
terms of the explicit monomial expansion of $T_{G,n}$ in \eqref{eq:graph-tensor} enables a direct connection to algebraic
complexity, as already witnessed
for the $k$-matching tensor in \cref{ex:matching}.
\end{remark}

\begin{remark}[Length rule and sum rule]\label{rem:flip}
    With the same proof as in Lemma~\ref{lem:decomp}, for $n_1,n_2\in\mathbb{N}$ we have the length rule 
    \[
    T_{G,n_1} \otimes T_{G,n_2} \equ T_{G,n_1\cdot n_2}\,.
    \]
    Indeed, the set of pairs of mappings $(f_1,f_2)$ with $f_i:\ E(G) \rightarrow [n_i]$ is in bijective correspondence with the set of mappings $E(G)\rightarrow [n_1 \cdot n_2]$.    
    Consequently, writing $2\cdot G = G+G$ for any graph $G$, and more generally $k\cdot G = G + \cdots + G$ for the $k$-fold sum for $k > 2$, we have the sum rule
    \[
        T_{k \cdot G,n} \equ T_{G,n^k}\,.
    \]
\end{remark}

The following lemma is similarly elementary, and we postpone the proof to \cref{sect:graph-tensor-properties}.

\begin{lemma}[Projection under topological subgraphs]\label{lem:graph-tensor-topological-subgraph}
    Let $G$ and $H$ be graphs and let $n\in\mathbb{N}$.
    If $H$ is isomorphic to a topological subgraph of $G$,
    then $T_{H,n}\pro T_{G,n}$.
\end{lemma}

Moreover, we will use \emph{contractions} of graphs: 
Given a set $U\subseteq V(G)$, the graph $G/U$ is obtained by replacing $U$ by a single vertex $w$ incident with all edges \emph{leaving} $U$; multiedges can be created in this process.
Contractions can be executed cheaply if the set of involved edges is not too large:
\begin{lemma}[Complexity of contractions]
\label{lem:contraction}
    For every $G$ and $n \in \mathbb{N}$
    and $U \subseteq V(G)$,
    it holds that $\CC(T_{G, n}) \leq  \CC(T_{G/U, n}) + |U|\cdot n^{a(U)}$, where $a(U) \coloneqq \sum_{u \in U} \deg_G(u) - |E(G[U])|$ is the number of edges incident with vertices in $U$.
\end{lemma}

A proof is provided in \cref{sect:graph-tensor-properties}. We remark that such contractions for graph tensors have been extensively studied in the context of probabilistic inference in factor graphs and Bayesian networks, cf.~\cite{BD2007,D1999,KFL2001,LS1988,Pearl1986,RD2000,SAS1994,SS1990,ZP1994}, and even a contraction-based model of computation for multilinear forms was studied~\cite{austrin_tensor_2022}.

\subsection{Asymptotic Rank}

We now prove the main result of this section, 
\cref{thm:asymptotic-submultiplicativity}. It will be convenient
to work with exponents of graphs, in line 
with Christandl, Vrana and Zuiddam~\cite{ChristandlVZ19}. 
Define the \emph{exponent} of a graph $G$ as 
\begin{align*}
    \omega(G) \coloneqq \inf\,\{ \beta \mid \RR(T_{G,n}) = O(n^\beta)\}\,.
\end{align*}
A related quantity is the \emph{exponent per edge} of $G$, defined as
\begin{align*}
    \tau(G) \coloneqq \omega(G)/|E(G)|\,.
\end{align*}
The following useful property can be shown in a standard manner. The proof is deferred to Appendix~\ref{sect:graph-tensor-properties}.
\begin{lemma}[Sum rule for exponents] \label{lem:power-mult}
Let $G$ be a graph and let $k \in \mathbb N$. Then, 
    \begin{align} \label{ref:omega_sum}
        \omega(k \cdot G) = k \cdot \omega(G).
    \end{align}
\end{lemma}

To motivate the use of graph tensors in our next result, let us recall Strassen's upper bound on the maximal asymptotic rank of bilinear maps (i.e., 3-mode tensors)
together with a high-level intuition of its proof. This property is often referred to as asymptotic submultiplicativity.
\begin{theorem}[Strassen~\cite{Strassen1988}; Upper bound on asymptotic rank for $3$-tensors] \label{thm:strassen}
Let $T \in (\mathbb F^n)^{\otimes 3}$ be an arbitrary $3$-mode tensor.
Then, 
\[
\AR(T) \le n^{2\omega/3}.
\]
\end{theorem}

Strassen's proof of \Cref{thm:strassen} works with the universal tensor corresponding to the canonical bilinear map $U_3: \mathbb F^n \times \mathbb F^n \rightarrow \mathbb F^n \otimes \mathbb F^n,\ (u,v) \mapsto u \otimes v$, which can be written down in coordinates as a $3$-tensor in our notation as 
\[
U_3 =\sum_{i,j\in[n]} x_i y_j z_{ij}\,.
\]
In the language of graph tensors, 
the tensor $U_3$ with the three modes $x,y,z$ admits immediate representation as a graph tensor, namely we have $U_3\equ T_{P_2,n}$, where $P_2$ is the path graph with two edges.
We can now generalize this observation on three modes to $d > 3$ modes as follows. The universal $(d-1)$-linear map $U_d: \mathbb F^n \times \cdots \times \mathbb F^n \rightarrow (\mathbb F^n)^{\otimes {d-1}}$ is represented as the $d$-mode tensor
\[
U_d = \sum_{i_1,\ldots,i_{d}\in[n]} x^{(1)}_{i_1} \cdots x_{i_{d-1}}^{(d-1)} \cdot x^{(d)}_{i_1,\ldots,i_{d-1}}\,.
\]
In the language of graph tensors, we observe that $U_d \equ T_{S_d,n}$, where $S_d$ is the star graph with $d-1$ edges, and $S_3\equ P_2$ in particular. 
This fact enables us to consider universality under Kronecker
powers combinatorially from the perspective of graph sums and graph decompositions (cf.~\Cref{lem:decomp}) and arrive at our main theorem. 
In fact, all our upper bounds are a consequence of this observation.

We now start the work towards our main theorem. 
Let $d \geq 3$ and write $K_d$ for the complete graph on $d$ vertices. We recall that $\tau(K_d) = \omega(K_d)/{\binom{d}{2}}$, and write $\tau(d) := \tau(K_d)$ for brevity in what follows.
\begin{remark}[The exponents $\tau(d)$ for $d\geq 3$]
    Recall that we write $\omega$ for the exponent of square matrix multiplication. Since $T_{K_3,n}$ is equivalent to the tensor
    of $n\times n$ matrix multiplication, we have $\tau(3) = \omega/3$.
    Moreover, it is known~\cite[Proposition~1.31]{ChristandlVZ19} that $\tau(d) \leq \tau(e)$ whenever $d \geq e.$
\end{remark}

We have the following generalization of Strassen's result in the
language of graph tensors and graph exponents. 

\begin{lemma}[Generalization of Strassen's upper bound to $d$ modes]\label{lem:strassen-generalized}
Let $T \in (\mathbb{F}^n)^{\otimes d}$ be an arbitrary $d$-mode tensor with $d \geq 3$. Then, 
\[
\AR(T) \le n^{(d-1)\tau(d)}.
\]
\end{lemma}
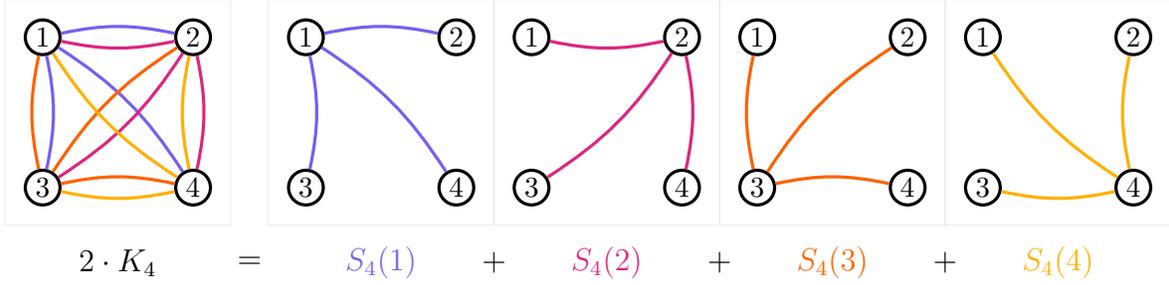
\begin{figure}[t]
\centering
\begin{tikzpicture}[main node/.style={circle,draw,very thick,inner sep=1.5pt}]
\clip (-8,-1.05) rectangle (8,3);
\draw [black,opacity=0] (-8,-1.05) rectangle (8,3);
\def\recsize{1}
\def\bdsize{0.5}

\def\addclaw#1#2#3#4#5
{
\path[very thick,#5]
(#1) edge [bend left=12] node {} (#2)
     edge [bend left=12] node {} (#3)
     edge [bend left=12] node {} (#4);
}
\def\defvertices
{
    \draw [gray!30,opacity=0.5] ({-\bdsize-\recsize},{-\bdsize-\recsize}) rectangle ({\bdsize+\recsize},{\bdsize+\recsize});
    \node[main node] (1) at ({-\recsize},{\recsize})  {1};
    \node[main node] (2) at ({\recsize}, {\recsize})  {2};
    \node[main node] (3) at ({-\recsize},{-\recsize}) {3};
    \node[main node] (4) at ({\recsize}, {-\recsize}) {4};
}

\begin{scope}[shift={(-6.25,1.25)}]
\defvertices{}
\addclaw{1}{2}{3}{4}{cbfp1}
\addclaw{2}{3}{4}{1}{cbfp2}
\addclaw{3}{4}{1}{2}{cbfp3}
\addclaw{4}{1}{2}{3}{cbfp4}
\end{scope}

\begin{scope}[shift={(-2.75,1.25)}]
\defvertices{}
\addclaw{1}{2}{3}{4}{cbfp1}
\end{scope}

\begin{scope}[shift={(0.25,1.25)}]
\defvertices{}
\addclaw{2}{3}{4}{1}{cbfp2}
\end{scope}

\begin{scope}[shift={(3.25,1.25)}]
\defvertices{}
\addclaw{3}{4}{1}{2}{cbfp3}
\end{scope}

\begin{scope}[shift={(6.25,1.25)}]
\defvertices{}
\addclaw{4}{1}{2}{3}{cbfp4}
\end{scope}

\node [] () at (-6.25,-0.75) {\large $2\cdot K_4$};
\node [] () at (-4.5,-0.75)  {\large $=$};
\node [] () at (-2.75,-0.75) {\large \textcolor{cbfp1}{$S_4(1)$}};
\node [] () at (-1.25,-0.75) {\large $+$};
\node [] () at (0.25,-0.75)  {\large \textcolor{cbfp2}{$S_4(2)$}};
\node [] () at (1.75,-0.75)  {\large $+$};
\node [] () at (3.25,-0.75)  {\large \textcolor{cbfp3}{$S_4(3)$}};
\node [] () at (4.75,-0.75)  {\large $+$};
\node [] () at (6.25,-0.75)  {\large \textcolor{cbfp4}{$S_4(4)$}};
\end{tikzpicture}
\caption{Decomposition of $2\cdot K_d$ into $d$ stars, \eqref{eq:star-clique}, for $d=4$. }
\label{fig:2kd-decomp}
\end{figure}
\begin{proof}
For $u \in [d]$, let us write $S_d(u)$ for the $d$-vertex star graph
with vertex set $[d]$ such that $u$ is the unique {\em center} vertex of degree $d-1$. An arbitrary $d$-mode tensor $T \in (\bbF^n)^{\otimes d}$ admits representation in coordinates as the polynomial
\begin{equation}
\label{eq:t-rep}
    T = \sum_{i_1,\dots,i_d \in [n]} T_{i_1,\dots,i_d} x^{(1)}_{i_1} \cdots x^{(d)}_{i_d} = \sum_{i_1,\dots,i_{d-1}\in [n]}
    x^{(1)}_{i_1}
    \cdots x^{(d-1)}_{i_{d-1}}
    \left( \sum_{i_d\in [n]} T_{i_1,\dots,i_d} x_{i_d}^{(d)} \right).
\end{equation}
Explicitly, the linear substitution
\[
    x_{i_1,\dots,i_{d-1}} \mapsto \sum_{i_d\in [n]} T_{i_1,\dots,i_d} x_{i_d}^{(d)}
\]
into the $d$-tensor $T_{S_d(d),n}$ shows 
by \eqref{eq:t-rep} and symmetry that
\begin{align}
\label{eq:star-restriction}
    T \leq T_{S_{d}(u),n}\quad\text{for all $u\in [d]$}\,.
\end{align}
We now observe the graph-sum identity (see \Cref{fig:2kd-decomp})
\begin{align}
\label{eq:star-clique}
    S_d(1)+\cdots +S_d(d) \isom 2 \cdot K_d\,,
\end{align}
where we recall that $K_d$ is the complete graph on $d$ vertices and $2 \cdot K_d$ is the graph obtained from $K_d$ by taking two copies of each edge. From \eqref{eq:star-restriction}, \Cref{lem:decomp}, \eqref{eq:star-clique}, 
and \Cref{rem:flip} thus
\begin{align} \label{eq:generic-projection}
    T^{\otimes d} \leq T_{2\cdot K_d, n} \equ T_{K_d, n^2}.
\end{align}
By definition of $\tau(d)$ and the properties of $\AR$, it follows
from \eqref{eq:generic-projection} that 
\[
\AR(T) \leq \AR(T_{K_d,n})^{2/d} \leq n^{2\binom{d}{2}\tau(d)/d} = n^{(d-1)\tau(d)}.\qedhere
\]
\end{proof}

For $d=3$, \Cref{lem:strassen-generalized} replicates  Strassen's upper bound from \Cref{thm:strassen}, and generalizes it for $d \geq 4$. Our main result of this
section is the following upper bound for $d\geq 4$ that combines
\cref{lem:strassen-generalized} with an improved upper bound for $\tau(d)$.

\strassengeneral*
\begin{proof}
    Christandl, Vrana, and Zuiddam~\cite{ChristandlVZ19} showed that $\AR(T_{K_d}) \leq n^{\binom{d}{2} \log_7(9/2)}$ for $d \geq 4$, which implies $\tau(d) < 0.7729$. 
    By a careful analysis and modification of their argument, we obtain the sharper bound $\tau(d) < 0.772318$; see Appendix~\ref{sec:ar}.
    The theorem then follows by \cref{lem:strassen-generalized}.
\end{proof}
\begin{remark}[Bounds on $\tau(d)$ for $d\geq 4$] \label{rem:strassen}
    To obtain the constant $0.772318$ in \cref{thm:asymptotic-submultiplicativity}, an upper bound for $\tau(4)$ suffices since $\tau(d) \leq \tau (4)$ for $d \geq 4.$ It is natural to ask whether better upper bounds can be obtained by finding better upper bounds on $\tau(d)$ for $d\geq 4$. 
    The best-known lower bounds on $\tau(d)$ allow for a scenario in which $\tau(d) \leq 1/2 + o(1)$, implying $\AR(T) \leq n^{(d-1)/2 + o(d)}$ would follow for $d$-mode tensors.
    On the other hand, the current best-known upper bound on $\omega$ yields an exponent greater than the exponent $(d-1)\cdot 0.772318$ in \cref{thm:asymptotic-submultiplicativity}, while $\omega = 2$ would imply $\tau(3) = \tau(4) = 2/3.$ It is an open problem whether there is a $d$ such that $\tau(d) < 2/3.$
    Recalling from \cref{ex:matching} that the $k$-matching tensor $T_{H_k,n}$ has $d=2k$ modes and its rank $\RR(T_{h_k,n})=n^k$ equals a matrix-flattening-rank lower bound, this lower bound applies to also asymptotic rank and thus $\AR(T_{H_k,n})=n^k$, implying $\tau(d)\geq 1/2$ by \cref{thm:asymptotic-submultiplicativity} for all $d\geq 4$.
\end{remark}

\section{Asymptotic Circuit Complexity} \label{sec:acc}

This section proceeds to study tensors with $d\geq 4$ modes from the standpoint of the asymptotic circuit complexity of a tensor viewed as a set-multilinear polynomial. Recalling \Cref{ex:matching} and our discussion in the introduction, for $d\geq 4$ the asymptotic circuit complexity and the asymptotic tensor rank are no longer tightly connected as is the case for $d=3$. Yet families tensors of interest from the standpoint of algorithms and complexity (e.g. matrix permanent, hypercliques, general convolutions, iterated matrix multiplication, \dots) can be captured as Kronecker powers of individual base tensors with $d\geq 4$ modes, making the asymptotic circuit complexity of these constant-size base tensors worth studying. Graph tensors provide a convenient tool for this study. 

\subsection{Complexity of Kronecker Powers} \label{subsec:impossible}
We ask whether tensors $U,T \in(\mathbb F^n)^{\otimes d}$ of low circuit complexity also have Kronecker products $U\otimes T$ of low circuit complexity. Recalling that tensor rank is submultiplicative for tensors of all orders, a natural first hope promoted by the $d=3$ connection of circuit complexity to tensor rank would be to establish the {\em submultiplicativity} property $\CC(T\otimes U) \leq \CC(T) \cdot \CC(U)$ of the complexity measure $\CC$ for all orders $d \in \mathbb N$. However, already at $d=4$ submultiplicativity would have a breakthrough consequence.  
\begin{theorem}[Fast matrix multiplication under submultiplicativity]
Submultiplicativity of $\CC$ on $4$-mode tensors implies $\omega = 2.$
\end{theorem}
\begin{proof}
The cycle graph $C_4$ is a sum of two edge-disjoint matchings. By virtue of Lemma~\ref{lem:decomp} and Example~\ref{ex:matching}, submultiplicativity implies $\CC(T_{C_4,n}) \leq O(n^2)$.
Moreover, $T_{C_3,n}$ is a projection of $T_{C_4,n} \equ \sum_{i,j,k,\ell} x_{ij}y_{jk}z_{k\ell}w_{\ell i}$ via $z_{k\ell} \mapsto \delta_{k\ell}$, see \cref{lem:graph-tensor-topological-subgraph}. Hence, $\CC(T_{C_3,n}) \leq O(n^2)$ and $\omega = \omega(C_3)=2$.
\end{proof}
The remainder of this subsection strengthens the case for the absence of submultiplicativity and illustrates the serendipity of graph tensors as a tool when working with tensors with more modes. 

We note that that submultiplicativity for $d$ modes implies the same for $d' < d,$ as discussed in \Cref{fact:submult-montone} in Appendix~\ref{sec:app-circuits}.

\subsubsection*{Reductions from Permanents}

\begin{figure}
\centering
\begin{tikzpicture}
\clip (-8,-1.8) rectangle (8,3);
\draw [black,opacity=0] (-8,-1.8) rectangle (8,3);

\def\addgrid#1#2#3#4#5#6#7#8{

\draw [gray!30,opacity=0.5] (-0.5,-0.5) rectangle (5.5,5.5);
\foreach \x in {0,...,5}
\foreach \y in {0,...,5} {
    \node [circle, draw=black, inner sep=2pt] (g\x\y) at ({\x},{\y}) {};
}

\foreach \x in {0,2,4} {
    \pgfmathparse{int(1+\x)} \xdef\xx{\pgfmathresult}
    \foreach \y in {0,2,4} {
        \pgfmathparse{int(1+\y)} \xdef\yy{\pgfmathresult}
        \path [draw=#1, #5] (g\x\y) edge node []  {} (g\x\yy);
        \path [draw=#3, #7] (g\x\y) edge node []  {} (g\xx\y);
    }
    \foreach \y in {1,3} {
        \pgfmathparse{int(1+\y)} \xdef\yy{\pgfmathresult}
        \path [draw=#2, #6] (g\x\y) edge node []  {} (g\x\yy);
        \path [draw=#3, #7] (g\x\y) edge node []  {} (g\xx\y);
    }
    \path [draw=#3, #7] (g\x5) edge node []  {} (g\xx5);
}
\foreach \x in {1,3} {
    \pgfmathparse{int(1+\x)} \xdef\xx{\pgfmathresult}
    \foreach \y in {0,2,4} {
        \pgfmathparse{int(1+\y)} \xdef\yy{\pgfmathresult}
        \path [draw=#1, #5] (g\x\y) edge node []  {} (g\x\yy);
        \path [draw=#4, #8] (g\x\y) edge node []  {} (g\xx\y);
    }
    \foreach \y in {1,3} {
        \pgfmathparse{int(1+\y)} \xdef\yy{\pgfmathresult}
        \path [draw=#2, #6] (g\x\y) edge node []  {} (g\x\yy);
        \path [draw=#4, #8] (g\x\y) edge node []  {} (g\xx\y);
    }
    \path [draw=#4, #8] (g\x5) edge node []  {} (g\xx5);
}
\foreach \x in {5} {
    \foreach \y in {0,2,4} {
        \pgfmathparse{int(1+\y)} \xdef\yy{\pgfmathresult}
        \path [draw=#1, #5] (g\x\y) edge node []  {} (g\x\yy);
    }
    \foreach \y in {1,3} {
        \pgfmathparse{int(1+\y)} \xdef\yy{\pgfmathresult}
        \path [draw=#2, #6] (g\x\y) edge node []  {} (g\x\yy);
    }
}
}
\begin{scope}[scale=0.5,shift={(-15,0)}, every node/.append style={transform shape}]
\addgrid{cbfp1}{cbfp2}{cbfp3}{cbfp4}{very thick}{very thick}{very thick}{very thick}
\end{scope}
\begin{scope}[scale=0.5,shift={(-8,0)}, every node/.append style={transform shape}]
\addgrid{cbfp1}{gray!30}{gray!30}{gray!30}{very thick}{}{}{}
\end{scope}
\begin{scope}[scale=0.5,shift={(-2,0)}, every node/.append style={transform shape}]
\addgrid{gray!30}{cbfp2}{gray!30}{gray!30}{}{very thick}{}{}
\end{scope}
\begin{scope}[scale=0.5,shift={(4,0)}, every node/.append style={transform shape}]
\addgrid{gray!30}{gray!30}{cbfp3}{gray!30}{}{}{very thick}{}
\end{scope}
\begin{scope}[scale=0.5,shift={(10,0)}, every node/.append style={transform shape}]
\addgrid{gray!30}{gray!30}{gray!30}{cbfp4}{}{}{}{very thick}
\end{scope}

\node [] () at (-6.25,-0.75) {\large $G$};
\node [] () at (-4.5,-0.75)  {\large $=$};
\node [] () at (-2.75,-0.75) {\large \textcolor{cbfp1}{$M_1$}};
\node [] () at (-1.25,-0.75) {\large $+$};
\node [] () at (0.25,-0.75)  {\large \textcolor{cbfp2}{$M_2$}};
\node [] () at (1.75,-0.75)  {\large $+$};
\node [] () at (3.25,-0.75)  {\large \textcolor{cbfp3}{$M_3$}};
\node [] () at (4.75,-0.75)  {\large $+$};
\node [] () at (6.25,-0.75)  {\large \textcolor{cbfp4}{$M_4$}};

\node [] () at (-6.25,-1.5) {\large $T_{G,N}$};
\node [] () at (-4.5,-1.5)  {\large $=$};
\node [] () at (-2.75,-1.5) {\large \textcolor{cbfp1}{$T_{M_1,N}$}};
\node [] () at (-1.25,-1.5) {\large $\otimes$};
\node [] () at (0.25,-1.5)  {\large \textcolor{cbfp2}{$T_{M_2,N}$}};
\node [] () at (1.75,-1.5)  {\large $\otimes$};
\node [] () at (3.25,-1.5)  {\large \textcolor{cbfp3}{$T_{M_3,N}$}};
\node [] () at (4.75,-1.5)  {\large $\otimes$};
\node [] () at (6.25,-1.5)  {\large \textcolor{cbfp4}{$T_{M_4,N}$}};

\end{tikzpicture}
\caption{The grid $G$ is the sum of four matchings, so \Cref{lem:decomp} shows that its graph tensor $T_{G,N}$ is the Kronecker product of four graph tensors of matchings.
But since grids are hard by \Cref{lem:per-from-grid} and matchings are easy by \Cref{ex:matching}, we do not expect complexity to be submultiplicative.}
\label{fig:grid-decompose}
\end{figure}
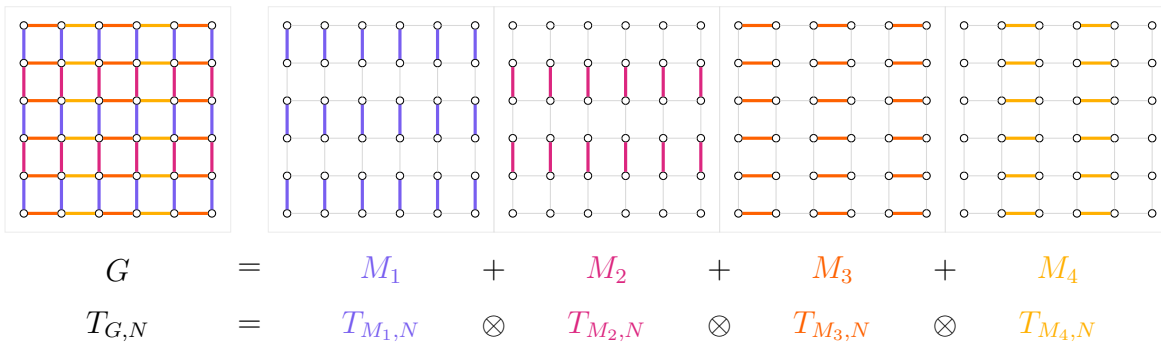

We proceed to reduce matrix permanent tensors to graph tensors of square grid graphs $G$. Since grids have large treewidth, results for related problems~\cite{dwivedi_lower_2026,CurticapeanM14} allow us to expect their graph tensors to have large algebraic complexity.
Indeed, we show in Lemma~\ref{lem:per-from-grid} that the $d\times d$ matrix permanent can be projected from the graph tensors of $(d+2)\times (d+2)$ grids, even for mode dimension $2$.
On the other hand, as shown in \Cref{fig:grid-decompose}, every grid $G$ can be written as the sum 
\begin{equation}
\label{eq:grid-matchings}
G\equ M_1+M_2+M_3+M_4
\end{equation} 
of four edge-disjoint matching graphs $M_1,M_2,M_3,M_4$.
Consequently, for any $N\in\mathbb{N}$ by Lemma~\ref{lem:decomp} we have 

\begin{equation}
\label{eq:grid-kronecker-matchings}
T_{G,N} = T_{M_1,N} \otimes T_{M_2,N} \otimes T_{M_3,N} \otimes T_{M_4,N}.
\end{equation} 
The graph tensors of matchings have low complexity by Example~\ref{ex:matching}.
Thus, if $\CC$ were submultiplicative, then $T_{G,N}$ would also have low complexity.
More generally, this holds for every graph $G$ of small edge-chromatic number, i.e., whenever $G$ can be partitioned into few matchings:
\begin{lemma}[Graphs with small edge-chromatic number under submultiplicativity]
\label{lem:submult-general}
Let $G$ be a graph with maximum edge-multiplicity $b\in\mathbb{N}$ such that $E(G)$ can be partitioned into $t$ matchings. 
Let $N\in\mathbb{N}$.
If $\CC$ is submultiplicative on $|V(G)|$-mode tensors, then
$\CC(T_{G,N}) \leq  2^t |V(G)|^t \cdot N^{tb}$.
\end{lemma}

\begin{proof}
Let $E(G)=M_1 \cup \ldots \cup M_t$.
Example~\ref{ex:matching} and Remark~\ref{rem:flip} give $\CC(M_i) \leq 2|V(G)| \cdot N^b$.
The assumed submultiplicativity of $\CC$, together with \Cref{lem:decomp} gives $\CC(T_{G,N}) \leq \prod_{i=1}^t \CC(T_{M_i,N}) \leq 2^t|V(G)|^t\cdot N^{tb}$, thus proving the claim.
\end{proof}
Using \eqref{eq:grid-kronecker-matchings}, we obtain the following corollary for grids.
\begin{corollary}
\label{cor:submult-grid}
Let $G$ denote the $n\times m$ grid, for $n,m\in \mathbb N$, with maximum edge-multiplicity $b\in\mathbb{N}$. 
If $\CC$ is submultiplicative on $nm$-mode tensors, then $\CC(T_{G,N}) \leq  16 n^4m^4 \cdot N^{4b}$.
\end{corollary}

Our results are obtained by invoking this corollary for various choices of grids $G$ and dimensions~$N$.
We begin under the comparatively weak assumption $\VP \neq \VNP$ \cite{valiant_completeness_1979}.
Towards this end, we first show how the permanent reduces to grid graph tensors.

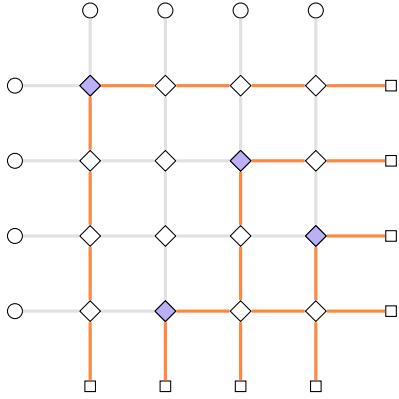
\begin{figure}[t]
\centering
\begin{tikzpicture}
\clip (-3,-3) rectangle (3,3);
\draw [black,opacity=0] (-3,-3) rectangle (3,3);

\begin{scope}[shift={(-2.5,-2.5)}]

\foreach \x in {1,...,4}
\foreach \y in {1,...,4} {
    \node [diamond, draw=black, inner sep=2pt] (g\x\y) at ({\x},{\y}) {};
}

\foreach \x in {1,...,4} {
    \node [rectangle, draw=black, inner sep=2pt] (g\x 0) at ({\x},{0}) {};
    \node [circle, draw=black, inner sep=2pt] (g0\x) at ({0},{\x}) {};
    \node [circle, draw=black, inner sep=2pt] (g\x 5) at ({\x},{5}) {};
    \node [rectangle, draw=black, inner sep=2pt] (g5\x) at ({5},{\x}) {};
}

\def\chippathx#1#2{
\pgfmathparse{int({#1})} \xdef\finalx{\pgfmathresult}
\pgfmathparse{int({#2})} \xdef\y{\pgfmathresult}
\foreach \x in {0,...,4} {
    \pgfmathparse{int(1+\x)} \xdef\xx{\pgfmathresult}
    \pgfmathparse{int(100-(\x<\finalx)*100)} \xdef\bval{\pgfmathresult}
    \path [draw=cbfp3!\bval!verylightgray, opacity=0.75, very thick] (g\x\y) edge node []  {} (g\xx\y);
}
}
\def\chippathy#1#2{
\pgfmathparse{int({#2})} \xdef\finaly{\pgfmathresult}
\pgfmathparse{int({#1})} \xdef\x{\pgfmathresult}
\foreach \y in {0,...,4} {
    \pgfmathparse{int(1+\y)} \xdef\yy{\pgfmathresult}
    \pgfmathparse{int((\y<\finaly)*100)} \xdef\bval{\pgfmathresult}
    \path [draw=cbfp3!\bval!verylightgray, opacity=0.75, very thick] (g\x\y) edge node []  {} (g\x\yy);
}
}
\def\flipvtx#1#2{
\node [diamond, draw=black, inner sep=2pt, fill=cbfp1!50] () at (#1,#2) {};
\chippathx{#1}{#2}
\chippathy{#1}{#2}
}

\flipvtx{1}{4}
\flipvtx{2}{1}
\flipvtx{3}{3}
\flipvtx{4}{2}

\end{scope}

\end{tikzpicture}
\caption{The reduction from the permanent to the graph tensors of grids, as shown in \Cref{lem:per-from-grid}. The  
\textcolor{gray}{\textbf{gray}} and \textcolor{cbfp3}{\textbf{orange}} edges are assigned $0$ and $1$, respectively.
Each horizontal and each vertical path flips from $0$ to $1$ exactly once; whenever this happens at some vertex $v_{i,j}$ (a \emph{flip vertex}, shown \textcolor{cbfp1}{\textbf{purple}}), we ensure that both the $i$-th horizontal and the $j$-th vertical path flip at $v_{i,j}$. In other words, flip vertices always lie on orange corners.
Each flip vertex $v_{i,j}$ contributes $x_{i,j}$ to the total weight of the assignment, otherwise $1$.
The weight of the assignment shown here is $x_{11}x_{23}x_{34}x_{42}$, which is indeed of the form $\prod_i x_{i,\pi(i)}$ for $\pi =1342$.}
\label{fig:per-from-grid}
\end{figure}

\begin{lemma}[The permanent tensor reduces to a grid-graph tensor]
\label{lem:per-from-grid}
    The $n\times n$ permanent tensor $\mathrm{per}_n$ is a projection of $T_{G,2}$ for the simple $(n+2) \times (n+2)$ grid graph $G$.
\end{lemma}

On a high level, the reduction in \Cref{lem:per-from-grid} proceeds as illustrated in Figure~\ref{fig:per-from-grid}:
A graph $G$ is obtained by attaching pendant vertices of degree $1$ to the boundary of an $n \times n$ grid; the pendant vertices are needed to ensure boundary conditions. 
Each edge of the resulting graph carries a Boolean state, and we substitute values into the indeterminates of $T_{G,2}$ such that the resulting polynomial counts assignments $a$ to the edges of the grid with the following properties: 
\begin{enumerate}
    \item We ensure that each horizontal and each vertical path $P$ starts with a $0$-edge and the edge states flip from $0$ to $1$ at exactly one vertex; we call this the \emph{flip vertex} of $P$ under $a$.
    \item For all $i,j\in [n]$, we ensure that a vertex $v_{i,j}$ is the flip vertex of the $i$-th horizontal path if and only if the same vertex $v_{i,j}$ also is the flip vertex of the $j$-th vertical path.
    \item Finally, we ensure that a flip vertex contributes a factor $x_{i,j}$ to the weight of the assignment, while non-flip vertices do not contribute (i.e., they contribute $1$).
\end{enumerate}
Under these conditions, the flip vertices under an assignment $a$ induce an $n\times n$ permutation matrix $P_\pi$, and the weight of $a$ is the product of the indeterminates $x_{i,\pi(i)}$ selected by $P_\pi$. Summing over all valid assignments $a$ gives precisely the permanent. 

\begin{proof}[Proof of Lemma~\ref{lem:per-from-grid}]
The proof implements the idea sketched above, but viewing the pendant vertices as the degree-$3$ border vertices of the $(n+2)\times (n+2)$ grid.
We label the vertices of $G$ as $v_{i,j}$ for $0\leq i,j\leq n+1$. For $v_{i,j}$ with $1\leq i,j\leq n$, define the signature $f_{i,j}:\{0,1\}^4 \to \{0,1\}$ as follows, where the inputs are read clockwise around~$v_{i,j}$, starting from the top edge, as $t,r,b,\ell$: 
\[
f_{i,j}(t,r,b,\ell ) =
\left\{
\begin{array}{lll}
1 & t=b \text{ and } \ell=r & \text{(no flip)}\,,\\
x_{ij}  & t=\ell=0 \text{ and } r=b=1 & \text{(flip)}\,,\\
0 & \text{otherwise} & \text{(invalid)}\,.
\end{array}
\right.
\]
\Cref{fig:perm-hol-sig} depicts all the valid assignments of any non-border signature. 
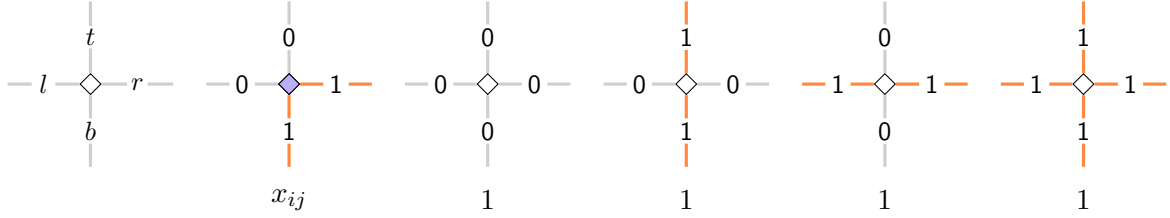
\begin{figure*}[t]
\centering
\begin{tikzpicture}
\clip (-8,-1.8) rectangle (8,1.2);
\draw [black,opacity=0] (-8,-1.8) rectangle (8,1.2);

\begin{scope}[scale=1.1,shift={(1.2,0)}]
\def\danglevtx#1#2#3#4#5#6#7{
\node [diamond, draw=black, inner sep=2pt] (t#3#4#5#6) at ({#1},{#2}) {};
\pgfmathparse{int(#3*100)} \xdef\bval{\pgfmathresult}
\path [draw=cbfp3!\bval!lightgray, opacity=0.75, very thick] 
(t#3#4#5#6) edge node [font=\sffamily\small, midway, inner sep=0.5pt, minimum size=0.2pt, circle, align=center, fill=white,opacity=1] {#3} +(-1,0);
\pgfmathparse{int(#4*100)} \xdef\bval{\pgfmathresult}
\path [draw=cbfp3!\bval!lightgray, opacity=0.75, very thick] 
(t#3#4#5#6) edge node [font=\sffamily\small, midway, inner sep=0.5pt, minimum size=0.2pt, circle, align=center, fill=white,opacity=1] {#4} +(0,1);
\pgfmathparse{int(#5*100)} \xdef\bval{\pgfmathresult}
\path [draw=cbfp3!\bval!lightgray, opacity=0.75, very thick] 
(t#3#4#5#6) edge node [font=\sffamily\small, midway, inner sep=0.5pt, minimum size=0.2pt, circle, align=center, fill=white,opacity=1] {#5} +(1,0);
\pgfmathparse{int(#6*100)} \xdef\bval{\pgfmathresult}
\path [draw=cbfp3!\bval!lightgray, opacity=0.75, very thick] 
(t#3#4#5#6) edge node [font=\sffamily\small, midway, inner sep=0.5pt, minimum size=0.2pt, circle, align=center, fill=white,opacity=1] {#6} +(0,-1);
\node [] () at ({#1},{#2-1.4}) {#7};
}

\node [diamond, draw=black, inner sep=2pt] (tex) at (-7.2,0) {};
\path [draw=lightgray, opacity=0.75, very thick, every node/.style={font=\sffamily\small, midway, inner sep=0.5pt, minimum size=0.2pt, circle, align=center, fill=white, opacity=1}] 
(tex) edge node [] {$l$} +(-1,0)
edge node [] {$t$} +(0,1)
edge node [] {$r$} +(1,0)
edge node [] {$b$} +(0,-1);

\danglevtx{-4.8}{0}{0}{0}{1}{1}{$x_{ij}$}
\node [diamond, draw=black, inner sep=2pt, fill=cbfp1!50] () at (-4.8,0) {};
\danglevtx{-2.4}{0}{0}{0}{0}{0}{$1$}
\danglevtx{0}{0}{0}{1}{0}{1}{$1$}
\danglevtx{2.4}{0}{1}{0}{1}{0}{$1$}
\danglevtx{4.8}{0}{1}{1}{1}{1}{$1$}
\end{scope}

\end{tikzpicture}
\caption{A non-border signature. 
Below an assignment is the value of the signature. 
Assignments that are not listed here evaluate to zero.}
\label{fig:perm-hol-sig}
\end{figure*}

The top border vertices $v_{0,j}$ are given signatures to ensure that the top-most edge of vertical paths is assigned $0$:  
\[
f_{0,j}(r,b,\ell ) =
\left\{
\begin{array}{llc}
1 & \text{if }r=b=\ell =0\,,\\
0  & \text{otherwise\,.}
\end{array}
\right.
\]
The bottom border vertices $v_{n+1,j}$ are given signatures to ensure that the bottom-most edge of vertical paths is assigned $1$: 
\[
f_{n+1,j}(t,r,\ell ) =
\left\{
\begin{array}{llc}
1 & \text{if }t=r=\ell =1\,,\\
0  & \text{otherwise\,.}
\end{array}
\right.
\]
Analogous signatures are defined for the left and right borders. The four corners of $G$ receive constant all-ones signatures.

Our choice of signatures ensures that, in an assignment $a \in \{0,1\}^{E(G)}$ with $w(a)\coloneqq \prod_{i,j}f_{i,j}(a) \neq 0$, the indices $(i,j)$ with $f_{i,j}(a)$ in state ``flip'' form a permutation matrix corresponding to some permutation $\pi$, and that $w(a)=\prod_{i} x_{i,\pi(i)}$. It follows that $\mathrm{per}(A)=\sum_\pi \prod _i x_{i,\pi(i)}\pro T_{G,2}$.
\end{proof}

\submultvpvnp*
\begin{proof}
Under the assumption of the theorem, \cref{cor:submult-grid} and~\cref{lem:per-from-grid} would imply $\CC(T_{G,2}) \leq \mathrm{poly}(n)$, contradicting the $\VNP$-hardness of the permanent.
\end{proof}

If we assume that permanents require exponential-size circuits, then we can obtain stronger bounds. Towards this, we first contract large grids into smaller grids with larger edge-multiplicities.
(Larger edge-multiplicities essentially correspond to larger mode dimensions: For every graph $H$ and $b,k,n \in \mathbb N$ with $k$ dividing $n$, and writing $B=b^{n/k}$, we have $T_{H,B} \equ T_{n/k\cdot H,b}$.)
Namely, first partition the $n\times n$ grid into $k\times k$ pieces of size $n/k\times n/k$ each, then contract each piece to a single vertex by \cref{lem:contraction}.
The complexity bound in \cref{lem:contraction} requires a little care in the choice of contraction sequence to obtain the desired lower bounds later.
\begin{lemma}[Grid graph tensor contraction by subgrids]
\label{lem:contract-grid}
Let $k$ divide $n$ and let $H$ and $G$ be the $k\times k$ and $n\times n$ grid graphs, respectively. Let $b\in\mathbb{N}$ and $B=b^{n/k}$.
Then, $\CC(T_{G,b}) \leq (2b^4+b)n^2\cdot b^{3(n/k)} + \CC(T_{H,B})$.
\end{lemma}

\begin{proof}
Partition the vertex set of $G$ into $k\times k$ square blocks $A_{i,j}$, each of size $n/k \times n/k$.
For each $i,j \in [k]$, we then perform the following:
\begin{enumerate}
    \item For $t = 1,\ldots,n/k$, contract the vertices in row $t$ of block $A_{i,j}$ by \cref{lem:contraction}.
    Each contraction involves $3n/k+1$ edges and thus incurs an additive cost term of $bn/k\cdot b^{3(n/k)}$ in \cref{lem:contraction}. Block $A_{i,j}$ now consists of a vertical path of length $n/k$, with edge-multiplicity $n/k$ between consecutive vertices, and with two additional edges per vertex.
    The total cost of the contractions in this step is $b(n/k)^2\cdot b^{3(n/k)}$.
    \item For $t = 1,\ldots,n/k-1$, we contract vertices $t$ and $t+1$ on the path remaining in block $A_{i,j}$. This involves $3n/k+4$ edges per contraction and incurs a cost term of $2b^4\cdot b^{3(n/k)}$ in \cref{lem:contraction}.
    Block $A_{i,j}$ now consists of a single vertex incident with all edges leaving $A_{i,j}$.
    The total cost of the contractions in this step is $2b^4(n/k)\cdot b^{3(n/k)}$.
    \end{enumerate}
In total over the $k^2$ cells, we incur cost at most $(2b^4+b)k^2(n/k)^2\cdot b^{3(n/k)}=(2b^4+b)n^2\cdot b^{3(n/k)}$ for all contractions and obtain the $k\times k$ grid with edge-multiplicity $n/k$ from $G$.
\end{proof}

A lower bound now follows easily.
\submultper*

\begin{proof}
Choose $k=\lceil 4/c \rceil$ so $4/k\leq c$, and let $H$ be the $k\times k$ grid on $k^2$ vertices.
By Remark~\ref{rem:flip} and Lemma~\ref{lem:contract-grid} with $b=2$, if the $k^2$-mode tensor $T_{n/k\cdot H,2}$ admits a circuit of size $s$, then the $n^2$-mode tensor $T_{G,2}$ for the $n \times n$ grid $G$ admits a circuit of size $s' \leq n^{O(1)} \cdot 8^{n/k} + s$, and Lemma~\ref{lem:per-from-grid} gives a circuit of size $s' \cdot n^{O(1)}$ for $\mathrm{per}_n$.

If submultiplicativity indeed held, then \cref{cor:submult-grid} would give a circuit of size $s \leq n^{O(1)} \cdot 2^{4n/k}$ for $T_{n/k\cdot H,2}$, so we obtain a circuit of size $s' \leq n^{O(1)} \cdot 2^{4n/k} \leq O(2^{cn})$ for $T_{G,2}$ and the permanent.
\end{proof}

The above strategy for conditional arguments could be refined in different ways to rule out submultiplicativity on fewer modes, e.g., by choosing graph tensors other than grids.

\subsubsection*{Reductions from Hyperclique Tensors}

For $3\leq h<k$, the \emph{$(h,k)$-hyperclique conjecture} states that for any $\varepsilon>0$, there is no $O(n^{k-\varepsilon})$-time algorithm that decides whether an $h$-uniform hypergraph contains a $k$-hyperclique. 
More explicitly, for $N \in \mathbb N$, define the following tensor on $\binom{k}{h}$ modes $X^{(S)}$ for $S \in \binom{[k]}{h}$ via
\[
H^N_{h,k} = \sum_{f\colon [k] \rightarrow [N]} \prod_{S \in \binom{[k]}{h}} x^{(S)}_{f|_S}.
\]
Under adequate uniformity assumptions, a circuit family $(C_N)_N$ for the sequence of tensors $(H^N_{h,k})_N$ of size $N^{k-\varepsilon}$ would give an algorithm for the problem that falsifies this conjecture. Even absent such assumptions, we can formulate the hypothesis that there is no such circuit family whatsoever, which we will refer to as the \emph{non-uniform algebraic $(h,k)$-hyperclique conjecture.}

\submulthc*

\begin{proof}
We illustrate the case $d = 8.$
The corresponding $(3,4)$-hyperclique conjecture rules out algorithms running in time $N^{4-\varepsilon}$ for detecting $4$-vertex hypercliques in $3$-uniform hypergraphs with $N$ vertices.
In this case, we can write
\[
H^N_{3,4} \equ \sum_{i,j,k,\ell = 1}^N x_{ijk} y_{ij\ell} z_{ik\ell} w_{jk\ell}.
\]

Consider the bipartite incidence (non-hyper-)graph $I = I_{3,4}$ associated with the $(3,4)$-hyperclique, which is $K_{4,4}$ minus a perfect matching, as depicted in Fig.~\ref{fig:hc34-incidence}.
The associated graph tensor $T_{I,N}$ can be written down using twelve indices, one for each edge in $I$, and has $8$ modes.
For the sake of legibility, we collect the indices into four groups $i,j,k,\ell \in [n]^3$, with three entries each, corresponding to the neighborhoods of the hyperedges in the bottom part of the graph.
We recall that the superscript indices indicate the vertex, and the subscript the (images of) incident edges.
\[
T_{I,N} \equ \sum_{i,j,k,\ell \in [n]^3} x^{(1)}_{i_1,i_2,i_3} \cdot x^{(2)}_{j_1,j_2,j_3} \cdot x^{(3)}_{k_1,k_2,k_3} \cdot  x^{(4)}_{\ell_1,\ell_2,\ell_3}\cdot  x^{(123)}_{i_1,j_1,k_1} \cdot x^{(124)}_{i_2,j_2,\ell_1} \cdot x^{(134)}_{i_3,k_2,\ell_2} \cdot x^{(234)}_{j_3,k_3,\ell_3}. 
\]
Observe now that $H^N_{3,4}$ is a linear projection of $T_{I,N}$ under the substitution
\begin{align*}
   & x^{(v)}_{\alpha\beta\gamma} \mapsto \begin{cases} 
        1,& \text{ if $\alpha=\beta=\gamma$} \\
        0, & \text{ otherwise}
    \end{cases} 
    \textnormal{ for $v \in \{1,2,3,4\}, \alpha,\beta,\gamma \in [n]$, and } \\
  &  x^{(123)}_{ijk} \mapsto x_{ijk},\ x^{(124)}_{ijk} \mapsto y_{ijk},\ x^{(134)}_{ik\ell} \mapsto z_{ik\ell},\ x^{(234)}_{jk\ell} \mapsto w_{jk\ell}\ \textnormal{for $i,j,k,\ell \in [n].$}
\end{align*}
Therefore, $\CC(H^N_{3,4}) \leq \CC(T_{I,N}).$

On the other hand, we can decompose $I$ as a union of three perfect matchings with four edges each, see again \cref{fig:hc34-incidence}. Via~\cref{lem:submult-general}, submultiplicativity of $\CC$ would imply a circuit of size $O(N^3)$ for $T_{I,N}$, and hence for $H^N_{3,4}$, by \cref{lem:decomp}.
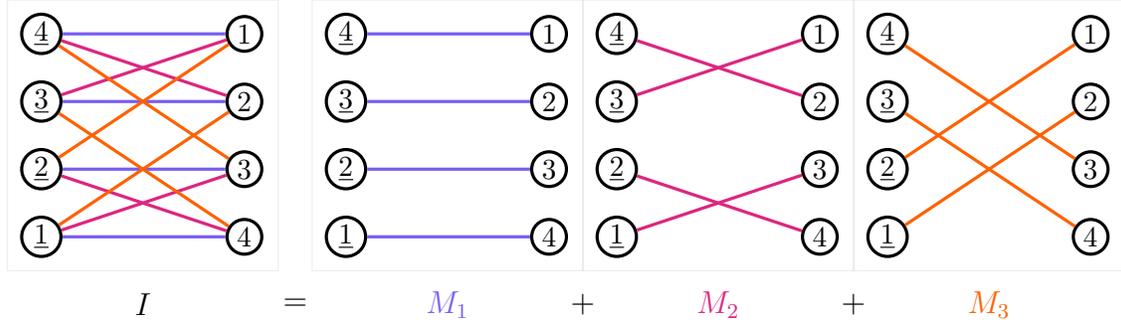
\begin{figure}[t]
\centering

\begin{tikzpicture}[main node/.style={circle,draw,very thick,inner sep=1.5pt},scale=0.9]
\clip (-8.55,-2.8) rectangle (8.05,2.15);
\draw[black,opacity=0] (-8.55,-2.8) rectangle (8.05,2.15);

\def\recsize{1.5}
\def\bdsize{0.5}

\node at (-6.5,-2.50)  {\large $I$};
\node[cbfp1] at (-2.0,-2.50)  {\large $M_1$};
\node[cbfp2] at ( 2.0,-2.50)  {\large $M_2$};
\node[cbfp3] at ( 6.0,-2.50)  {\large $M_3$};
\node at (-4.25,-2.50)  {\large $=$};
\node at ( 0.0,-2.50)  {\large $+$};
\node at ( 4.0,-2.50)  {\large $+$};

\begin{scope}[shift={(-6.5,0)}]
  \node[main node] (a4) at (-\recsize, \recsize) {$\underline{4}$};
  \node[main node] (a3) at (-\recsize,  0.5)    {$\underline{3}$};
  \node[main node] (a2) at (-\recsize, -0.5)    {$\underline{2}$};
  \node[main node] (a1) at (-\recsize,-\recsize){$\underline{1}$};

  \node[main node] (b1) at ( \recsize, \recsize) {1};
  \node[main node] (b2) at ( \recsize,  0.5)    {2};
  \node[main node] (b3) at ( \recsize, -0.5)    {3};
  \node[main node] (b4) at ( \recsize,-\recsize){4};

  \path[very thick,cbfp1]
    (a4) edge (b1)
    (a3) edge (b2)
    (a2) edge (b3)
    (a1) edge (b4);

  \path[very thick,cbfp2]
    (a4) edge (b2)
    (a3) edge (b1)
    (a2) edge (b4)
    (a1) edge (b3);

  \path[very thick,cbfp3]
    (a4) edge (b3)
    (a3) edge (b4)
    (a2) edge (b1)
    (a1) edge (b2);

  \draw [gray!30,opacity=0.5]
    ({-\bdsize-\recsize},{-\bdsize-\recsize})
    rectangle
    ({\bdsize+\recsize},{\bdsize+\recsize});
\end{scope}
\begin{scope}[shift={(-2,0)}]
  \node[main node] (a4) at (-\recsize, \recsize) {$\underline{4}$};
  \node[main node] (a3) at (-\recsize,  0.5)    {$\underline{3}$};
  \node[main node] (a2) at (-\recsize, -0.5)    {$\underline{2}$};
  \node[main node] (a1) at (-\recsize,-\recsize){$\underline{1}$};

  \node[main node] (b1) at ( \recsize, \recsize) {1};
  \node[main node] (b2) at ( \recsize,  0.5)    {2};
  \node[main node] (b3) at ( \recsize, -0.5)    {3};
  \node[main node] (b4) at ( \recsize,-\recsize){4};

  \path[very thick,cbfp1]
    (a4) edge (b1)
    (a3) edge (b2)
    (a2) edge (b3)
    (a1) edge (b4);

  \draw [gray!30,opacity=0.5]
    ({-\bdsize-\recsize},{-\bdsize-\recsize})
    rectangle
    ({\bdsize+\recsize},{\bdsize+\recsize});
\end{scope}

\begin{scope}[shift={(2,0)}]
  \node[main node] (a4) at (-\recsize, \recsize) {$\underline{4}$};
  \node[main node] (a3) at (-\recsize,  0.5)    {$\underline{3}$};
  \node[main node] (a2) at (-\recsize, -0.5)    {$\underline{2}$};
  \node[main node] (a1) at (-\recsize,-\recsize){$\underline{1}$};

  \node[main node] (b1) at ( \recsize, \recsize) {1};
  \node[main node] (b2) at ( \recsize,  0.5)    {2};
  \node[main node] (b3) at ( \recsize, -0.5)    {3};
  \node[main node] (b4) at ( \recsize,-\recsize){4};

  \path[very thick,cbfp2]
    (a4) edge (b2)
    (a3) edge (b1)
    (a2) edge (b4)
    (a1) edge (b3);

  \draw [gray!30,opacity=0.5]
    ({-\bdsize-\recsize},{-\bdsize-\recsize})
    rectangle
    ({\bdsize+\recsize},{\bdsize+\recsize});
\end{scope}

\begin{scope}[shift={(6,0)}]
  \node[main node] (a4) at (-\recsize, \recsize) {$\underline{4}$};
  \node[main node] (a3) at (-\recsize,  0.5)    {$\underline{3}$};
  \node[main node] (a2) at (-\recsize, -0.5)    {$\underline{2}$};
  \node[main node] (a1) at (-\recsize,-\recsize){$\underline{1}$};

  \node[main node] (b1) at ( \recsize, \recsize) {1};
  \node[main node] (b2) at ( \recsize,  0.5)    {2};
  \node[main node] (b3) at ( \recsize, -0.5)    {3};
  \node[main node] (b4) at ( \recsize,-\recsize){4};

  \path[very thick,cbfp3]
    (a4) edge (b3)
    (a3) edge (b4)
    (a2) edge (b1)
    (a1) edge (b2);

  \draw [gray!30,opacity=0.5]
    ({-\bdsize-\recsize},{-\bdsize-\recsize})
    rectangle
    ({\bdsize+\recsize},{\bdsize+\recsize});
\end{scope}
\end{tikzpicture}
\caption{Decomposition of the incidence graph $I$ of the $(3,4)$-hyperclique into matchings ${\color{cbfp1} M_1},{\color{cbfp2} M_2},{\color{cbfp3} M_3}$, where $\underline{1},\underline{2},\underline{3},\underline{4}$ indicate the hyperedges $234,134,124 ,123$, respectively.}
\label{fig:hc34-incidence}
\end{figure}
\end{proof}

\begin{remark}
In the above proof, it seems cumbersome to transition to the incidence graph, which in particular doubles the degree of the resulting tensors.
A more direct approach could involve the hypergraph-tensor of the $(3,4)$-hyperclique itself, which only has four modes.
In this case, however, the only possible decompositions into sub-hypergraphs are either into (a) a single hyperedge and a hypergraph with three hyperedges, or into (b) two pairs of hyperedges. 
Both decompositions would yield circuits of size at least $\Omega(n^4),$ even assuming submultiplicativity, simply because of the numbers of variables involved.
Intuitively, passing to the incidence graph allows us to trade higher overall degree (number of modes) for a smaller total number of variables (dimension per mode), which is beneficial when ``Kroneckering up''.
\end{remark}

\subsection{Upper Bounds via Line Graph Treewidth}
\label{sec:ltw-dp}

The {\em line-treewidth} $\ltw G$ is the treewidth $\tw {L(G)}$ of the \emph{line graph} $L(G)$ of $G$, where the vertex set of $L(G)$ is $E(G)$ and two vertices corresponding to edges $e,e'$ are adjacent in $L(G)$ if they share a vertex in $G$.
We collect a few facts about the treewidth of a line graph. We recall that $\Delta(H)$ is the maximal degree of any vertex in $H$.

\begin{fact}[Treewidth versus line-treewidth \protect{\cite[Equations (2) and Proposition 2.3]{harvey_treewidth_2018}}]\label{fact:ltw}
    For every graph~$H$, we have
    \[
        \tw{H} - 1 \leq \ltw{H} \leq (\tw{H} +1) \Delta(H) - 1.
    \]
\end{fact}

Of special interest is the treewidth of the line graph of the complete graph, which has been determined exactly.

\begin{theorem}[Line-treewidth of a complete graph \protect{\cite[Theorem 1]{harvey_treewidth_2015}}]\label{thm:line-clique-tw}
    For $d \geq 1$, we have 
    \begin{equation*}
    \ltw{K_d} = \begin{cases}
\left(\frac{d-1}{2}\right)^{2}+d-2, & \text{if } d\text{ is odd},\\
\left(\frac{d-2}{2}\right)\left(\frac{d}{2}\right)+d-2, &\text{if } d\text{ is even}.
\end{cases}
\end{equation*}
\end{theorem}

Similar to standard dynamic programming arguments on tree-decompositions in, e.g., \cite[Section~5.1]{yin_approximate_2013}, \cite[Lemma~5.23]{dawar_symmetric_2025}, and \cite[Example~3.4]{flum_parameterized_2004}, we can show:

\begin{lemma}\label{lem:circuit-from-ltw}
    For every $H$ and $n \in \mathbb{N}$, we have $\CC(T_{H, n}) \leq |V(H)| \cdot n^{\ltw H + 1}$.
\end{lemma}
\begin{proof}
    Write $k \coloneqq \ltw{H}$.
    Let $\beta \colon V(T) \to 2^{E(H)}$ be a tree decomposition for the line graph $L(H)$ of width $\leq k$.
    We contract $H$ in a bottom-up fashion along the tree decomposition.
    
    Since the set of edges $I(v)$ incident to a vertex $v \in V(H)$
    forms a clique in $L(H)$,
    there exists a node $t_v \in V(T)$ such that $I(v)$ is fully contained in $\beta(t_v)$; there may be multiple such nodes $t_v$, from which we choose one arbitrarily.
    Define a map $\chi \colon V(T) \to 2^{V(H)}$ by $\chi(t) \coloneqq \{v \in V(H) \mid t_v = t\}$.
    By definition, every $t \in V(T)$ satisfies
    $\bigcup_{v \in \chi(t)} I(v) \subseteq \beta(t)$.
    In particular,
    the number of edges incident to some vertex in $\chi(t)$ is
    \begin{equation}\label{eq:tensor-contraction}
        \left\lvert \bigcup_{v \in \chi(t)} I(v) \right\rvert \leq k+1.
    \end{equation}
    
    Let $\ell \in V(T)$ be a leaf of $T$ and write $t \in V(T)$ for its parent.
    Note that every tree with at least two vertices contains a leaf.
    We contract the vertices in $\chi(\ell)$ to a fresh vertex $v$.
    By \eqref{eq:tensor-contraction},
    the cost of this contraction is bounded by $n^{k+1}$.
    The edges incident with the new vertex $v$ are in $\beta(\ell) \cap \beta(t)$.
    We add the fresh vertex $v$ to $\chi(t)$.
    This maintains the property \eqref{eq:tensor-contraction} of $\gamma(t)$ having at most $k+1$ incident edges.
    We may now delete $\ell$ from $T$ and proceed with the next leaf. When the process terminates, a single vertex is left.
\end{proof}

As an immediate consequence the above lemma proves
\begin{align} \label{eq:ltw-ac}
\AC(T_{H,n}) \leq n^{\ltw H+1}
\end{align} 
for every graph $H$.
For general tensors, recall that we write $\eta(T)$ for the exponent of the asymptotic circuit complexity $\AC(T) = n^{\eta(T)}$ of a tensor $T \in (\mathbb F^n)^{\otimes d}$. We then have the following more precise version of \cref{thm:intro-cc-ltw}:
\begin{theorem}[Upper bound on the asymptotic circuit complexity of a $d$-tensor]\label{thm:clique-treewidth}
    Let $T \in (\mathbb F^n)^{\otimes d}$ be a $d$-mode tensor of dimension $n$.
    Then, we have
    \[
    \eta(T) \leq \frac2d (\ltw{K_d} + 1) = d/2 + 1 - \frac{7+(-1)^d}{4d}\,.
    \]
\end{theorem}
\begin{proof}
    Similar to the proof of \cref{lem:strassen-generalized} and using \cref{lem:decomp} as well as \cref{rem:flip}, 
    we have that 
    \[
    T^{\otimes dm} \equ (T^{\otimes d})^{\otimes m} \res T_{2\cdot K_d,n^{m}} \equ T_{K_d,n^{2m}}\text{ for all $m$.}
    \]
    By \cref{fact:prod-res} and more specifically \cref{lem:graph-tensor-length}, we have that 
    \[
    \CC(T^{\otimes k}) \leq \CC(T^{\otimes d \cdot \lceil k/d \rceil})\ \text{for all $k$,}
    \]
    whereas \cref{lem:circuit-from-ltw} shows 
    \[
    \CC(T^{\otimes {dm}}) \leq d^{O(1)} \cdot n^{2m\cdot (\ltw{K_d} + 1)}\ \text{for all $m$}.
    \]
    In combination, for $m = \lceil k/d\rceil,$ we find that 
    \[
    \CC(T^{\otimes k}) \leq d^{O(1)} \cdot n^{2\lceil k/d \rceil (\ltw{K_d} +1)}
    \leq d^{O(1)} \cdot n^{2(k/d+1)(\ltw{K_d}+1)}.
    \]
    By definition of $\eta(T)$, this shows that $\eta(T) \leq 2/d\cdot(\ltw{K_d} + 1)$.
    The final bound then follows from~\cref{thm:line-clique-tw}.
\end{proof}

\begin{remark}
    It is natural to ask whether the exponent from \cref{thm:clique-treewidth} is asymptotically tight under standard assumptions from complexity theory. 
    To the best of our knowledge, current techniques would only rule out exponents of the form $o(d /\log d)$.
 
    For instance, known results on the complexity of subgraph isomorphism~\cite{marx_can_2010,KarthikMPS2024,curticapean_et_al:LIPIcs.STACS.2025.28} imply $\eta(T_{H,n})\in\Omega(d/\log d)$ for a sequence of $d$-vertex graphs $H$ of maximum degree $3$.
    It is also possible to adapt a recent bound by Pratt~\cite{pratt_stronger_2024} to a higher number of modes, which would yield a lower bound under the set cover conjecture. 
    Since all of the resulting lower bounds are fairly loose, we refrain from formally stating and proving them here.
\end{remark}

\section{Exploratory Upper Bounds in Intermediate Orders}
In this final section, we show upper bounds on the asymptotic circuit complexity of generic $d$-mode tensors for $d=4$ and $d=5$.
So far, we have presented various techniques to obtain upper bounds on both the asymptotic circuit complexity and rank as $d$ grows.
We now combine some of these techniques with computational search to find better upper bounds for $4$-mode and $5$-mode tensors.

\subsection{Restricted Submultiplicativity}
The results in Section~\ref{subsec:impossible} strongly suggest the absence of circuit complexity submultiplicativity for $d$-mode tensors with $d \geq 4$: For $d$-mode tensors $T$ and $U$ of low complexity, no useful upper bounds for the circuit complexity $\CC(T\otimes U)$ as a function of $\CC(T)$ and $\CC(U)$ are known.
However, we can construct small circuits for $T\otimes U$ from small circuits for $T$ and $U$, provided that one of the circuits is sufficiently simple.

To this end, we will need the following property of tensor contraction.
Let $\ell_1,\ldots,\ell_d\colon \mathbb F^m \rightarrow \mathbb F$ and $h_1,\ldots,h_d\colon \mathbb F^n \rightarrow \mathbb F$ be linear forms. 
We consider the Kronecker product of the corresponding rank-one tensors $L = \prod_{i=1}^d \ell_i(x^{(i)})$ and $H = \prod_{i=1}^d h_i(y^{(i)})$,
where the $i$-th mode of $L \otimes H$ has coordinates $z^{(i)}_{jk} \coloneqq x^{(i)}_j \otimes y^{(i)}_k$, with $j \in [m], k \in [n].$
For brevity, we write $z^{(i)}$ for the $m \times n$ matrix with entries $z^{(i)}_{jk}$, and identify $h_i$ with its vector of coefficients in $\mathbb F^n.$
Then, one can check that
\[
L \otimes H = \prod_{i=1}^d \ell_i(h_i(z^{(i)}_{11},\ldots,z^{(i)}_{1n}),\ldots,h_i(z^{(i)}_{m1}, \ldots, z^{(i)}_{mn}))) = L(z^{(1)} \cdot h_1, \ldots,z^{(d)} \cdot h_d).
\]
By linearity and the fact that every tensor $U$ has a decomposition into rank-one tensors, we obtain that
\begin{align} \label{eq:rank-one-contraction}
U \otimes H = U(z^{(1)} \cdot h_1,\ldots,z^{(d)} \cdot h_d)
\end{align}
holds for all tensors $U$ with $d$ modes in coordinates $z^{(i)}_{jk}$.
This can be applied in order to obtain the following result.
\begin{lemma} \label{lem:rank_and_circuit}
    Let $T \in (\mathbb F^n)^{\otimes d}$ and $U \in \mathcal (\mathbb F^m)^{\otimes d}$ with $T$ of rank at most $r$, with $T = \sum_{i=1}^r \prod_{j=1}^d \ell_{ij}(x^{(j)}).$
    For $1 \leq j \leq d$, let $\Lambda^{(j)}$ be the $n \times r$ matrix of linear forms $(\ell_{ij})_i$ appearing in the $j$-th mode of the rank-decomposition of $T$, and we recall that $z^{(i)}$ is an $m \times n$ matrix of indeterminates $z^{(i)}_{j,k}$ for $1 \leq i \leq d.$
    Then,
    \[
    \CC(U\otimes T) \leq r(\CC(U) + 1) + \sum_{i=1}^d \CC(z^{(i)} \cdot \Lambda^{(i)}) + r \cdot m \cdot d.
    \]
\end{lemma}
\begin{proof}
Distributing~\eqref{eq:rank-one-contraction} linearly across the terms of the assumed rank-$r$ decomposition of $T$, we see that
\begin{align} \label{eq:rank-circuit}
U \otimes T = \sum_{i=1}^r \left( U \otimes  \prod_{j=1}^d \ell_{ij}(x^{(j)}) \right) = \sum_{i=1}^r  U(z^{(1)} \cdot \ell_{i,1}, \ldots, z^{(d)} \cdot \ell_{i,d}).
\end{align}
On the right-hand side of~\eqref{eq:rank-circuit}, observe that the $j$-th mode of the instance of $U$ appearing in the $i$-th term in the sum obtains as an input precisely the $i$-th row of $z^{(j)} \cdot \Lambda^{(j)}$, which is a column vector of dimension $m$. To implement~\eqref{eq:rank-circuit} as a circuit, we place a fan-in $r$ $+$-gate at the top, contributing $r$ wires. The $i$-th input of this topmost $+$-gate is given by a copy of the circuit for $U$, contributing $r \cdot \CC(U)$ to the circuit size overall. In addition, we need to compute the inputs to the individual $r$ copies of $U$ (which might sound negligible, but becomes crucial in the asymptotic regime.) As mentioned, for this it suffices to compute the $d$ matrix products $z^{(1)} \cdot \Lambda^{(1)}, \ldots, z^{(d)}\cdot \Lambda^{(d)}$.
Finally, we need to wire each of the computed entries of $z^{(j)}\cdot \Lambda^{(j)}$ to the respective input of $U$, amounting to a total of $r \cdot m\cdot d$ wires. This yields the claimed bound.
\end{proof}

To become useful in our applications, we need an asymptotic variant of \cref{lem:rank_and_circuit}, that is, for higher Kronecker powers of $T$ and $U$.
As mentioned in the proof, here the complexity bound on the computation of the matrix-vector products that comprise the inputs of $U$ becomes relevant.
\rankcirc*
\begin{proof}
We need to show that $\AC(U \otimes T) \leq \AR(T) \cdot \AC(U)$ holds.
Let $T \in (\mathbb F^n)^{\otimes d}$ and $U \in (\mathbb F^m)^{\otimes d}$.
For $T \in (\mathbb F^n)^{\otimes d}$, a rank-$r$ decomposition of $T = \sum_{i=1}^r \prod_{j=1}^d \ell_{ij}(x^{(j)})$ leads to a rank-$r^k$ decomposition
\[
T^{\otimes k} = \sum_{i_1,\ldots,i_k = 1}^r \prod_{j=1}^d  \left( \bigotimes_{\kappa = 1}^k \ell_{\kappa j}(x^{(j)}) \right)
\]
of $T^{\otimes k} \in (\mathbb F^{n^k})^{\otimes d}$, where the rank-1 term corresponding to indices $i_1,\ldots,i_k$ has as its $j$-th mode the Kronecker product $\ell_{i_1 j} \otimes \cdots \otimes \ell_{i_k,j} \in \mathbb F^{n^k}.$
Hence, the matrices constituting the $j$-th mode of the rank decomposition become $n^k \times r^k$-matrices given by $(\Lambda^{(j)})^{\otimes k}.$
Now, when applying \cref{lem:rank_and_circuit}, the naive circuit for the product $z^{(j)} \cdot (\Lambda^{(j)})^{\otimes k}$ would simply spell out each of the $m^k \cdot r^k$ inner products of dimension-$n^k$ vectors, and hence be of size at least $(rnm)^k$. (Here, $z^{(i)}$ now has become a $m^k \times n^k$ matrix.)
However, we can now use the bound~\eqref{eq:yates-bound} provided by Yates's algorithm from Theorem~\ref{thm:yates} in its most basic form for $d=2$ to compute a single $n^k \times n^k \times r^k$ matrix-vector product of $z^{(i)}_{j} \cdot \Lambda^{(i)}$ using a circuit of size $O(r^{k+1} k)$ (as opposed to the naive $\Omega((nr)^k)$), assuming $r \geq n$. This gives the improved bound of $O((mr)^{k} r k)$ for $\CC(z^{(i)}\cdot \Lambda^{(i)})$.

Therefore, using $(U \otimes T)^{\otimes k} \equ U^{\otimes k} \otimes T^{\otimes k}$ and $\CC(U^{\otimes k}) > d\cdot m^k$ simply by counting the number of variables appearing in $U$, 
we find from \cref{lem:rank_and_circuit} that 
\[
\CC((U \otimes T)^{\otimes k}) \leq r^k \cdot (\CC(U^{\otimes k}) + 1) + O(kd(mr)^{k+1}) \leq O(kdr \cdot r^{k}\CC(U^{\otimes k})),
\]
which implies the claim, seeing that $(kdr)^{1/k} \to 1$ as $k \to \infty.$
\end{proof}

\subsection{Fractional Graph Tensors}\label{sec:frac_graphs}
As it turns out, our arguments can be stated cleanly using \emph{fractional graphs}, i.e., graphs with edge weights from $\mathbb Q$ (as opposed to multigraphs, which can be interpreted as graphs with edge-weights from $\mathbb N$.) 

In the following, we will identify multigraphs with integer weights with ordinary multigraphs where every weighted edge $e$ has been duplicated $w(e)$ times.

More specifically, let $G = (V,E,w)$ be a weighted multigraph with $w\colon E \rightarrow \mathbb Q_{> 0}$ a weight function taking on only positive values. We call such graphs \emph{fractional} in the following. Define $d_G \in \mathbb N$ to be the least common denominator of $w$, that is, the smallest positive integer such that $d_w \cdot w(e) \in \mathbb N $ for all $e \in E.$
Then, we define 
\[
    T_{G,n} \coloneqq T_{d_G \cdot G, n}.
\]
Consequently, for a \emph{fractional} graph $G$, we set
\[
    \omega(G) \coloneqq \omega(d_G \cdot G)/d_G,\ 
    \tau(G) \coloneqq \tau(d_G \cdot G).
\]
\begin{remark} \label{rem:omega-frac-def}
    Note that by Lemma~\ref{lem:power-mult}, the choice of $d_G$ is immaterial as long as the resulting graph is a proper multigraph, that is, has integer edge weights, in that
    \[
        \omega(G) = \omega(d \cdot G)/d
    \]
    for all $d \in \mathbb N$ such that $d \cdot G$ is a multigraph (which implies that $d$ is a multiple of $d_G$).
\end{remark}

Clearly, we can extend the notion of multigraph sums and multiplication with positive rationals to weighted multigraphs via pointwise addition and scaling of the weight functions.
In particular, we can identify an ordinary multigraph with the weighted multigraph that has all edge weights equal to $1$, and the two definitions of $\omega$ coincide in this case.
We observe that Lemma~\ref{lem:power-mult} carries over.
\begin{lemma} \label{lem:frac-power-mult}
    Let $G$ be a multigraph, and let $\alpha \in \mathbb Q_{>0}$. We denote by $\alpha \cdot G$ the fractional graph with edge weights multiplied by $\alpha$.
    Then, 
    \[
    \omega(\alpha \cdot G) = \alpha \omega(G).
    \]
\end{lemma}
\begin{proof}
    Pick $q \in \mathbb N$ such that $q\alpha \in \mathbb N$.
    By definition of $\omega$ for fractional graphs and Lemma~\ref{lem:power-mult}, we have
    \[
    \omega(\alpha \cdot G) = \omega(q\alpha \cdot G)/q = q\alpha \cdot \omega(G)/q = \alpha\omega(G). \qedhere
    \]
\end{proof}
\begin{remark} 
Lemma~\ref{lem:power-mult} also implies that 
\begin{align} \label{eq:frac-subadd}
\omega\left(\sum_{i} \alpha_i \cdot G_i\right) \leq \sum_{i} \alpha_i \cdot \omega(G)
\end{align}
holds for any finite sum of arbitrary fractional graphs $\alpha_i \cdot G_i.$
\end{remark}

For a fractional graph $G = (V,E,w)$, we define a \emph{conic decomposition} of $G$ to be a finite sequence $\mathcal H$ of fractional graphs $H$ such that
\[
\sum_{H \in \mathcal H} H = G,
\]
and we define 
\[
\omega(\mathcal H) = \sum_{H \in \mathcal H} \omega(H).
\]
We denote the set of all such conic decompositions of $G$ by $\mathcal C(G).$ Note that we consider $G$ as its own conic decomposition, hence the 1-element sequence $(G)$ is a valid conic decomposition of $G$.
\begin{lemma} \label{lem:cone}
    Let $G$ be a multigraph. Then
    \[
    \omega(G) = \min_{\mathcal H \in \mathcal C(G)} \omega(\mathcal H).
    \]
\end{lemma}
\begin{proof}
    This is simply by~\eqref{eq:frac-subadd}, observing that the minimum is attained by definition at $\omega(G)$ itself.
\end{proof}
\begin{remark}
    Lemma~\ref{lem:cone} seems trivial, but the point is that it can be used to obtain upper bounds on $\omega(G)$ for graphs $G$ by decomposing $G$ into fractional subgraphs $H$ for which we have good bounds on $\omega(H),$ such as, for instance, fractional multiples of triangles. 
\end{remark}

\subsection{Deriving Concrete Bounds}
The above notions and results we will now use to prove tighter upper bounds in the intermediate regime of $d=4,5$. Let us reiterate here that this selection of results is to be understood as a basic demonstration of the viability and utility of the approach, rather than a systematic and exhaustive exploitation of its merits.

\paragraph*{Sums of Stars.}
Recall from the proof of \cref{lem:strassen-generalized} that, for any tensor $T \in (\mathbb{F}^n)^{\otimes d}$, we have that $T \le T_{S_d(d), n}$, where $S_d(d)$ is the star graph on $d$ vertices with the central vertex labeled by $d$.
Define $\cat{k}{d} = T_{\sum_{u = 1}^k S_d(u), n}$.
Now, consider some $k \le d$, and recall that
\begin{align} \label{eq:cat-res}
T^{\otimes k} \res T_{\cat{k}{d},n} 
\end{align}
where we can visualize $\cat{k}{d}$ as the result of overlaying $k$ star graphs with different central vertices.

One can view $\cat{k}{d}$ as a $k$-clique with double edges between each pair of vertices, 
with $d-k$ extra vertices that are connected to each vertex in this $k$-clique, but not to each other.
For example, \cref{fig:cat24} shows the graphs $\cat{2}{4}$ and $\cat35$.
\begin{figure}[t]
    \centering
\begin{tikzpicture}[auto,node distance=2cm,
                    thick,main node/.style={circle,draw=black,fill=white,scale=1}]
  \node[main node] (1) {};
  \node[main node] (2) [right of=1] {};
  \node[main node] (3) [above of=1] {};
  \node[main node] (4) [above of=2] {};

  \path[every node/.style={font=\sffamily\small}]
    (1) edge [bend left=20] node {} (2)
        edge [bend right=20] node {} (2)
        edge [] node {} (3)
        edge [] node {} (4)
    (2) edge [] node {} (3)
        edge [] node {} (4);
\end{tikzpicture}
\hspace{2cm}
\begin{tikzpicture}[auto,node distance=2cm,
                    thick,main node/.style={circle,draw=black,fill=white,scale=1}]
   \def\radius{1.5}
  \node[main node] (1) at (-1,0) {};
  \node[main node] (2) at (-1,2) {};
  \node[main node] (3) at (1,2) {};
  \node[main node] (4) at (1,0) {};
  \node[main node] (5) at (0,0.5) {};

  \path[every node/.style={font=\sffamily\small}]
    (1) edge [] node {} (2)
        edge [] node {} (3)
        edge [bend left=10] node {} (4)
        edge [bend right=10] node {} (4)
        edge [bend left=10] node {} (5)
        edge [bend right=10] node {} (5)
    (2) edge [] node {} (5)
        edge [] node {} (4)
    (3) edge [] node {} (4)
        edge [] node {} (5)
    (4) edge [bend left=10] node {} (5)
        edge [bend right=10] node {} (5);
\end{tikzpicture}
    \caption{The graph $\cat{2}{4}$ (left) and $\cat{3}{5}$ (right).}
    \label{fig:cat24}
\end{figure}
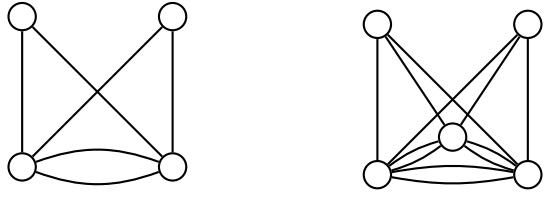
Then, our upper bounds follow by using \cref{thm:circuit-and-rank} to separate a given $\cat{k}{d}$ into two graphs $G_1, G_2$ such that $\cat{k}{d} = G_1 + G_2$.
Specifically, we have that 
\[
\AC(T_{\cat{k}{d}, n}) \le \AR(T_{G_1,n}) \cdot \AC(T_{G_2,n}).
\]
Thus, to upper bound the asymptotic circuit complexity of a sum of star graphs, we just need to upper bound the asymptotic rank and asymptotic circuit complexities of a partitioning of the edges into two parts.
Upper-bounding $\AC(T_{G_2, n})$ is conceptually easy using \cref{lem:circuit-from-ltw} and~\eqref{eq:ltw-ac}, as we simply need to compute the line graph and then find its tree width.

\paragraph*{Fractional Triangle Coverings.}
Now, we need a way to upper bound $\AR(T_{G_1, n})$, which we will do using  fractional triangle coverings.
Beyond the fractional tensors introduced in \Cref{sec:frac_graphs},
this requires one additional insight. 
Consider $\Delta_{d,t}(i,j,k)$, the fractional $d$-vertex graph with a single edge of weight $t \in \mathbb Q$ between $j,k$, and an edge of weight $1$ between $i,j$ and $k,i$ each. We call this graph a \emph{$t$-triangle}.
Moreover, let $\mathrm{MM}_n(t)$ denote the matrix multiplication tensor of an $n \times n^t$ matrix and an $n^t \times n$ matrix.
It is easy to see that for all choices of $d,t,i,j,k$ and $n$, we have that
\[
T_{\Delta_{d,t}(i,j,k),n} \equ \mathrm{MM}_n(t)
\]
holds. As is common in the literature, we write $\omega(t)$ as a shorthand for $\omega(\mathrm{MM}_n(t))$, the tensor exponent of $\mathrm{MM}_n(t).$
As an application of \Cref{lem:cone}, we hence find that 
\[
\omega(T_{G,n}) \leq \sum_{i=1}^\ell \omega(\alpha_i)
\]
holds for every conic decomposition of $G$ into fractional graphs $\Delta_1,\ldots,\Delta_\ell$, where $\Delta_i$ is an $\alpha_i$-triangle for $i = 1,\ldots,\ell.$ 

It is therefore enough to find a partition of the graph into two parts, one being decomposed into such a set of fractional triangles, and the other according to the bound on the asymptotic circuit complexity via line graph treewidth, as in \Cref{lem:circuit-from-ltw}. The objective is to minimize the sum of the resulting asymptotic exponents in both parts.
From a practical computational perspective, this can be done by a combination of brute-force and a simple linear program.
For the upper bounds on $\omega(k)$, we used the values calculated in \cite{gall_urrutia_18}: 
\[
\omega(0.5)\leq 2.046681, \qquad \omega(2)\leq 3.256689.
\]
We can then formalize this idea into the following lemma. Here, the graph $G_2$ can be viewed as the ``leftover edges'' that are handled in a brute-force manner. 
\begin{lemma}\label{lem:compute_helper}
Let $G$ be a weighted graph and $G_1, G_2, G_3$ such that $G = G_1 + G_2 + G_3$, $G_1$ has a conic decomposition $\mathcal H$ into fractional triangles and $G_2$ as well as $G_3$ are ordinary multigraphs, then we have that
\[
\AC(T_G) \le n^{\omega(\mathcal H) + |E(G_2)| + \ltw{G_3} + 1}.
\]
\end{lemma}
\begin{proof}
By \Cref{thm:circuit-and-rank} and \Cref{lem:circuit-from-ltw}, 
$\AC(T_{G,n}) \le \AR(T_{G_1,n}) \cdot \AR(T_{G_2,n}) \cdot \AC(T_{G_3,n}) \le n^{\omega(\mathcal H)} \cdot n^{|E(G_2)|} \cdot n^{\ltw{G_3} + 1}.$
\end{proof}

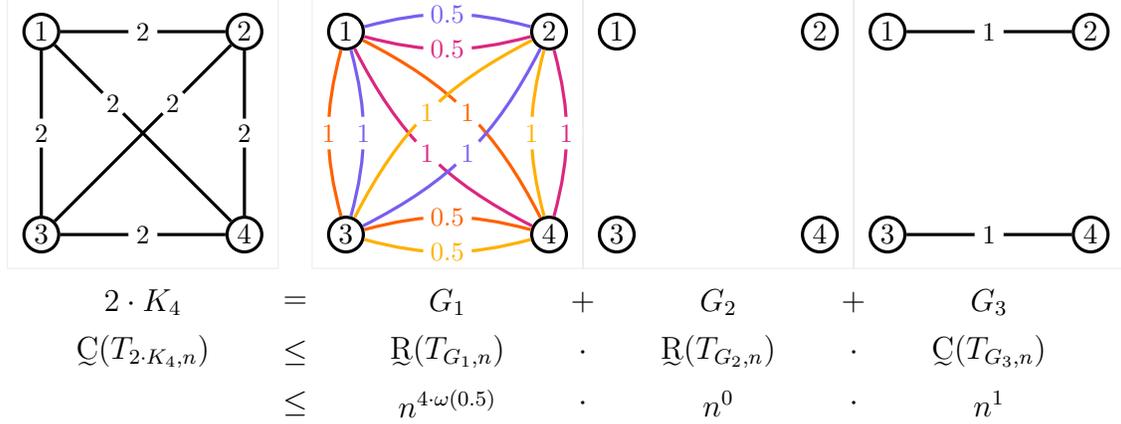
\begin{figure}[t]
\centering
\begin{tikzpicture}[main node/.style={circle,draw,very thick,inner sep=1.5pt},scale=0.9]
\clip (-8.6,-4.5) rectangle (8.1,2.25);
\draw [black,opacity=0] (-8.6,-4.5) rectangle (8.1,2.25);
\def\recsize{1.5}
\def\bdsize{0.5}

\node [] () at (-6.5,-2.5)  {\large $2\cdot K_4$};
\node [] () at (-2,-2.5)  {\large $G_1$};
\node [] () at ( 2,-2.5)  {\large $G_2$};
\node [] () at ( 6,-2.5)  {\large $G_3$};
\node [] () at (-4.25,-2.5)  {\large $=$};
\node [] () at ( 0,-2.5)  {\large $+$};
\node [] () at ( 4,-2.5)  {\large $+$};
\node [] () at (-6.5,-3.25) {\large $\AC(T_{2\cdot K_4,n})$};
\node [] () at (-2,-3.25) {\large $\AR(T_{G_1,n})$};
\node [] () at ( 2,-3.25) {\large $\AR(T_{G_2,n})$};
\node [] () at ( 6,-3.25) {\large $\AC(T_{G_3,n})$};
\node [] () at (-4.25,-3.25) {\large $\leq$};
\node [] () at ( 0,-3.25) {\large $\cdot$};
\node [] () at ( 4,-3.25) {\large $\cdot$};
\node [] () at (-2,-4.00) {\large $n^{4\cdot\omega(0.5)}$};
\node [] () at ( 2,-4.00) {\large $n^{0}$};
\node [] () at ( 6,-4.00) {\large $n^{1}$};
\node [] () at (-4.25,-4.00) {\large $\leq$};
\node [] () at ( 0,-4.00) {\large $\cdot$};
\node [] () at ( 4,-4.00) {\large $\cdot$};

\begin{scope}[shift={(6,0)}]
      \node[main node] (aa1) at ({-\recsize},{\recsize})  {1};
      \node[main node] (aa2) at ({\recsize}, {\recsize})  {2};
      \node[main node] (aa3) at ({-\recsize},{-\recsize}) {3};
      \node[main node] (aa4) at ({\recsize}, {-\recsize}) {4};
    
      \path[very thick, every node/.style={font=\sffamily\small,midway,fill=white,auto=false,align=center,inner sep=1pt,circle}]
        (aa1) edge [] node {$1$} (aa2)
        (aa3) edge [] node {$1$} (aa4);
    \draw [gray!30,opacity=0.5] ({-\bdsize-\recsize},{-\bdsize-\recsize}) rectangle ({\bdsize+\recsize},{\bdsize+\recsize});
\end{scope}

\begin{scope}[shift={(2,0)}]
      \node[main node] (aa1) at ({-\recsize},{\recsize})  {1};
      \node[main node] (aa2) at ({\recsize}, {\recsize})  {2};
      \node[main node] (aa3) at ({-\recsize},{-\recsize}) {3};
      \node[main node] (aa4) at ({\recsize}, {-\recsize}) {4};
    \draw [gray!30,opacity=0.5] ({-\bdsize-\recsize},{-\bdsize-\recsize}) rectangle ({\bdsize+\recsize},{\bdsize+\recsize});
\end{scope}

\begin{scope}[shift={(-2,0)}]
      \node[main node] (1) at ({-\recsize},{\recsize})  {1};
      \node[main node] (2) at ({\recsize}, {\recsize})  {2};
      \node[main node] (3) at ({-\recsize},{-\recsize}) {3};
      \node[main node] (4) at ({\recsize}, {-\recsize}) {4};
    
      \path[very thick, every node/.style={font=\sffamily\small,midway,fill=white,auto=false,align=center,inner sep=1.5pt,circle}]
        (1) edge [bend left =15, cbfp1] node {$0.5$} (2)
            edge [bend right=15, cbfp2] node {$0.5$} (2)
            edge [bend left =15, cbfp1] node {$1$} (3)
            edge [bend right=15, cbfp3] node {$1$} (3)
            edge [bend left =18, cbfp3] node {$1$} (4)
            edge [bend right=18, cbfp2] node {$1$} (4)
        (2) edge [bend left =18, cbfp1] node {$1$} (3)
            edge [bend right=18, cbfp4] node {$1$} (3)
            edge [bend left =15, cbfp2] node {$1$} (4)
            edge [bend right=15, cbfp4] node {$1$} (4)
        (3) edge [bend left =15, cbfp3] node {$0.5$} (4)
            edge [bend right=15, cbfp4] node {$0.5$} (4);
    \draw [gray!30,opacity=0.5] ({-\bdsize-\recsize},{-\bdsize-\recsize}) rectangle ({\bdsize+\recsize},{\bdsize+\recsize});
\end{scope}

\begin{scope}[shift={(-6.5,0)}]
      \node[main node] (1) at ({-\recsize},{\recsize})  {1};
      \node[main node] (2) at ({\recsize}, {\recsize})  {2};
      \node[main node] (3) at ({-\recsize},{-\recsize}) {3};
      \node[main node] (4) at ({\recsize}, {-\recsize}) {4};
    
      \path[very thick, every node/.style={font=\sffamily\small,fill=white,auto=false,align=center,inner sep=1pt,circle}]
        (1) edge [midway] node {$2$} (2)
            edge [midway] node {$2$} (3)
            edge [pos=0.33] node {$2$} (4)
        (2) edge [pos=0.33] node {$2$} (3)
            edge [midway] node {$2$} (4)
        (3) edge [midway] node {$2$} (4);
    \draw [gray!30,opacity=0.5] ({-\bdsize-\recsize},{-\bdsize-\recsize}) rectangle ({\bdsize+\recsize},{\bdsize+\recsize});
\end{scope}
\end{tikzpicture}
\caption{Decomposition of $2 \cdot K_4$ into $G_1,G_2,G_3$, where $G_1$ has a conic decomposition into four $0.5$-triangles indicated by edge colors, $G_2$ is an empty graph, and $G_3$ is a matching. }
\label{fig:four_clique_triangulization}
\end{figure}
\paragraph*{Bounds.} Our experimental programs were able to obtain results on $4$ and $5$ vertices.
Larger number of vertices seemed to take too long to run to completion.
For these two cases, we provide our results in the following two theorems.
\begin{theorem} \label{thm:4-mode-exponent}
Every $4$-mode tensor $T \in (\mathbb F^n)^{\otimes 4}$ satisfies
$\AC(T) \le n^{1/4 + \omega(0.5)} \le n^{2.296681}$.
\end{theorem}
\begin{proof}
From~\eqref{eq:cat-res}, we see that $T^{\otimes 4} \le T_{2 \cdot K_4,n}$, where $2 \cdot K_4 = \cat44$ is the doubled complete $4$-vertex graph.
Using \Cref{lem:compute_helper} and the decomposition of \Cref{fig:four_clique_triangulization}, we observe that
$\AC(T_{2\cdot K_4,n}) \le n^{1 + 4 \cdot \omega(0.5)}.$ Hence, 
$\AC(T) \le \AC(T^{\otimes 4})^{1/4} \le n^{1/4 + \omega(0.5)}.$
\end{proof}
\begin{theorem} \label{thm:5-mode-exponent}
Every $5$-mode tensor $T \in (\mathbb F^n)^{\otimes 5}$ satisfies
$\AC(T) \le n^{(3 + \omega + \omega(2))/3} \le n^{2.877389}$.
\end{theorem}
\begin{proof}
By~\eqref{eq:cat-res}, $T^{\otimes 3} \le T_{\cat{3}{5}}$.
Then, using the decomposition given in \Cref{fig:five_three_cat_triangulization} that picks $G_1 = \Delta_{5,1}(1,3,5) + \Delta_{5,2}(2,1,3)$, $G_2$ as the single edge $4,2$, and $G_3$ as the $4$-edge path $3,4,1,2,5$, we observe via \Cref{lem:compute_helper} that
\[
\AC(T_{\cat{3}{5},n}) \le n^{\omega + \omega(2) + |E(G_2)| + \ltw{G_3} + 1} \leq n^{\omega + \omega(2) + 3},
\]
given that the line graph of a path is a path and hence has treewidth $1$.
Hence, we have that
\[
\AC(T) \le \AC(T^{\otimes 3})^{1/3} \le \AC(T_{\cat{3}{5},n})^{1/3} \le \AC(T_{\cat{3}{5}})^{1/3} \le n^{1 + (\omega + \omega(2))/3}.\qedhere
\]
\end{proof}
The preceding theorems together constitute \cref{thm:intro-45ub} from the introduction.
We remark here that our results do not rely on the improved bounds obtained for the rank of $T_{K_4,n}$ referred to in~\cref{thm:asymptotic-submultiplicativity}.
Judging from computational evidence, the choice of $d = 4,5$ seems to be too small to profit from the edge afforded by this bound.
For higher values of $d$ than those handled by our code, this could become more relevant.

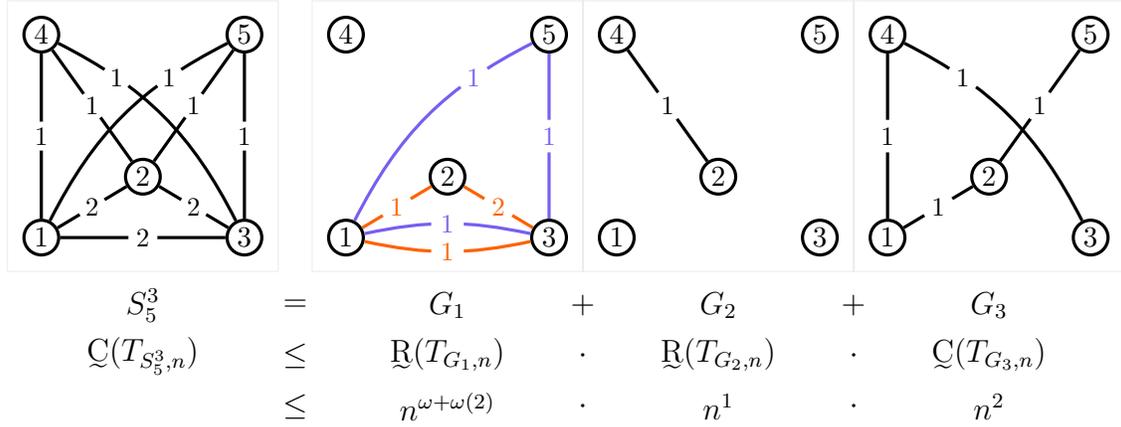
\begin{figure}
\centering
\begin{tikzpicture}[main node/.style={circle,draw,very thick,inner sep=1.5pt},scale=0.9]
\clip (-8.6,-4.5) rectangle (8.1,2.25);
\draw [black,opacity=0] (-8.6,-4.5) rectangle (8.1,2.25);
\def\recsize{1.5}
\def\bdsize{0.5}

\node [] () at (-6.5,-2.5)  {\large $S_5^3$};
\node [] () at (-2,-2.5)  {\large $G_1$};
\node [] () at ( 2,-2.5)  {\large $G_2$};
\node [] () at ( 6,-2.5)  {\large $G_3$};
\node [] () at (-4.25,-2.5)  {\large $=$};
\node [] () at ( 0,-2.5)  {\large $+$};
\node [] () at ( 4,-2.5)  {\large $+$};
\node [] () at (-6.5,-3.25) {\large $\AC(T_{S_5^3,n})$};
\node [] () at (-2,-3.25) {\large $\AR(T_{G_1,n})$};
\node [] () at ( 2,-3.25) {\large $\AR(T_{G_2,n})$};
\node [] () at ( 6,-3.25) {\large $\AC(T_{G_3,n})$};
\node [] () at (-4.25,-3.25) {\large $\leq$};
\node [] () at ( 0,-3.25) {\large $\cdot$};
\node [] () at ( 4,-3.25) {\large $\cdot$};
\node [] () at (-2,-4.00) {\large $n^{\omega+\omega(2)}$};
\node [] () at ( 2,-4.00) {\large $n^{1}$};
\node [] () at ( 6,-4.00) {\large $n^{2}$};
\node [] () at (-4.25,-4.00) {\large $\leq$};
\node [] () at ( 0,-4.00) {\large $\cdot$};
\node [] () at ( 4,-4.00) {\large $\cdot$};

\begin{scope}[shift={(6,0)}]
    \draw [gray!30,opacity=0.5] ({-\bdsize-\recsize},{-\bdsize-\recsize}) rectangle ({\bdsize+\recsize},{\bdsize+\recsize});
    \node[main node] (1) at ({-\recsize},{-\recsize}) {1};
    \node[main node] (4) at ({-\recsize},{ \recsize}) {4};
    \node[main node] (5) at ({ \recsize},{ \recsize}) {5};
    \node[main node] (3) at ({ \recsize},{-\recsize}) {3};
    \node[main node] (2) at (0,{-\recsize/2.5}) {2};
    \path[very thick, every node/.style={font=\sffamily\small,fill=white,auto=false,align=center,inner sep=1.5pt,circle}]
            (5) edge [midway] node {$1$} (2)
            (2) edge [midway] node {$1$} (1)
            (1) edge [midway] node {$1$} (4)
            (4) edge [bend left=18,pos=0.25] node {$1$} (3);
\end{scope}

\begin{scope}[shift={(2,0)}]
    \draw [gray!30,opacity=0.5] ({-\bdsize-\recsize},{-\bdsize-\recsize}) rectangle ({\bdsize+\recsize},{\bdsize+\recsize});
    \node[main node] (1) at ({-\recsize},{-\recsize}) {1};
    \node[main node] (4) at ({-\recsize},{ \recsize}) {4};
    \node[main node] (5) at ({ \recsize},{ \recsize}) {5};
    \node[main node] (3) at ({ \recsize},{-\recsize}) {3};
    \node[main node] (2) at (0,{-\recsize/2.5}) {2};
    \path[very thick, every node/.style={font=\sffamily\small,fill=white,auto=false,align=center,inner sep=1.5pt,circle}]
        (2) edge [midway] node {$1$} (4);
\end{scope}

\begin{scope}[shift={(-2,0)}]
    \draw [gray!30,opacity=0.5] ({-\bdsize-\recsize},{-\bdsize-\recsize}) rectangle ({\bdsize+\recsize},{\bdsize+\recsize});
          \node[main node] (a1) at ({-\recsize},{-\recsize}) {1};
          \node[main node] (a4) at ({-\recsize},{ \recsize}) {4};
          \node[main node] (a5) at ({ \recsize},{ \recsize}) {5};
          \node[main node] (a3) at ({ \recsize},{-\recsize}) {3};
          \node[main node] (a2) at (0,{-\recsize/2.5}) {2};
    
      \path[very thick, every node/.style={font=\sffamily\small,fill=white,auto=false,align=center,inner sep=1.5pt,circle}]
            (a1) edge [pos=0.75,bend left=18,cbfp1] node {$1$} (a5)
                edge [midway,bend left=12, cbfp1] node {$1$} (a3)
                edge [midway,bend right=12, cbfp3] node {$1$} (a3)
                edge [midway,cbfp3] node {$1$} (a2)
            (a3) edge [midway,cbfp1] node {$1$} (a5)
            (a2) edge [midway,cbfp3] node {$2$} (a3);
\end{scope}

\begin{scope}[shift={(-6.5,0)}]
    \draw [gray!30,opacity=0.5] ({-\bdsize-\recsize},{-\bdsize-\recsize}) rectangle ({\bdsize+\recsize},{\bdsize+\recsize});
          \node[main node] (1) at ({-\recsize},{-\recsize}) {1};
          \node[main node] (4) at ({-\recsize},{ \recsize}) {4};
          \node[main node] (5) at ({ \recsize},{ \recsize}) {5};
          \node[main node] (3) at ({ \recsize},{-\recsize}) {3};
          \node[main node] (2) at (0,{-\recsize/2.5}) {2};
    
      \path[very thick, every node/.style={font=\sffamily\small,fill=white,auto=false,align=center,inner sep=1pt,circle}]
        (1) edge [midway] node {$2$} (2)
            edge [midway] node {$2$} (3)
            edge [midway] node {$1$} (4)
            edge [pos=0.75, bend left=18] node {$1$} (5)
        (2) edge [midway] node {$2$} (3)
            edge [midway] node {$1$} (4)
            edge [midway] node {$1$} (5)
        (3) edge [pos=0.75, bend right=18] node {$1$} (4)
            edge [midway] node {$1$} (5);
\end{scope}
\end{tikzpicture}
\caption{Decomposition of $\cat35$ into $G_1,G_2,G_3$, where $G_1$ has a conic decomposition into a \textcolor{cbfp3}{$2$-triangle} and a \textcolor{cbfp1}{$1$-triangle}, $G_2$ contains a single edge, and $G_3$ is a path. }
\label{fig:five_three_cat_triangulization}
\end{figure}

\bibliographystyle{plainurl}
\bibliography{literature}

\appendix
\addtocontents{toc}{\protect\setcounter{tocdepth}{1}}

\section{Properties of Graph Tensors}
\label{sect:graph-tensor-properties}

This appendix collects and proves some further basic facts about graph tensors in relation to their graphs, including proofs of \cref{lem:graph-tensor-topological-subgraph}, \ref{lem:contraction} and \ref{lem:power-mult}.

\begin{lemma}[Isomorphism implies equivalence]
    Let $G$ and $H$ be graphs and let $n\in\mathbb{N}$. If $G\isom H$, then $T_{G,n}\equ T_{H,n}$.
\end{lemma}
\begin{proof}
    Recall \cref{eq:graph-tensor} and
    write $x^{(v)}_{g}$ for $v \in V(G)$ and $g : I_G(v) \to [n]$ for the indeterminates of the polynomial $T_{G, n}$.
    Similarly, write $y^{(v)}_{h}$ for $v \in V(H)$ and $g : I_{H}(v) \to [n]$ for the indeterminates of 
    the polynomial $T_{H, n}$.
    Let the bijections $\phi:V(H)\rightarrow V(H)$ and $\psi:E(G)\rightarrow E(H)$ be an isomorphism from $G$ to $H$.
    The substitution $x^{(v)}_{g} \coloneqq y^{(\phi(v))}_{\psi(g)}$ obtains $T_{G, n}$ from $T_{H, n}$,
    and the substitution $y^{(v)}_{g} \coloneqq x^{(\phi^{-1}(v))}_{\psi^{-1}(g)}$ obtains $T_{H, n}$ from $T_{G, n}$.
\end{proof}

\begin{lemma}[Projection under length reduction] \label{lem:graph-tensor-length}
    Let $G$ be a graph. If $m\leq n$, then $T_{G,m}\pro T_{G,n}$.
\end{lemma}
\begin{proof}
    In the polynomial $T_{G,n}$, substitute $0$ to every indeterminate $x^{(v)}_{g}$ such that $v\in V(G)$, $g: I_G(v) \rightarrow [n]$, and there is an $e \in I_G(v)$ such that $g(e) \geq m+1$. The resulting polynomial equals the polynomial $T_{G,m}$.
\end{proof}

\begin{lemma}[Projection under subgraphs] \label{lem:graph-tensor-subgraph}
    Let $G$ and $H$ be graphs and let $n \in \mathbb{N}$. If $H$ is a subgraph of $G$, then $T_{H, n}\pro T_{G, n}$.
\end{lemma}
\begin{proof}
    Recall that
    \[
        T_{G,n} = \sum_{f : E(G)\to [n]}\prod_{v\in V(G)} x^{(v)}_{f|_{I_{G}(v)}}\,.
    \]
    Now $T_{G,n}$ can be projected by substituting, for all $v \in V(G)$ and $g : I_{G}(v) \to [n]$,
    \[
      x^{(v)}_{g} \coloneqq \begin{cases}
          1, & \text{if } v \not\in V(H),\\
          x^{(v)}_{g|_{I_H(v)}}, & \text{if } v \in V(H) \text{ and } g(e) = 1 \text{ for all } e\in I_{G}(v) \setminus I_H(v), \\
          0, & \text{otherwise}.
      \end{cases}  
    \]
    Indeed, under this substitution, it holds that
    \begin{align*}
        T_{H, n}&=\sum_{f :: E(H) \to [n]} \prod_{v \in V(H)} x^{(v)}_{f|_{I_{H}(v)}}
        \pro
        \sum_{\substack{f : E(H) \to [n]\\f' : E(G) \setminus E(H) \to [n]}}\prod_{v \in V } x^{(v)}_{f\cup f'|_{I_{G}(v)}}
        = T_{G,n}
    \end{align*}
    as desired.
\end{proof}

We are now ready to prove \cref{lem:graph-tensor-topological-subgraph}.

\begin{proof}[Proof of \cref{lem:graph-tensor-topological-subgraph}]
    From \cref{lem:graph-tensor-subgraph} it follows that 
    it suffices to consider the case when $G$ is obtained from $H$ by subdividing some edge $e\in E(H)$ with ends $u,v\in V(G)$.
    Write $w \in V(G)$ for the vertex introduced by this subdivision. Let us also write $uv$ and $vw$ for the edges introduced 
    by the subdivision. We thus have $V(G) = V(H) \cup \{w\}$ and $E(G) = (E(H) \setminus \{e\}) \cup \{uw, vw\}$.
    It holds that
    \begin{align*}
        T_{G,n} 
        &= \sum_{f : E(G)\to [n]}\prod_{a \in V(G)} x^{(a)}_{f|_{I_G(a)}} \\
        &= \sum_{f : E(H) \setminus \{e\} \to [n]} \sum_{i,j\in [n]}
        x^{(w)}_{uw \mapsto i, vw \mapsto j} \cdot x^{(u)}_{\{uw \mapsto i\} \cup f|_{I_{H}(u) \setminus \{e\} }}\cdot x^{(v)}_{\{vw \mapsto j\} \cup f|_{I_{H}(v) \setminus \{e\} }} \cdot \prod_{b \in V(H)\setminus \{u,v\}} x^{(b)}_{f|_{I_{H}(b)}}\,.
    \end{align*}
        By substituting $x^{(w)}_{uw \mapsto i, vw \mapsto j} \coloneqq \delta_{ij}$ into the polynomial above, we obtain the polynomial
    \[
        \sum_{f : E  \to [n]} \prod_{v \in V } x^{(v)}_{f|_{I_{H}}(v)} = T_{H, n}\,,
    \]
    implying $T_{H,n}\pro T_{G,n}$ as desired.
\end{proof}
With a standard argument, we can also prove \Cref{lem:contraction} on contractions:
\begin{proof}[Proof of \cref{lem:contraction}]
Let $E \coloneqq E(G)$. Write $A = E(G[U])$ for the edges within $U$, and let $w$ be the vertex replacing $U$ in $G/U$; 
this vertex is incident with the edges in $G$ that have exactly one endpoint in $U$.

Given functions $f,g$ on disjoint domains, we write $fg$ for the function that agrees with $f$ and $g$ on their respective domains.
We have 
\begin{align*}
T_{G,n} & = \sum_{f \colon E\to [n]}\prod_{v\in V} x^{(v)}_{f|_{I(v)}} \\
& = \sum_{f \colon E\setminus A \to [n]} \underbrace{\left( \sum_{g \colon A \to [n]}\prod_{v\in U} x^{(v)}_{fg|_{I(v)}} \right)}_{=h_{f|_{I(w)}}} \left( \prod_{v\in V\setminus U} x^{(v)}_{f|_{I(v)}} \right)\\
& = \sum_{f \colon E\setminus A \to [n]} h_{f|_{I(w)}} \prod_{v\in V\setminus U} x^{(v)}_{f|_{I(v)}}.
\end{align*}
Each polynomial $h_{f|_{I(w)}}$ admits a circuit of size $O(|U|\cdot n^{|A|})$, and the last right-hand side is obtained from $T_{G/U,n}$ by the substitution $x^{(w)}_{f|_{I(w)}} \gets h_{f|_{I(w)}}$ in the mode for $w$, which represents $U$.  
\end{proof}

We also include for completeness:
\begin{proof}[Proof \cref{lem:power-mult}]
By \cref{lem:decomp},
    we have $\RR(T_{G+H,n})=\RR(T_{G,n}\otimes T_{H,n})\leq \RR(T_{G,n})\cdot\RR(T_{H,n})$, so $\omega(G+H) \leq \omega(G)+\omega(H)$ and $\omega(k \cdot G) \leq k\cdot\omega(G)$ follow.
    For the other inequality, let $\beta = \omega(k \cdot G)$, and hence $\RR(T_{k \cdot G,n}) \leq c n^\beta$ for some $c>0$ and all large enough $n$. By the sum rule (\cref{rem:flip}) this means that also $\RR(T_{G,n^k}) \leq cn^\beta$ for large enough $n$.
    For $\bar n = \lceil n^{1/k} \rceil^k \geq n$ we have that 
    $T_{G,n}\pro T_{G,\bar n}$. 
    Hence, for large enough $n$, $\RR(T_{G,n}) \leq \RR(T_{G,\bar n}) \leq c \lceil n^{1/k} \rceil^{\beta} \leq c(n^{1/k}+1)^\beta \leq c 2^\beta n^{\beta/k} = O(n^{\beta/k})$.
    By definition of $\omega(G)$ we conclude $\omega(G) \leq \beta/k = \omega(k\cdot G)/k$.
\end{proof}

\section{Facts on Asymptotic Circuit Complexity}
\label{sec:app-circuits}

This appendix collects select facts on asymptotic circuit complexity.

\begin{fact} \label{fact:res-submult-circuit}
Let $U,V \in (\mathbb F^n)^{\otimes d}$ with $V \res U$. Then
\[
\AC(V) \leq \AC(U)
\]
holds.
\end{fact}
\begin{proof}
    The claim follows by an argument close to the proof of~\cref{lem:rank_and_circuit} and \cref{thm:circuit-and-rank} as an application of tensor contraction.
    Let $L_1,\ldots,L_d$ be the linear maps that witness $V \leq U$, that is, $V = U(L_1(x^{(1)}), \ldots, L_d(x^{(d)})).$ 
    Evaluation can be shown to be compatible with Kronecker powering, so that $V^{\otimes k} = U^{\otimes k}(L_1^{\otimes k}(z^{(1)}),\ldots,L_d^{\otimes k}(z^{(d)}))$.
    The claim then follows by Yates's algorithm~\eqref{eq:yates-bound} as in the proof of \cref{thm:circuit-and-rank}.
\end{proof}

\begin{fact} \label{fact:prod-res}
    Let $U \in (\mathbb F^n)^{\otimes d}, V \in (\mathbb F^m)^{\otimes d}$ with $V \neq 0$.
    Then, $U \res U \otimes V$, and $\CC(U) \leq \CC(U\otimes V).$
\end{fact}
\begin{proof}
    Suppose the coefficients of $U$ and $V$ are 
    \[
    U = \sum_{i_1,\ldots,i_d=1}^n u_{i_1,\ldots,i_d} \cdot x^{(1)}_{i_1}\cdots x^{(d)}_{i_d},\quad V = \sum_{j_1,\ldots,j_d=1}^m v_{j_1,\ldots,j_d} \cdot y^{(1)}_{j_1}\cdots y^{(d)}_{j_d}
    \]
    and assume without loss of generality that $v_{1,\ldots,1} \neq 0.$
    
    Let $C$ be circuit computing $U \otimes V$ with modes $Z^{(i)}$ containing variables $z^{(i)}_{k,\ell},$ that is
    \[
    U \otimes V = \sum_{i_1,\ldots,i_d=1}^n \sum_{j_1,\ldots,j_d=1}^m u_{i_1,\ldots,i_d} \cdot v_{j_1,\ldots,j_d} \cdot z^{(1)}_{i_1,j_1} \cdots z^{(d)}_{i_d,j_d}.
    \]
    Define the substitution $\mu$ via $\mu(z^{(1)}_{k,1}) = \frac{1}{v_{1,\ldots,1}} x^{(1)}_k,\ \mu(z^{(i)}_{k,1}) = x^{(i)}_k$ for $i > 1$ and all $k$, and $\mu(z^{(i)}_{k,\ell}) = 0$ for all $i$ and $\ell > 1.$
    Since only terms with $j_1,\ldots,j_d = 1$ survive in $U\otimes V$ under this substitution, we find
    \[
    (U \otimes V)^\mu = \sum_{i_1,\ldots,i_d=1}^n u_{i_1,\ldots,i_d} \cdot v_{1\ldots 1} \cdot \frac{1}{v_{1\ldots 1}} \cdot x^{(1)}_{i_1} \cdots x^{(d)}_{i_d} = U.
    \]
    Note that in our circuit model, this substitution can be performed without any size increase in $C$, by either replacing inputs with $0$, or substituting inputs for another with the scalar factor $1/v_{1\ldots 1}$ possibly at the appropriate wires, which is free. The claim follows. (Note that for general projections as afforded via $\res$ instead of $\pro$, the translation on the circuit sizes is not automatic, necessitating a closer look as given in this proof.)
\end{proof}

\begin{fact} \label{fact:submult-montone}
    If $\CC$ is submultiplicative on $d'$-mode tensors, for any fixed $d' \in \mathbb N$, then $\CC$ is submultiplicative on $d$-mode tensors for all $1 \leq d \leq d'$.
\end{fact}
\begin{proof}
    By induction, it suffices to consider the case when $d' = d+1$.
    Suppose that $\CC$ is submultiplicative on $(d+1)$-mode tensors.
    Let $S \in \mathbb{C}^{n_1} \otimes \dots \otimes \mathbb{C}^{n_d}$
    and $T \in \mathbb{C}^{m_1} \otimes \dots \otimes \mathbb{C}^{m_d}$
    for $n_1, \dots, n_1, m_1, \dots, m_d \geq 0$.
    Extend $S$ and $T$ to $(d+1)$-mode tensors $\widehat{S}$ and $\widehat{T}$ by tensoring with $1$, i.e.\ appending a $0$-dimensional mode.
    It holds that $\CC(S) = \CC(\widehat{S})$
    and furthermore
    \[
        \CC(S \otimes T)
        = \CC(\widehat{S} \otimes \widehat{T})
        \leq \CC(\widehat{S}) \cdot \CC(\widehat{T})
        = \CC(S) \cdot \CC(T),
    \]
    as desired.
\end{proof}

\paragraph*{Yates's Algorithm.}
We now briefly recall Yates's algorithm and give a generalized exposition that includes the case $d \geq 4$.
Suppose $T \in (\mathbb F^n)^{\otimes d}$ is concise, such that $T = \sum_{i=1}^r \prod_{j=1}^d \ell_{i}^{(j)}(x^{(j)}),$ where $\ell_{i}^{(j)}(x^{(j)}) = \sum_{p=1}^n \ell_{ip}^{(j)} x^{(j)}_p$ are linear forms and $r \geq n$. For $j\in[d]$ and $k\in\mathbb{N}$, let us write $z^{(j)}$ for the indeterminates of mode $j$ of $T^{\otimes k}$. We can thus view the tensor $T^{\otimes k}$ as the polynomial
\begin{equation}
\label{eq:yates-top}
T^{\otimes k}(z^{(1)},\ldots,z^{(d)}) = \sum_{i \in [r]^k} \prod_{j=1}^d \ell^{(j)}_{i}(z^{(j)})\,,
\end{equation}
where for $i=(i_1,\ldots,i_k)\in [r]^k$ we define $\ell^{(j)}_{i} = \ell^{(j)}_{i_1} \otimes \cdots \otimes \ell^{(j)}_{i_k}$.
Yates's algorithm gives a circuit for evaluating the polynomial $T^{\otimes k}$ as follows. Let us write $\varepsilon$ for the empty tuple.  
For $j\in[d]$, $q\in\{0,1,\ldots,k\}$, $i \in [r]^q$, and $h \in [n]^{k-q}$, define the intermediate polynomial  
\[
L^{(j)}_{h,i}(z^{(j)}) = \sum_{p \in [n]^{q}}
\ell^{(j)}_{i_1 p_1}\cdots \ell^{(j)}_{i_q p_q}z^{(j)}_{h,p}\,.
\]
We observe in particular that for $q=0$ we have $L^{(j)}_{h,\varepsilon}(z^{(j)}) = z^{(j)}_h$ for $h\in[n]^k$. 
These indeterminates $z^{(j)}_h$ for all $j\in[d]$ and $h\in[n]^k$ are the $dn^k$ inputs to the circuit. We construct the circuit one stage $q\in\{0,1,\ldots,k\}$ at a time, ensuring by induction that stage $q$ has gates that evaluate to the polynomial $L^{(j)}_{h,i}(z^{(j)})$ for each $h\in [n]^{k-q}$ and $i\in[r]^{q}$. The input gates at stage $q=0$ form the base case. Assume stage $q-1$ has been constructed and construct stage $q$ using the recurrence
\begin{equation}
\label{eq:yates-recurrence}
L^{(j)}_{h,it}(z^{(j)}) = \sum_{p=1}^n\ell^{(j)}_{tp}L^{(j)}_{hp,i}(z^{(j)})
\end{equation}
for all $t\in[r]$, $i \in [r]^{q-1}$, and $h \in [n]^{k-q+1}$. From \eqref{eq:yates-recurrence} we immediately observe that stage $q\in[k]$ can be constructed using at most $O(n^{k-q+1} \cdot r^q) = O(n^{k+1} \cdot (r/n)^q)$ wires. Furthermore, stage $q=k$ has gates that evaluate to $L^{(j)}_{\varepsilon,i}(z^{(j)}) = \ell^{(j)}_{i}(z^{(j)})$ for all $i\in[r]^k$. Complete the circuit for $T^{\otimes k}$ using \eqref{eq:yates-top} and wiring in from gates at stage $k$; at most $O(dr^k)$ wires suffice for this completion. We thus have a circuit of size at most $O(dkr^{k+1})$ for $T^{\otimes k}$, as desired for~\eqref{eq:yates-bound}, immediately implying also \Cref{thm:yates}.

\section{Improved Asymptotic Rank of the 4-Clique} \label{sec:ar}

The purpose of this appendix is to present an improved upper bound on $\tau(K_4)$.

\begin{theorem}\label{thm:tauK4-improved}
For the $4$-clique, we have $\tau(K_4) < 0.772318$.
\end{theorem}

This improves upon the previous bound $\tau(K_4) < 0.772943$ due to
Christandl, Vrana, and Zuiddam~\cite{ChristandlVZ19}.
The improvement arises from a refinement of the analysis of Coppersmith--Winograd tensors (CW tensors)~\cite{coppersmith1990matrix}
in the $k$-mode setting.  While~\cite{ChristandlVZ19} extends the classical $3$-mode CW construction to $k$ modes,
their analysis is restricted to the \emph{small} CW tensors (corresponding to the \emph{simple construction} in~\cite{coppersmith1990matrix}).
Here we incorporate the \emph{big} CW tensors (the \emph{complicated construction}), which leads to a slightly improved bound.

Throughout this appendix we follow the terminology and general framework of~\cite{ChristandlVZ19}.
Our notation differs slightly: our small CW tensor $\cw_q^k$ corresponds to the CW tensor analyzed in~\cite{ChristandlVZ19}.

\subsection{CW tensors in $k$ modes}

\begin{definition}[Small and big CW $k$-tensors]\label{def:CWk}
Let $k,q\ge 2$ be integers. For each mode $u\in[k]$, let $A^{(u)}\cong \bbF^{q+2}$ be equipped with a basis
\[
x_0^{(u)},x_1^{(u)},\dots,x_q^{(u)},x_{q+1}^{(u)}.
\]
The \emph{small} CW $k$-tensor $\cw_q^k$ is the element of
$\bigotimes_{u=1}^k \langle x_0^{(u)},\dots,x_q^{(u)}\rangle \subseteq \bigotimes_{u=1}^k A^{(u)}$
defined by
\[
\cw_q^k
\;:=\;
\sum_{1\le u < v \le k}
\left(\sum_{i=1}^q x_i^{(u)}x_i^{(v)}\right)
x_0^{(1)}\cdots \ov{x_0^{(u)}}\cdots \ov{x_0^{(v)}}\cdots x_0^{(k)},
\]
where the overlines indicate that the corresponding factors are omitted.

The \emph{big} CW $k$-tensor is
\[
\CW_q^k
\;:=\;
\cw_q^k
\;+\;
\sum_{u=1}^k x_{q+1}^{(u)}\cdot x_0^{(1)}\cdots \ov{x_0^{(u)}}\cdots x_0^{(k)}
\;\in\;
\bigotimes_{u=1}^k A^{(u)}.
\]
\end{definition}

\paragraph{Block structure and outer support.}
We partition the basis in each mode $u\in[k]$ into three blocks
\[
b_0^{(u)}=\{x_0^{(u)}\},\qquad
b_1^{(u)}=\{x_1^{(u)},\dots,x_q^{(u)}\},\qquad
b_2^{(u)}=\{x_{q+1}^{(u)}\}.
\]
This induces a block decomposition of $\CW_q^k$ indexed by the alphabet $B=\{0,1,2\}$.
The set of block indices of nonzero blocks (the \emph{outer support}) is
\[
\Phi \;=\; \bigl\{(i_1,\dots,i_k)\in B^k : i_1+\cdots+i_k=2\bigr\}.
\]
There are two types of nonzero blocks:
\begin{itemize}
    \item (\emph{$1+1$ blocks}) If $i_u=i_v=1$ for some $u<v$ and $i_w=0$ for $w\notin\{u,v\}$, then the block is
    \[
    T_{\{u,v\},q}
    \;=\;
    \left(\sum_{i=1}^q x_i^{(u)}x_i^{(v)}\right)
    x_0^{(1)}\cdots \ov{x_0^{(u)}}\cdots \ov{x_0^{(v)}}\cdots x_0^{(k)}.
    \]
    \item (\emph{$2$ blocks}) If $i_u=2$ for some $u$ and $i_w=0$ for $w\neq u$, then the block is rank-one
    and (up to relabelling) is the graph tensor of the empty graph, denoted $T_{\emptyset,q}$.
\end{itemize}

\begin{lemma}\label{lem:brank-bigCW}
The border rank of $\CW_q^k$ is at most $q+2$.
\end{lemma}

\begin{proof}
The following degeneration is a direct $k$-mode generalization of the original Coppersmith--Winograd construction~\cite{coppersmith1990matrix}:
\begin{align*}
    &\phantom{=}
    \sum_{i=1}^q \varepsilon
    (x_0^{(1)} + \varepsilon^2 x_i^{(1)})
    \cdots
    (x_0^{(k)} + \varepsilon^2 x_i^{(k)}) \\
    &\quad -
    \left( x_0^{(1)} + \varepsilon^3 \sum_{i=1}^q x_i^{(1)} \right)
    \cdots
    \left( x_0^{(k)} + \varepsilon^3 \sum_{i=1}^q x_i^{(k)} \right) \\
    &\quad +
    (1 - q \varepsilon)
    (x_0^{(1)} + \varepsilon^5 x_{q+1}^{(1)})
    \cdots
    (x_0^{(k)} + \varepsilon^5 x_{q+1}^{(k)}) \\
    &= \varepsilon^5 \CW_q^k + O(\varepsilon^6).
\end{align*}
The left-hand side is a sum of $q+2$ rank-one tensors, hence $\underline{\mathrm{R}}(\CW_q^k)\le q+2$.
\end{proof}

\subsection{Laser Method on CW $k$-tensors}\label{subsec:laser-CWk}

We work with the block decomposition of $\CW_q^k$ induced by the alphabet $B=\{0,1,2\}$.
This decomposition lifts to tensor powers as follows.  For each mode $u\in[k]$ and each word
$w=(w_1,\dots,w_n)\in B^n$, define the block
\[
b_w^{(u)} \;:=\; b_{w_1}^{(u)}\otimes \cdots \otimes b_{w_n}^{(u)} \;\subseteq\; (A^{(u)})^{\otimes n}.
\]
A block subtensor of $(\CW_q^k)^{\otimes n}$ is indexed by a $k$-tuple of words
$(w^{(1)},\dots,w^{(k)})\in (B^n)^k$ and equals
\[
T_{w^{(1)},\ldots,w^{(k)}}
\;=\;
T_{w^{(1)}_1,\ldots,w^{(k)}_1}\ \otimes\ \cdots\ \otimes\ T_{w^{(1)}_n,\ldots,w^{(k)}_n}.
\]
Whenever this block is nonzero, each factor $T_{w^{(1)}_p,\ldots,w^{(k)}_p}$ is (up to isomorphism) a graph tensor.
Consequently $T_{w^{(1)},\ldots,w^{(k)}}$ is isomorphic to the graph tensor of the multigraph obtained by
the multiset union of the one-edge factors.

\paragraph{Type classes and principal subtensors.}
A standard step in the Laser Method is to restrict to a principal subtensor in which the relevant edge count is fixed.
Let $\alpha,\beta,\gamma\in\mathbb{Q}_{\ge 0}$ with $\alpha+\beta+\gamma=1$, and assume $n$ is chosen such that
$\alpha n,\beta n,\gamma n$ are integers.  Let
\[
W_{\alpha,\beta,\gamma}
\;:=\;
\bigl\{\,w\in B^n : \#\{p:w_p=0\}=\alpha n,\ \#\{p:w_p=1\}=\beta n,\ \#\{p:w_p=2\}=\gamma n\,\bigr\}
\]
be the corresponding type class.  We restrict $(\CW_q^k)^{\otimes n}$ to the subtensor obtained by
keeping, in each mode $u\in[k]$, only the blocks indexed by $W_{\alpha,\beta,\gamma}$.

Now consider a nonzero block $T_{w^{(1)},\dots,w^{(k)}}$ with $w^{(u)}\in W_{\alpha,\beta,\gamma}$.
At each position $p\in[n]$, the factor $T_{w^{(1)}_p,\ldots,w^{(k)}_p}$ is either a one-edge tensor or an empty-graph tensor.
Let $m$ be the number of one-edge positions.  Counting symbols across all $k$ modes yields:
\begin{itemize}
    \item each one-edge factor contributes two symbols ``$1$'' (and an empty factor contributes none), hence $2m = k\beta n$;
    \item each empty factor contributes one symbol ``$2$'' (and a one-edge factor contributes none), hence $n-m = k\gamma n$.
\end{itemize}
Therefore
\[
\frac12\,\beta+\gamma=\frac{1}{k},
\qquad\text{so}\qquad
\beta=\frac{2}{k}-2\gamma,\qquad \alpha = 1-\frac{2}{k}+\gamma.
\]
In particular, feasibility forces $0\le \gamma \le 1/k$, and the associated multigraph has exactly
\[
m \;=\; (1-k\gamma)\,n
\]
edges.  We abbreviate $W_{\alpha,\beta,\gamma}$ by $W_\gamma$, and write $(\CW_q^k)^{\otimes n}[\gamma]$
for the resulting principal subtensor.

\medskip

\paragraph{Principal subtensors with prescribed marginals.}
For later use, we record a convenient general definition.  (In our application, $B=\{0,1,2\}$
and the marginals are all equal to $(\alpha,\beta,\gamma)$.)

\begin{definition}[Principal subtensor of prescribed marginals]\label{def:principal-subtensor}
Let $T$ be a $k$-tensor written in a fixed basis $B$ in each mode.
Let $P_1,\dots,P_k$ be probability distributions on $B$ with rational entries.
For $n$ such that $nP_u(b)\in\mathbb{Z}$ for all $u\in[k]$ and $b\in B$,
let $W_{P_u}\subseteq B^n$ be the type class of words whose empirical distribution equals $P_u$.
We define the \emph{principal subtensor} $T^{\otimes n}[P_1,\dots,P_k]$ as the block subtensor of $T^{\otimes n}$
obtained by restricting, in mode $u$, to the direct sum of blocks indexed by $W_{P_u}$ and zeroing out all other blocks.
\end{definition}

\paragraph{A fixed-marginal variant of \cite[Theorem~1.52]{ChristandlVZ19}.}
We next isolate the form of the Christandl--Vrana--Zuiddam bound that we use.
It is already implicit in the proof of~\cite[Theorem~1.52]{ChristandlVZ19}:
once the marginals are fixed, the proof constructs a large diagonal tensor inside the corresponding
principal subtensor, and the only additional maximization in the theorem statement is over the choice of marginals.

\begin{theorem}[Fixed-marginal Laser-Method bound]\label{thm:fixed-marginal-laser}
Let $T$ be a tight $k$-tensor written in a fixed basis $B$, with support $\Phi\subseteq B^k$.
Fix marginal distributions $P_1,\dots,P_k$ on $B$ such that there exists a distribution on $\Phi$
with these marginals, and let $P^\star$ be a maximum-entropy distribution on $\Phi$ with marginals $P_1,\dots,P_k$.
Write $\eH(\cdot)$ for Shannon entropy in bits.

Define the diagonal set $\Delta_\Phi:=\{(x,x):x\in\Phi\}$ and, for each $i\in[k]$,
\[
R_i \;:=\; \{(x,y)\in\Phi\times\Phi : x_i=y_i\}.
\]
Let
\[
\mathscr{R}_\Phi \;:=\; \{\,R\subseteq \Phi\times\Phi : R\not\subseteq \Delta_\Phi\ \text{ and }\ R\subseteq R_i\ \text{for some }i\in[k]\,\}.
\]
For $R\in\mathscr{R}_\Phi$, let $\mathscr{Q}_{\Phi,(P_1,\dots,P_k),R}$ be the set of all probability distributions $Q$ on $R$
such that the $2k$ coordinate-marginals satisfy
\[
Q_i = Q_{k+i} = P_i \qquad (i\in[k]),
\]
where $(Q_1,\dots,Q_k)$ are the marginals on the first component $x\in\Phi$ and $(Q_{k+1},\dots,Q_{2k})$ those on the second component $y\in\Phi$.
Finally, let $\alpha_1,\dots,\alpha_k:B\to\mathbb{Z}$ be tightness maps for $T$ (unrelated to the scalar parameter $\alpha$ above), and define
$r_\alpha(R)$ to be the rank over $\mathbb{Q}$ of the matrix with rows $\alpha(x)-\alpha(y)$ for $(x,y)\in R$, where
$\alpha(x):=(\alpha_1(x_1),\dots,\alpha_k(x_k))$.

Set
\[
\mu(P_1,\dots,P_k)
\;:=\;
\eH(P^\star)
\;-\;
(k-2)\cdot
\max_{R\in\mathscr{R}_\Phi}\;
\max_{Q\in\mathscr{Q}_{\Phi,(P_1,\dots,P_k),R}}
\frac{\eH(Q)-\eH(P^\star)}{r_\alpha(R)} .
\]
Then for every sufficiently large $n$, the principal subtensor $T^{\otimes n}[P_1,\dots,P_k]$
can be zeroed out to a diagonal tensor of size
\[
2^{\,(\mu(P_1,\dots,P_k) - o(1))\, n}.
\]
\end{theorem}

\begin{remark}\label{rem:fixed-marginal-implicit}
For computations it is convenient to rewrite the expression for $\mu$ as a \emph{minimum} over $R$.
For each $R\in\mathscr{R}_\Phi$, let $Q_R^\star\in\arg\max_{Q\in\mathscr{Q}_{\Phi,(P_1,\dots,P_k),R}} \eH(Q)$ and define
\[
F_R(P_1,\dots,P_k)
\;:=\;
\eH(P^\star)
\;-\;
(k-2)\cdot\frac{\eH(Q_R^\star)-\eH(P^\star)}{r_\alpha(R)}.
\]
Then $\mu(P_1,\dots,P_k)=\min_{R\in\mathscr{R}_\Phi} F_R(P_1,\dots,P_k)$.
\end{remark}

\subsection{Analysis in 4 modes}\label{subsec:analysis-4modes}

In this subsection we specialize to $k=4$.  With respect to $B=\{0,1,2\}$, the outer support of $\CW_q^4$ is
\[
\Phi
\;=\;
\bigl\{(i_1,i_2,i_3,i_4)\in \{0,1,2\}^4 : i_1+i_2+i_3+i_4=2\bigr\}.
\]
We fix identical marginals
\[
P_1=P_2=P_3=P_4=(\alpha,\beta,\gamma),
\qquad
\alpha=\tfrac12+\gamma,\qquad \beta=\tfrac12-2\gamma,
\qquad
\gamma\in\Bigl[0,\tfrac14\Bigr].
\]
We write $\mu(\gamma)$ for the value of $\mu(P_1,\dots,P_4)$ in \Cref{thm:fixed-marginal-laser}.
Our goal is to show that
\[
\mu(\gamma)=\eH(P_1)=\eH(\alpha,\beta,\gamma),
\]
i.e.\ the rank-one candidates $R$ do not decrease the minimum in \Cref{rem:fixed-marginal-implicit} below the value attained by $R_1$.

\paragraph{Reducing the enumeration of $R$.}
Since the marginals are identical and $\Phi$ is invariant under permuting the four modes,
the quantity $F_R(P_1,\dots,P_4)$ is unchanged under simultaneous permutation of coordinates in $(x,y)\in\Phi\times\Phi$.
Therefore, when enumerating $R\in\mathscr{R}_\Phi$, we may assume
\[
R\subseteq R_1:=\{(x,y)\in \Phi\times\Phi : x_1=y_1\}.
\]
Moreover, by \cite[Lemma~3.18]{ChristandlVZ19} it suffices to consider sets $R$ that are maximal (w.r.t.\ inclusion)
among those with fixed rank $r_\alpha(R)$.  Since for $k=4$ one always has $1\le r_\alpha(R)\le k-2=2$,
this yields the following reduction:
\begin{itemize}
    \item for $r_\alpha(R)=2$, it suffices to take the unique maximal set $R=R_1$;
    \item for $r_\alpha(R)=1$, it suffices to take maximal rank-one relations inside $R_1$.
\end{itemize}

In the rank-one case, all difference vectors $\alpha(x)-\alpha(y)$ with $(x,y)\in R$ must lie on a common line.
Up to permuting modes, this line is generated by one of the two types of difference vectors
\[
(0,1,1,-2)=(0,1,1,0)-(0,0,0,2),
\qquad\text{or}\qquad
(0,1,-1,0)=(1,1,0,0)-(1,0,1,0),
\]
(the second type also occurs as $\tfrac12\bigl((0,2,0,0)-(0,0,2,0)\bigr)$).
Consequently, it suffices to consider the following three representatives, written as equivalence relations on $\Phi$
(we list only the nontrivial classes; all remaining points are singleton classes):
\[
R_{(1)}:=R_1,
\qquad
R_{(2)}:\ (0,1,1,0)\sim(0,0,0,2),
\]
and
\[
R_{(3)}:\ (0,0,1,1)\sim(0,1,0,1),\quad (1,1,0,0)\sim(1,0,1,0),\quad (0,2,0,0)\sim(0,0,2,0).
\]
In particular $r_\alpha(R_{(1)})=2$ and $r_\alpha(R_{(2)})=r_\alpha(R_{(3)})=1$.
The proofs that $F_{R_{(2)}}(\gamma)\ge \eH(P_1)$ and $F_{R_{(3)}}(\gamma)\ge \eH(P_1)$
(for all $\gamma\in[0,1/4]$) will be given in the subsequent lemmas; together with the explicit evaluation
$F_{R_{(1)}}(\gamma)=\eH(P_1)$, this implies $\mu(\gamma)=\eH(P_1)$.

\paragraph{Analysis of $R_{(1)}$.}
Recall $R_{(1)}:=R_1=\{(x,y)\in\Phi\times\Phi:\ x_1=y_1\}$ and $r_\alpha(R_{(1)})=2$.
We first show that the rank-two candidate $R_{(1)}$ yields the baseline value
$F_{R_{(1)}}(\gamma)=\eH(P_1)$.

\begin{lemma}\label{lem:FR1-equals-marginal}
Let $k=4$ and fix identical marginals $P_1=\cdots=P_4=(\alpha,\beta,\gamma)$.
Let $P^\star$ be the maximum-entropy distribution on $\Phi$ with these marginals.
Then for $R_{(1)}=R_1$ we have
\[
F_{R_{(1)}}(\gamma)=\eH(P_1)=\eH(\alpha,\beta,\gamma).
\]
\end{lemma}

\begin{proof}
Let $Q\in\mathscr{Q}_{\Phi,(P_1,\dots,P_4),R_{(1)}}$ and let $(X,Y)\sim Q$.
By definition of $R_{(1)}$ we have $X_1=Y_1$ almost surely; write $Z:=X_1=Y_1$.
Then, by the chain rule and subadditivity,
\[
\eH(X,Y)=\eH(Z)+\eH(X_{2,3,4},Y_{2,3,4}\mid Z)
\le \eH(Z)+\eH(X_{2,3,4}\mid Z)+\eH(Y_{2,3,4}\mid Z).
\]
Moreover, $\eH(Z)=\eH(P_1)$ since $Z=X_1$ has marginal $P_1$.
Using again the chain rule,
\[
\eH(X)=\eH(Z)+\eH(X_{2,3,4}\mid Z),
\qquad
\eH(Y)=\eH(Z)+\eH(Y_{2,3,4}\mid Z).
\]
Since $X$ (resp.\ $Y$) is supported on $\Phi$ and has marginals $P_1,\dots,P_4$,
maximality of $P^\star$ gives $\eH(X)\le \eH(P^\star)$ and $\eH(Y)\le \eH(P^\star)$.
Hence
\[
\eH(X_{2,3,4}\mid Z)\le \eH(P^\star)-\eH(P_1),
\qquad
\eH(Y_{2,3,4}\mid Z)\le \eH(P^\star)-\eH(P_1),
\]
and therefore
\begin{equation}\label{eq:HQ_R1_upper}
\eH(Q)=\eH(X,Y)\le 2\eH(P^\star)-\eH(P_1).
\end{equation}

This upper bound is tight: sample $Z\sim P_1$, and then sample $X$ and $Y$
\emph{independently conditioned on $Z$}, with $X\mid(Z=z)\sim P^\star(\,\cdot\mid X_1=z)$ and
$Y\mid(Z=z)\sim P^\star(\,\cdot\mid Y_1=z)$. This produces a feasible $Q$ supported on $R_{(1)}$
and with $\eH(Q)=2\eH(P^\star)-\eH(P_1)$, hence \eqref{eq:HQ_R1_upper} is optimal and
\[
\eH(Q_{R_{(1)}}^\star)=2\eH(P^\star)-\eH(P_1).
\]
Finally, since $r_\alpha(R_{(1)})=2$ and $k-2=2$,
\[
F_{R_{(1)}}(\gamma)
=
\eH(P^\star)-2\cdot \frac{\eH(Q_{R_{(1)}}^\star)-\eH(P^\star)}{2}
=
\eH(P^\star)-\bigl(\eH(Q_{R_{(1)}}^\star)-\eH(P^\star)\bigr)
=
\eH(P_1),
\]
as claimed.
\end{proof}

\paragraph{Analysis of $R_{(2)}$.}
We next treat the rank-one relation $R_{(2)}$, whose only nontrivial class is
\[
a:=(0,1,1,0)\ \sim\ b:=(0,0,0,2).
\]

\begin{lemma}\label{lem:FR2-ge-marginal}
For all $\gamma\in[0,1/4]$ we have
\[
F_{R_{(2)}}(\gamma)\ \ge\ \eH(P_1)=\eH(\alpha,\beta,\gamma).
\]
\end{lemma}

\begin{proof}
Let $Q\in\mathscr{Q}_{\Phi,(P_1,\dots,P_4),R_{(2)}}$ and $(X,Y)\sim Q$.
Since $R_{(2)}$ is diagonal except on the class $\{a,b\}$, we have:
if $X\notin\{a,b\}$ then $Y=X$, while if $X\in\{a,b\}$ then $Y\in\{a,b\}$.
Therefore $H(Y\mid X=x)\le 1$ if $x\in\{a,b\}$ and $0$ otherwise, hence
\[
\eH(Y\mid X)\ \le\ \Pr[X\in\{a,b\}].
\]
By the chain rule,
\[
\eH(Q)=\eH(X,Y)=\eH(X)+\eH(Y\mid X)\ \le\ \eH(X)+\Pr[X\in\{a,b\}].
\]
The marginal constraints imply that $X$ is supported on $\Phi$ and has coordinate marginals
$P_1,\dots,P_4$, hence $\eH(X)\le \eH(P^\star)$ by maximality of $P^\star$.
Moreover, $X=a$ forces $X_2=1$ and $X=b$ forces $X_4=2$, so
\[
\Pr[X=a]\le \Pr[X_2=1]=\beta,
\qquad
\Pr[X=b]\le \Pr[X_4=2]=\gamma,
\]
and thus $\Pr[X\in\{a,b\}]\le \beta+\gamma=\frac12-\gamma$.
Altogether,
\begin{equation}\label{eq:HQ_R2_upper}
\eH(Q_{R_{(2)}}^\star)\ \le\ \eH(P^\star)+\Bigl(\tfrac12-\gamma\Bigr).
\end{equation}

Since $r_\alpha(R_{(2)})=1$ and $k-2=2$, we have
\[
F_{R_{(2)}}(\gamma)
=
\eH(P^\star)-2\bigl(\eH(Q_{R_{(2)}}^\star)-\eH(P^\star)\bigr)
\ \ge\
\eH(P^\star)-2\Bigl(\tfrac12-\gamma\Bigr)
=
\eH(P^\star)-1+2\gamma.
\]
It remains to show $\eH(P^\star)-1+2\gamma\ge \eH(P_1)$, i.e.
\begin{equation}\label{eq:HPstar_minus_Hmarg_needed}
\eH(P^\star)-\eH(P_1)\ \ge\ 1-2\gamma.
\end{equation}

Let $X^\star\sim P^\star$. Then $\eH(P^\star)-\eH(P_1)=\eH(X^\star_{2,3,4}\mid X^\star_1)$.
If $X^\star_1=1$, then exactly one of $X^\star_2,X^\star_3,X^\star_4$ equals $1$, hence
$\eH(X^\star_{2,3,4}\mid X^\star_1=1)=\lg 3$.
If $X^\star_1=0$, then $(X^\star_2,X^\star_3,X^\star_4)$ lies either in the three ``$2$-type'' outcomes
or in the three ``$1+1$-type'' outcomes; these two sets are disjoint and each is uniform under $P^\star$
by symmetry, so $\eH(X^\star_{2,3,4}\mid X^\star_1=0)\ge \lg 3$.
Finally, if $X^\star_1=2$ the remaining coordinates are deterministic, so the conditional entropy is $0$.
Therefore
\[
\eH(P^\star)-\eH(P_1)
=
\eH(X^\star_{2,3,4}\mid X^\star_1)
\ \ge\
\Pr[X^\star_1\neq 2]\cdot \lg 3
=
(1-\gamma)\lg 3.
\]
Since $\lg 3>1$ and $2-\lg 3>0$, we have
\[
(1-\gamma)\lg 3-(1-2\gamma)=(\lg 3-1)+\gamma(2-\lg 3)\ge 0,
\]
which implies \eqref{eq:HPstar_minus_Hmarg_needed} and completes the proof.
\end{proof}

\paragraph{Analysis of $R_{(3)}$.}
We finally treat the remaining rank-one representative, where (listing only the nontrivial classes)
\[
R_{(3)}:\ (0,0,1,1)\sim(0,1,0,1),\quad (1,1,0,0)\sim(1,0,1,0),\quad (0,2,0,0)\sim(0,0,2,0).
\]

\begin{lemma}\label{lem:FR3-ge-marginal}
For all $\gamma\in[0,1/4]$ we have
\[
F_{R_{(3)}}(\gamma)\ \ge\ \eH(P_1)=\eH(\alpha,\beta,\gamma).
\]
\end{lemma}

\begin{proof}
Let $Q\in\mathscr{Q}_{\Phi,(P_1,\dots,P_4),R_{(3)}}$ and $(X,Y)\sim Q$.
Let
\[
S:=\{(0,0,1,1),(0,1,0,1),(1,1,0,0),(1,0,1,0),(0,2,0,0),(0,0,2,0)\}\subseteq \Phi
\]
be the set of points lying in the nontrivial $R_{(3)}$-classes.
If $X\notin S$ then $Y=X$ deterministically, while if $X\in S$ then $Y$ has at most two possible values.
Hence $\eH(Y\mid X)\le \Pr[X\in S]$ and
\begin{equation}\label{eq:R3_basic_upper}
\eH(Q)=\eH(X,Y)=\eH(X)+\eH(Y\mid X)\ \le\ \eH(X)+\Pr[X\in S].
\end{equation}

We now write the law of $X$ explicitly. Since each coordinate takes the value $2$ with probability $\gamma$,
and in $\Phi$ the event $\{X_u=2\}$ uniquely identifies the point $2e_u$, we have
\[
\Pr[X=(2,0,0,0)]=\Pr[X=(0,2,0,0)]=\Pr[X=(0,0,2,0)]=\Pr[X=(0,0,0,2)]=\gamma.
\]
For the remaining six points (of type $(1,1,0,0)$), define
\[
x_{12}:=\Pr[X=(1,1,0,0)],\ x_{13}:=\Pr[X=(1,0,1,0)],\ x_{14}:=\Pr[X=(1,0,0,1)],
\]
\[
x_{23}:=\Pr[X=(0,1,1,0)],\ x_{24}:=\Pr[X=(0,1,0,1)],\ x_{34}:=\Pr[X=(0,0,1,1)].
\]
The marginal constraint $\Pr[X_u=1]=\beta$ for $u=1,2,3,4$ yields
\begin{align*}
x_{12}+x_{13}+x_{14}&=\beta,\\
x_{12}+x_{23}+x_{24}&=\beta,\\
x_{13}+x_{23}+x_{34}&=\beta,\\
x_{14}+x_{24}+x_{34}&=\beta,
\end{align*}
from which it follows that $x_{12}=x_{34}$, $x_{13}=x_{24}$, and $x_{14}=x_{23}$.
Thus there exist $x,y,z\ge 0$ with $x+y+z=\beta$ such that
\[
x_{12}=x_{34}=x,\qquad x_{13}=x_{24}=y,\qquad x_{14}=x_{23}=z.
\]
In these variables,
\[
\eH(X)= -4\gamma\lg\gamma \;-\;2x\lg x\;-\;2y\lg y\;-\;2z\lg z,
\]
and, since $S$ contains two of the ``$2$-type'' points and four of the ``$1+1$-type'' points,
\[
\Pr[X\in S]=2\gamma+2(x+y).
\]
Substituting into \eqref{eq:R3_basic_upper} gives the bound
\begin{equation}\label{eq:R3_relaxed_upper_xy}
\eH(Q_{R_{(3)}}^\star)\ \le\
-4\gamma\lg\gamma \;-\;2x\lg x\;-\;2y\lg y\;-\;2z\lg z
\;+\;2\gamma+2(x+y),
\end{equation}
where $x,y,z\ge 0$ and $x+y+z=\beta$.

The right-hand side of \eqref{eq:R3_relaxed_upper_xy} is concave in $(x,y,z)$ over the simplex
$x+y+z=\beta$. A short Lagrange-multiplier computation shows that it is maximized at
\[
x=y=\frac{2\beta}{5},
\qquad
z=\frac{\beta}{5},
\]
which yields the explicit upper bound
\begin{equation}\label{eq:R3_Ugamma_def}
U(\gamma)
:=
-4\gamma\lg\gamma
\;-\;\frac{8\beta}{5}\lg\!\Bigl(\frac{2\beta}{5}\Bigr)
\;-\;\frac{2\beta}{5}\lg\!\Bigl(\frac{\beta}{5}\Bigr)
\;+\;2\gamma+\frac{8\beta}{5},
\qquad \beta=\tfrac12-2\gamma,
\end{equation}
so that $\eH(Q_{R_{(3)}}^\star)\le U(\gamma)$.

Next note that all difference vectors $\alpha(x)-\alpha(y)$ for $(x,y)\in R_{(3)}$ are multiples of
$(0,1,-1,0)$, hence $r_\alpha(R_{(3)})=1$. Therefore
\[
F_{R_{(3)}}(\gamma)
=
\eH(P^\star)-2\bigl(\eH(Q_{R_{(3)}}^\star)-\eH(P^\star)\bigr)
=
3\eH(P^\star)-2\eH(Q_{R_{(3)}}^\star)
\ \ge\
3\eH(P^\star)-2U(\gamma).
\]
It remains to show that $3\eH(P^\star)-2U(\gamma)\ge \eH(P_1)$ for $\gamma\in[0,1/4]$.

Using $\beta=\tfrac12-2\gamma$ and the explicit form of $P^\star$ on $\Phi$, we have
\[
\eH(P^\star)
=
-4\gamma\lg\gamma \;-\;2\beta\lg\!\Bigl(\frac{\beta}{3}\Bigr).
\]
Define
\[
D(\gamma):=\bigl(3\eH(P^\star)-2U(\gamma)\bigr)-\eH(\alpha,\beta,\gamma),
\qquad \alpha=\tfrac12+\gamma,\ \beta=\tfrac12-2\gamma.
\]
A direct differentiation yields, for $\gamma\in(0,1/4)$,
\[
D''(\gamma)=\frac{3}{\gamma(2\gamma+1)(4\gamma-1)\ln 2}\ <\ 0,
\]
so $D$ is concave on $[0,1/4]$ and hence attains its minimum at an endpoint.
At $\gamma=0$ one has $\eH(P^\star)=\lg 6$, $U(0)=\lg 10$, and $\eH(\alpha,\beta,0)=1$, giving
\[
D(0)=3\lg 6-2\lg 10-1=\lg\!\Bigl(\frac{27}{25}\Bigr)>0.
\]
At $\gamma=\tfrac14$ one has $\beta=0$, $\eH(P^\star)=2$, $U(\tfrac14)=\tfrac52$, and
$\eH(\tfrac34,0,\tfrac14)= -\tfrac34\lg\tfrac34-\tfrac14\lg\tfrac14$, giving
\[
D\!\Bigl(\tfrac14\Bigr)=1-\eH\!\Bigl(\tfrac34,0,\tfrac14\Bigr)
=-1+\frac34\lg 3\ >0.
\]
Thus $D(\gamma)\ge 0$ on $[0,1/4]$, which implies $F_{R_{(3)}}(\gamma)\ge \eH(P_1)$.
\end{proof}

\paragraph{Coda.}
Combining \Cref{lem:FR1-equals-marginal,lem:FR2-ge-marginal,lem:FR3-ge-marginal} with
$\mu(\gamma)=\min_{R\in\mathscr{R}_\Phi}F_R(\gamma)$ (cf.\ \Cref{rem:fixed-marginal-implicit}) yields:

\begin{corollary}\label{cor:mu-equals-marginal}
For all $\gamma\in[0,1/4]$,
\[
\mu(\gamma)=\eH(\alpha,\beta,\gamma),
\qquad
\alpha=\tfrac12+\gamma,\qquad \beta=\tfrac12-2\gamma.
\]
\end{corollary}

We now explain how the fixed-marginal laser bound (Theorem~\ref{thm:fixed-marginal-laser})
specializes to an explicit upper bound on $\tau(K_4)$.
Recall that for $k=4$ we write
\[
P_1=P_2=P_3=P_4=(\alpha,\beta,\gamma),
\qquad
\alpha=\tfrac12+\gamma,
\qquad
\beta=\tfrac12-2\gamma,
\qquad
0\le \gamma \le \tfrac14,
\]
and that the associated principal subtensor $(\CW_q^4)^{\otimes n}[\gamma]$ corresponds to
a multigraph on $K_4$ with
\[
m=(1-4\gamma)n
\]
edge-factors (and $n-m=4\gamma n$ empty factors).

\begin{proposition}\label{prop:tauK4-bound-gamma}
Let $q\ge 2$ and let $0\le \gamma < \tfrac14$ be rational.
Then
\begin{equation}\label{eq:tauK4-gamma-bound}
\tau(K_4)
\;\le\;
\frac{1}{1-4\gamma}\,
\log_q\!\left(\frac{q+2}{2^{\mu(\gamma)}}\right),
\qquad
\mu(\gamma)=H\!\left(\tfrac12+\gamma,\ \tfrac12-2\gamma,\ \gamma\right).
\end{equation}
In particular, taking $\gamma=0$ recovers the Christandl--Vrana--Zuiddam bound
\[ \tau(K_4)\le \min_{q\ge 2}\log_q\!\left(\frac{q+2}{2}\right)=\log_7(\tfrac92). \]
\end{proposition}

\begin{proof}
Fix $q\ge 2$ and $\gamma\in[0,\tfrac14)$ rational, and take $n$ such that
$\alpha n,\beta n,\gamma n$ are integers.
Equip $\CW_q^4$ with the $3$-block partition described in
Definition~\ref{def:CWk} (blocks indexed by $B=\{0,1,2\}$), so that the outer support is
\[
\Phi=\{(i_1,i_2,i_3,i_4)\in\{0,1,2\}^4 : i_1+i_2+i_3+i_4=2\}.
\]
Let $\varphi$ denote the corresponding outer-structure tensor on alphabet $B$.

By \Cref{cor:mu-equals-marginal},
the principal outer subtensor $\varphi^{\otimes n}[P_1,P_2,P_3,P_4]$ can be zeroed out to a
diagonal $4$-tensor of size
\[
p \;=\; 2^{(\mu(\gamma)-o(1))n},
\qquad
\mu(\gamma)=H(P_1)=H(\alpha,\beta,\gamma).
\]
It lifts to a zeroing-out on the full tensor, yielding a direct sum
\[
\phi_1 \oplus \cdots \oplus \phi_p \ \leq\ (\CW_q^4)^{\otimes n}
\]
where each $\phi_i$ is an inner-structure tensor obtained as a tensor product of $n$ block
tensors of $\CW_q^4$.
For our fixed principal subtensor $[\gamma]$, each $\phi_i$ is isomorphic to a graph tensor $T_{G_i,q}$ with
$m=(1-4\gamma)n$ edges in the (multi)graph $G_i$.

On the other hand, \Cref{lem:brank-bigCW} gives $\BR(\CW_q^4)\le q+2$,
and therefore
\[
\BR\left(\phi_1\oplus\cdots\oplus\phi_p\right)
\;\le\;
\BR\left((\CW_q^4)^{\otimes n}\right)
\;\le\;
(q+2)^n.
\]
Applying the generalized asymptotic sum inequality
\cite[Theorem 2.1.6]{ChristandlVZ19} to $K_4$ yields
\[
\tau(K_4)
\;\le\;
\log_{q^{m}}\!\left(\frac{(q+2)^{n}}{p}\right)
\;=\;
\frac{n}{m}\,\log_q\!\left(\frac{q+2}{2^{\mu(\gamma)-o(1)}}\right).
\]
Since $m=(1-4\gamma)n$, letting $n\to\infty$ gives~\eqref{eq:tauK4-gamma-bound}.
\end{proof}

\begin{proof}[Proof of Theorem~\ref{thm:tauK4-improved}]
By Proposition~\ref{prop:tauK4-bound-gamma}, for any $q\ge 2$ and any rational
$\gamma\in[0,\tfrac14)$ we have
\[
\tau(K_4)
\;\le\;
\frac{1}{1-4\gamma}\,
\log_q\!\left(\frac{q+2}{2^{\eH(\frac12+\gamma,\ \frac12-2\gamma,\ \gamma)}}\right).
\]
Substituting $q=7$ and $\gamma = 0.0012105179$
into~\eqref{eq:tauK4-gamma-bound} gives
\[
\tau(K_4)
\;\approx\;
0.77231702\ <\ 0.772318.
\]
This proves $\tau(K_4)<0.772318$.
\end{proof}

\end{document}